\RequirePackage{nameref}
\PassOptionsToPackage{table,dvipsnames}{xcolor}
\PassOptionsToPackage{hypertexnames=false}{hyperref}
\documentclass[a4paper, onecolumn, unpublished]{quantumarticle}
\pdfoutput=1
\usepackage[utf8]{inputenc}
\usepackage{maths}
\usepackage[ruled,vlined]{algorithm2e}
\usepackage{mathformatting}
\usepackage{formatting}
\usepackage{tikz}
\usepackage{comment}
\usepackage{amsmath, amssymb, amsthm, amsfonts, graphicx}
\usepackage{subcaption}
\theoremstyle{plain}
\usepackage{url}
\usepackage{enumerate}
\usepackage{enumitem}
\usepackage{quantikz}
\usepackage{multirow}
\usepackage{longtable}
\makeatletter
\let\LT@err\@gobbletwo
\makeatother
\usepackage{array}
\usepackage{booktabs}
\usepackage{rotating}
\usepackage{adjustbox}
\usepackage{placeins}
\usepackage[numbers,sort&compress]{natbib}
\usetikzlibrary{fit}
\usetikzlibrary{positioning, arrows.meta, shapes.geometric, backgrounds, fit}

\newcommand{\be}{\begin{equation}}
\newcommand{\ee}{\end{equation}}
\newcommand{\bea}{\begin{eqnarray}}
\newcommand{\eea}{\end{eqnarray}}

\newcommand{\tAut}{\colorbox{green!12}{\small\textsf{Aut}}}
\newcommand{\tH}{\colorbox{blue!12}{\small\textsf{H}}}
\newcommand{\tS}{\colorbox{orange!16}{\small\textsf{S}}}
\newcommand{\tCX}{\colorbox{red!12}{\small\textsf{CX}}}
\tikzset{lqstyle/.style={row sep={0.55cm,between origins}, column sep=0.30cm}}
\newcommand{\circwrap}[1]{\raisebox{-0.5\height}{#1}}

\newif\ifnotes
\notestrue

\author[1]{AJ Davenport}
\author[2]{John Blue}
\author[1,2]{Isaac Chuang}
\affil[1]{Electrical Engineering and Computer Science Department, Massachusetts Institute of Technology}
\affil[2]{Department of Physics, Massachusetts Institute of Technology}

\title{Generalized Bicycle Codes as Cyclic Submodules and their Automorphism Structure}

\begin{document}
\maketitle

\begin{abstract}
Automorphisms of quantum codes, when they exist, offer a pathway toward fault-tolerant gate implementation via qubit relabeling. Although useful, the conditions under which automorphisms appear in a given code remain poorly understood. In this paper, we develop an algebraic framework for systematically analyzing and engineering automorphisms in Generalized Bicycle (GB) codes. Central to our approach is the derivation of a three-space dependency between the polynomial ring space, the parity check matrix space, and the $\bbF_2^{2\ell}$ qubit space, similar to the structure found in the study of classical cyclic codes. By expressing GB codes as a pair of cyclic submodules of $R_\ell^2$, where $R_\ell \cong \bbF_2[x]/\langle x^\ell-1\rangle$, we reduce the search for code automorphisms to a deterministic algebraic problem, deriving necessary and sufficient conditions for the existence of block-separable automorphisms built from cyclic shifts, ring automorphisms and block-swaps. We connect these conditions to the fold-transversal gate framework, providing explicit criteria for the existence of $H$-, $S$-, and $CX$-type fold-transversal gates. We further discuss structured bases for logical operators in order to determine the logical action of a given automorphism. Finally, we introduce the Maximal Cube Root (MCR) code family, a family of GB codes constructed around the principle of maximizing automorphism flexibility and fold-CX gates. We demonstrate a collection of $k=2$ MCR codes up to $d=13$ generating the 2-qubit Clifford group via automorphism and fold-transversal gates, with stabilizer weight ranging from 8 to 16, and $k>2$ MCR codes with a minimum of 20 distinct logical gates achievable from automorphisms. This serves as a first demonstration of inverse design: using these methods to build codes around a rich automorphism structure from the ground up.
\end{abstract}

\tableofcontents

\section{Introduction}

Fault-tolerant quantum computation requires the ability to perform logical operations on encoded qubits while actively suppressing the accumulation of physical errors. At present, it is largely accepted that the primary path toward fault tolerance relies on quantum error correction, which involves encoding a small number of logical qubits into many physical qubits and adding ancilla qubits to facilitate the execution of memory and computation. Understanding and reducing the overhead required to achieve this fault-tolerantly, from qubit count to hardware challenges in implementation, remains one of the central engineering challenges standing between current hardware and practically useful quantum computers.

Historically, the surface code and small-$k$ sister codes such as the toric code have been the canonical codes of choice due to their low stabilizer weight, geometric locality, and ease of implementation in hardware. However, their $[[n, k, d]]$ parameters scale poorly, and in recent years there has been a push to investigate quantum Low-Density Parity-Check (qLDPC) codes. While these codes simultaneously bound stabilizer weight and qubit check degree, they typically require higher-weight stabilizers and long-range (non-local) connectivity compared to the surface code in order to achieve higher encoding rates.

The search for codes that simultaneously achieve good parameters across all of these axes has driven substantial recent progress, with constructions such as hypergraph product codes \cite{og_HGP, eczoo_hgp}, lifted and balanced product codes \cite{PK_asmyptoticallygood, qianzhao_constantoverhead, breuckmann2021balancedproduct}, bivariate-bicycle codes \cite{Original_BBpaper}, and others \cite{Zhao_highrate, kasai2026}. However, achieving good parameters and efficient syndrome extraction is only part of the quantum computing stack. Encoded qubits are useful only if one can perform a universal set of logical gates on them, ideally with minimal overhead. 

Particularly desirable are logical gates implemented either by permutation or constant-depth unitary circuits, as these circuits propagate errors in a controlled way that respects the code's error correction guarantees with minimal overhead. Such constant-depth logical gate implementations are possible with CSS codes and include transversal, automorphism~\cite{grassl_autos}, and fold-transversal gates \cite{Breuckmann_foldtransversal}. While these gates are often insufficient for universal computation, 
recent work on extractor architectures \cite{he2025extractors,yoder2025tourgrossmodularquantum} has provided an avenue toward universal quantum computation with arbitrary qLDPC codes using qLDPC surgery~\cite{cohen2022low-overhead,cross2024improved,williamson2024low-overhead,Ide_2025,swaroop2024universal,parsimonioussurgery} and available constant-depth unitaries. Notably, since surgery is almost always more costly than constant-depth gates, computation on qLDPC codes remains an active line of research.

The dominant paradigm in recent qLDPC work has been to identify codes with good distance, rate, and stabilizer weight, and then to ask whether the resulting code happens to admit such desirable gates. This approach has yielded important constructions for determining, given a code, when automorphisms and fold-transversal gates exist \cite{Breuckmann_foldtransversal, sayginel2025faulttolerantlogicalcliffordgates}, but it leaves much to be desired, such as a systematic principle connecting the algebraic choices made during code design to the gate set that ultimately emerges. This leaves code design, to some extent, at the mercy of good fortune, or at least, computationally expensive numerical searches.

Recent innovative studies such as Refs.~\cite{phantomcodes, ExhaustiveAtomorphisms, automorphisms_BBcodes, berthusen2025automorphismgadgetshomologicalproduct, berthusen2025simplelogicalquantumcomputation} begin to address the inverse question: what is it about a code's structure that yields good automorphisms? Of particular relevance to this work, \cite{automorphisms_BBcodes} established a framework for identifying bases of logical operators and fold-transversal gates in two-block group algebra (2GBA) codes, and used this framework to identify BB codes with rich symmetries. Motivated by these works, we move to assess Generalized Bicycle (GB) codes using the language of ring theory and ask: for a given GB code, precisely \textit{when} and \textit{how} do automorphisms arise, and what are the algebraic conditions on the code's defining polynomials that control them? By answering this question for any GB code, we aim to lay the groundwork for a more principled design methodology: one in which a code can be engineered from the start to realize a targeted logical gate set, without sacrificing the distance and rate properties that make it useful.

Central to this work is a framing of GB codes through what we refer to as a three-space dependency, a perspective that bridges the following three spaces that define a GB code: (a) the parity-check matrix space, built from binary circulant matrices; (b) the qubit space, represented by elements of $\bbF_2^{2\ell}$; and (c) the algebra underlying the code, the polynomial ring space $(\bbF_2[x]/\langle x^\ell-1\rangle )^2 = R_\ell^2$. Under the three-space dependency, a natural map emerges by first lifting from the check matrix space to the polynomial ring space, where automorphisms in the polynomial ring space can then be mapped to the qubit space. The three-space dependency framework, depicted in Figure \ref{fig:3spaceiso}, is a central tool for bridging logical gates, code automorphisms, and code definitions, and allows us to determine guaranteed automorphisms of a GB code given only the block length $\ell$ and the code's defining polynomials.

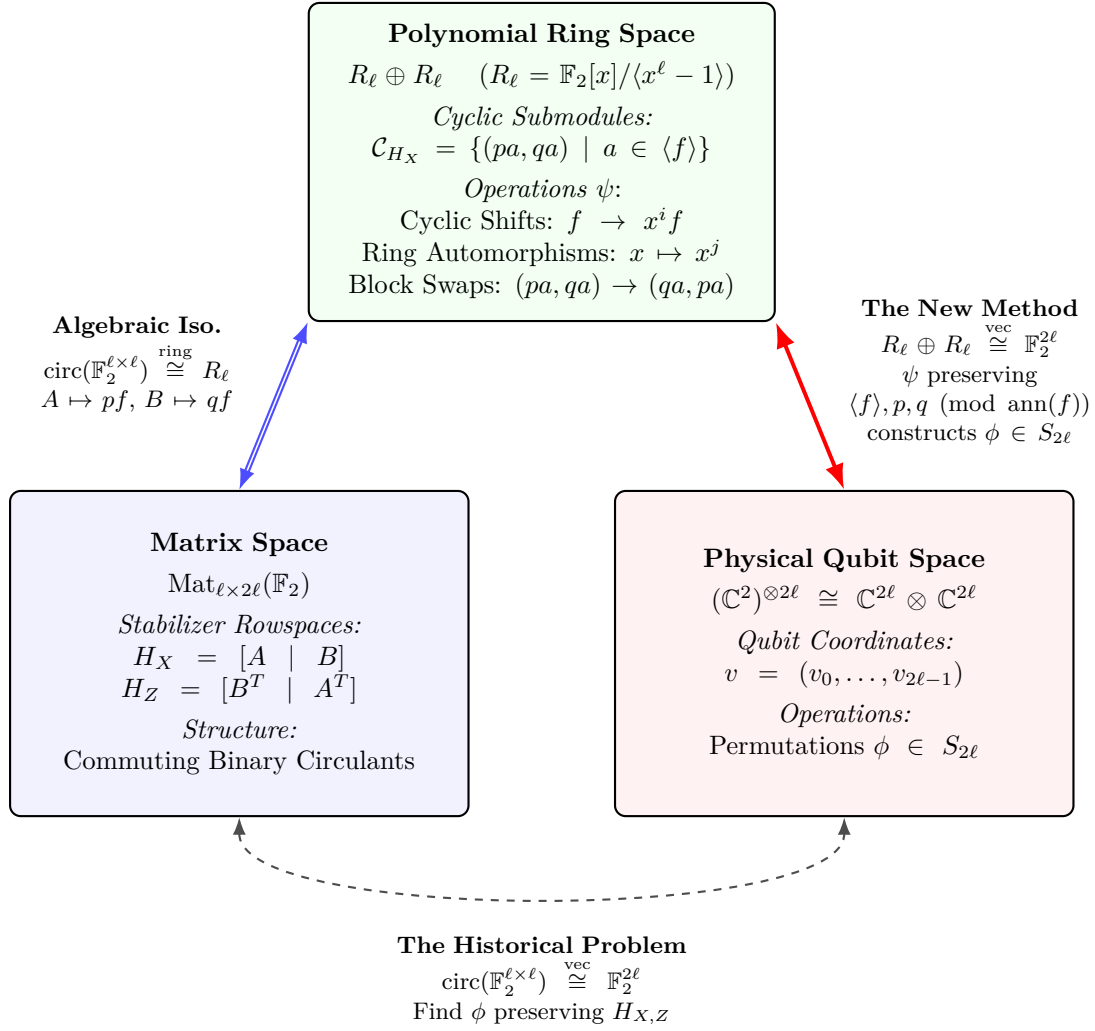
\begin{figure}[t]
\centering
\begin{tikzpicture}[
    >=Latex,
    spacebox/.style={
        rectangle, draw=black, thick, rounded corners,
        align=center, inner sep=8pt, minimum height=2.5cm
    },
    action/.style={
        ->, thick, draw=black!70
    },
    iso/.style={
        <->, double, thick, draw=blue!70
    },
    labeltext/.style={
        align=center, font=\small
    }
]


\node[spacebox, fill=green!5, text width=5.6cm] (Poly) at (0, 6.5) {
    \textbf{Polynomial Ring Space} \\[0.8ex]
    $R_\ell \oplus R_\ell$ \quad ($R_\ell = \mathbb{F}_2[x]/\langle x^\ell-1 \rangle$) \\[0.8ex]
    \textit{Cyclic Submodules:}\\
    $\mathcal{C}_{H_X} = \{(pa, qa) \mid a \in \langle f \rangle\}$ \\[0.8ex]
    \textit{Operations }$\psi$:\\
    Cyclic Shifts: $f \rightarrow x^i f$\\
    Ring Automorphisms: $ x \mapsto x^j$\\
    Block Swaps: $(pa, qa) \rightarrow (qa, pa)$
};

\node[spacebox, fill=blue!5, text width=5.5cm, minimum height=4.3cm] (Matrix) at (-4, 0) {
    \textbf{Matrix Space} \\[0.8ex]
    $\text{Mat}_{\ell \times 2\ell}(\mathbb{F}_2)$ \\[0.8ex]
    \textit{Stabilizer Rowspaces:}\\
    $H_X = [A \mid B]$\\
    $H_Z = [B^T \mid A^T]$ \\[0.8ex]
    \textit{Structure:}\\
    Commuting \mbox{Binary} Circulants
};

\node[spacebox, fill=red!5, text width=5.5cm, minimum height=4.3cm] (Physical) at (4, 0) {
    \textbf{Physical Qubit Space} \\[0.8ex]
    $(\mathbb{C}^2)^{\otimes 2\ell}\cong \mathbb{C}^{2\ell} \otimes \mathbb{C}^{2\ell}$\\[0.8ex]
    \textit{Qubit Coordinates:}\\
    $v = (v_0, \dots, v_{2\ell-1})$ \\[0.8ex]
    \textit{Operations:}\\
    Permutations $\phi \in S_{2\ell}$
};


\draw[iso] (Matrix.north) --
    node[labeltext, pos=0.75, left=0.4cm, text width=3cm] {
        \textbf{Algebraic Iso.}\\[0.3ex]
        $\mathrm{circ}(\mathbb{F}_2^{\ell \times \ell}) \overset{\text{ring}}\cong R_\ell$\\
        $A \mapsto pf$, $B \mapsto qf$
    }
    (Poly.south west);

\draw[action, <->, draw=red, line width=1.5pt] (Poly.south east) --
    node[labeltext, pos=0.3, right=0.4cm, text width=3.5cm] {
        \textbf{The New Method}\\[0.3ex]
        $R_\ell \oplus R_\ell \overset{\text{vec}}{\cong} \bbF_2^{2\ell}$\\
        $\psi$ preserving \\ 
        $\langle f \rangle, p, q \pmod{\ann(f)}$\\
        constructs $\phi \in S_{2\ell}$
    }
    (Physical.north);

\draw[action, <->, dashed] (Matrix.south)
    to[out=-90, in=-90, looseness=0.45]
    node[labeltext, midway, below=0.4cm, text width=4.5cm] {
        \textbf{The Historical Problem}\\[0.3ex]
        $\mathrm{circ}(\bbF_2^{\ell \times \ell}) \overset{\text{vec}}{\cong} \bbF_2^{2\ell}$\\
        Find $\phi$ preserving $H_{X,Z}$\\
    }
    (Physical.south);

\end{tikzpicture}
\caption{\textbf{The three-space framework} for Generalized Bicycle codes.  The historical approach (dashed line) attempts to find physical qubit permutations that preserve the global matrix rowspaces, an opaque and computationally heavy search outside of the trivial cyclic shift solutions (recently improved by ~\cite{sayginel2025faulttolerantlogicalcliffordgates}, however does not provide insight into why automorphisms arise). The fundamental insight of this work (double blue line and red arrow) is to map the matrices into the Polynomial Ring Space, where physical permutations are recontextualized as formal ring automorphisms ($x \mapsto x^j$), and can be further paired with cyclic shifts and block swaps. This allows deterministic automorphism code design from the top down by choosing polynomials $p$, $q$, and $f$ that satisfy specific constraints with respect to $\ell$.  The algebraic ring isomorphism between the Polynomial Ring and Matrix Spaces (blue arrow) and the cyclic submodule description of the rowspaces (top box) are established in Sec.~\ref{sec:3SpaceIso}; the lift from polynomial-ring operations $\psi$ to physical qubit permutations $\phi \in S_{2\ell}$ (red arrow), together with the algebraic conditions on $(f, p, q)$ that determine which gates exist, is developed in Sec.~\ref{sec:AutomorphismStructure}; the resulting logical action on the encoded qubits is the subject of Sec.~\ref{sec:LogicalOps}.  See main text for notation definitions.}
\label{fig:3spaceiso}
\end{figure}

We start in Sec.~\ref{sec:3SpaceIso} by introducing the concept of the three-space dependency based on well-understood classical cyclic coding theory. This allows us to present a generalized ansatz for GB codes as cyclic submodules of $(\bbF_2[x]/\langle x^\ell - 1 \rangle)^2 = R_\ell^{2}$, enabling the extension of the three-space dependency framework to GB codes. In Sec.~\ref{sec:AutomorphismStructure} we demonstrate that this ansatz can be leveraged, along with knowledge of the dimension $\ell$, the base ideal $f$, and what we call transfer polynomials $p, q$, to determine the full class of block-separable automorphisms built from cyclic shifts, block-swapping, and $R_\ell$ ring automorphisms. We further connect this work to the methods outlined in Ref.~\cite{Breuckmann_foldtransversal} for $H$- and $S$-type fold-transversal gates, and we derive sufficient conditions that ensure the existence of a CX-type fold-transversal gate. In Sec.~\ref{sec:LogicalOps} we discuss canonical logical operators and assignments for these codes in order to assess the logical action of a given automorphism. Finally, in Sec.~\ref{sec:MaxCubeRoot}, we present a new family of generalized bicycle codes (the ``Maximal Cube Root code family'') characterized by maximal automorphism flexibility and a large number of fold-CX gates\footnote{Scripts for finding MCR codes and determining the automorphisms and logical actions of a GB code given any valid $f, p, q$, and $\ell$ can be found at \url{https://github.com/ajdav136/GBAutomorphisms}}. For example, we demonstrate a collection of $k=2$ MCR codes ranging from $[[18, 2, 5]]$ to $[[66, 2, 13]]$ whose automorphisms and fold gates generate the 2-qubit Clifford group with stabilizer weight between 8 and 16. For larger values of $k$, we present a $[[102, 18, \leq 12]]$ code featuring 76 unique logical gates achievable via automorphism and fold-transversal operations.

We emphasize that this work does not complete the program of designing GB codes around targeted gate sets while simultaneously maintaining relative locality and low stabilizer weight. The inverse design strategy presented here represents a first step: this shows that the algebraic levers are well-understood and can be pulled deliberately to yield codes with many automorphisms, and when $k$ is small, the full Clifford group. However, while being designed to have many automorphisms, the resulting codes of Sec.~\ref{sec:MaxCubeRoot} do not target a pre-specified logical gate set, and at larger $k$ the automorphism-induced gates fall short of the Clifford group. Further, the codes have not yet been fully optimized for distance or low stabilizer weight. We view the reconciliation of automorphism-driven design with the numerical search techniques that have produced the best-known GB code parameters as the natural next direction for this line of work, and we hope the framework developed here provides useful foundations for that endeavor.

\section{Generalized Bicycle Codes as Cyclic Submodules of $R_\ell^2$}
\label{sec:3SpaceIso}

We begin in Sec.~\ref{sec:threespaces} by reviewing the three interconnected algebraic spaces at the heart of classical cyclic coding theory before reviewing the Generalized Bicycle code construction and motivating the central theorem used in this work, contained in Sec.~\ref{sec:gbreview}. Our central thesis is that the rich algebraic structure underpinning cyclic coding theory can be naturally extended to Generalized Bicycle codes through cyclic submodules, yielding a similar three-space dependency framework under which the symmetries and automorphisms of GB codes can be assessed (Sec.~\ref{sec:3spaceGB}). We call this a three-space dependency, and it serves as the organizing principle for the results that follow.

\subsection{The Three Spaces}
\label{sec:threespaces}

We start with a well-known chain of isomorphisms from classical cyclic coding theory:
\be
\bbF_q[x]/\langle x^\ell - 1\rangle \underset{\text{ring}}{\cong}  \mathrm{circ}(\bbF_q^{\ell\times \ell}) \underset{\text{vec}}{\cong}  \bbF_q^\ell
\ee
where $q$ is a prime power representing the size of the alphabet, $\bbF_q$ is the Galois field containing $q$ elements, $\bbF_q[x]$ is the ring of polynomials in the formal variable $x$ with coefficients chosen from $\bbF_q$, $\ell$ is a positive integer representing the block length of the cyclic code, $\bbF_q[x]/\langle x^\ell - 1\rangle$ is the quotient ring of polynomials modulo $x^\ell -1$, and 
$\mathrm{circ}(\bbF_q^{\ell\times \ell})$ denotes the space of circulant matrices of dimension $\ell \times \ell$ with coefficients over $\bbF_q$. We note that in this work, we adopt the convention to define circulant matrices via row shifts, as is standard in classical coding theory, as opposed to column shifts, as is standard in prior works on GB codes \cite{Panteleev_2021, wang2022distanceboundsgeneralizedbicycle}. It should be noted that circulant matrices are simultaneously circulant with respect to either row or column shifts via the $x \rightarrow x^{-1}$ involution. Our decision to canonicalize the row shift perspective is well motivated, and will become clear in Sec.~\ref{sec:gbreview}.

A specific example of these isomorphisms is demonstrated in Fig.~\ref{fig:3spaceexample}. It is important to note that the isomorphism between the polynomial ring and the space of circulant matrices is an isomorphism of rings, while the second is an isomorphism of vector spaces\footnote{Indeed, there is an isomorphism of vector spaces between $\bbF_q[x]/\langle x^\ell - 1\rangle$ and $\bbF_q^\ell$ as well}. In particular, as the isomorphism to $\bbF_q^\ell$ is an isomorphism of vector spaces, $\bbF_q^\ell$ is blind to the ring structure of $\bbF_q[x]/\langle x^\ell-1 \rangle $. Though the space of circulant matrices inherits this ring structure, much of the algebraic structure of the code is obfuscated by large arrays of field elements.

This chain of isomorphisms plays a central role in classical coding theory, the cornerstone of which is the following fact: All cyclic codes are isomorphic to ideals of $R_{q, \ell} = \bbF_q[x]/\langle x^\ell - 1\rangle$. As such, the study of cyclic codes and their properties reduces to the study of ideals of polynomial rings, an incredibly rich algebraic field of study.  Thus, the polynomial ring space is absolutely crucial to understanding the algebraic properties of a classical cyclic code, and the vector space isomorphism back to $\bbF_q^\ell$ provides a map for understanding how the ring structure governs and acts on the codespace, elements of $\bbF_q^\ell$. 

As we will show, the algebraic properties of $R_{q,\ell}$ and $R_{q,\ell}^2 = R_{q, \ell} \oplus R_{q,\ell}$ play just as important a role in the study of GB codes as they do in the study of classical cyclic codes. 

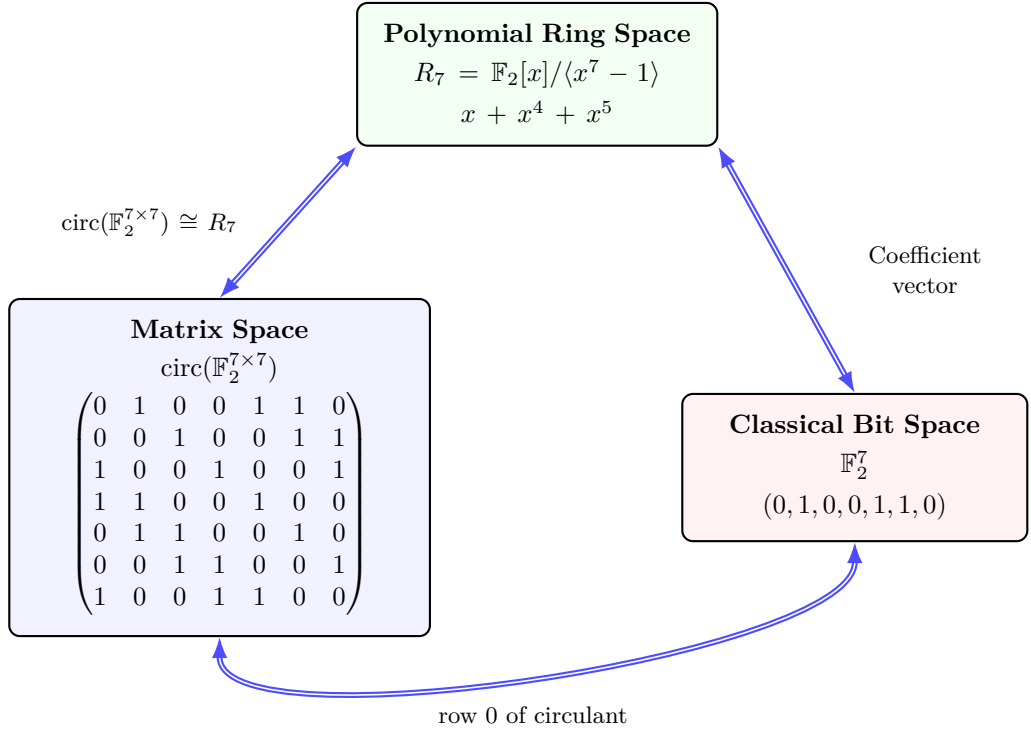
\begin{figure}[htb]
\centering
\begin{tikzpicture}[
    >=Latex,
    spacebox/.style={
        rectangle, draw=black, thick, rounded corners,
        align=center, inner sep=8pt
    },
    iso/.style={
        <->, double, thick, draw=blue!70
    },
    labeltext/.style={
        align=center, font=\small
    }
]


\node[spacebox, fill=green!5, text width=4.2cm] (Poly) at (0, 5.2) {
    \textbf{Polynomial Ring Space} \\[0.6ex]
    $R_7 = \mathbb{F}_2[x]/\langle x^7-1\rangle$ \\[0.8ex]
    $x + x^4 + x^5$
};

\node[spacebox, fill=blue!5, text width=5.0cm] (Matrix) at (-4.2, 0) {
    \textbf{Matrix Space} \\[0.5ex]
    $\mathrm{circ}(\mathbb{F}_2^{7\times 7})$ \\[0.5ex]
    $\begin{pmatrix}
    0 & 1 & 0 & 0 & 1 & 1 & 0 \\
    0 & 0 & 1 & 0 & 0 & 1 & 1 \\
    1 & 0 & 0 & 1 & 0 & 0 & 1 \\
    1 & 1 & 0 & 0 & 1 & 0 & 0 \\
    0 & 1 & 1 & 0 & 0 & 1 & 0 \\
    0 & 0 & 1 & 1 & 0 & 0 & 1 \\
    1 & 0 & 0 & 1 & 1 & 0 & 0
    \end{pmatrix}$
};

\node[spacebox, fill=red!5, text width=4.0cm] (Bit) at (4.2, 0) {
    \textbf{Classical Bit Space} \\[0.6ex]
    $\mathbb{F}_2^7$ \\[0.8ex]
    $(0,1,0,0,1,1,0)$
};


\draw[iso] (Matrix.north) --
    node[labeltext, pos=0.5, left=0.3cm, text width=2.8cm] {
        $\mathrm{circ}(\mathbb{F}_2^{7\times 7}) \cong R_7$
    }
    (Poly.south west);

\draw[iso] (Poly.south east) --
    node[labeltext, pos=0.5, right=0.3cm, text width=2.8cm] {
        Coefficient\\vector
    }
    (Bit.north);

\draw[iso] (Matrix.south)
    to[out=-90, in=-90, looseness=0.45]
    node[labeltext, midway, below=0.3cm] {
        row 0 of circulant
    }
    (Bit.south);

\end{tikzpicture}
\caption{The classical three-space dependency at $\ell = 7$: the element $x + x^4 + x^5 \in R_7$ and its equivalent representations as a binary circulant matrix in $\mathrm{circ}(\mathbb{F}_2^{7\times 7})$ and a vector in $\mathbb{F}_2^7$. Each row of the circulant is the coefficient vector of $x^i \cdot a(x) \bmod (x^7+1)$. }
\label{fig:3spaceexample}
\end{figure}

We recommend that readers who are not familiar or comfortable with ring theory review the content of Appendix~\ref{appen:codingtheory}. At many points throughout this work we use facts about polynomial rings from a standard course in abstract algebra; however, we do try to note when a conclusion may not be obvious if one is not familiar with ring theory. Many abridged lemmas are contained in Appendix~\ref{appen:codingtheory} as well.

Going forward and throughout this work we assume $\bbF_q = \bbF_2$ and note that we expect the central results of Sec.~\ref{sec:3SpaceIso} and \ref{sec:AutomorphismStructure} to extend trivially when $q > 2$; however, we leave an explicit analysis for future work. We assign $R_\ell = \bbF_2[x]/\langle x^\ell - 1\rangle$. 

\subsection{Generalized Bicycle Codes as Cyclic Submodules}
\label{sec:gbreview}

This subsection has two purposes. First, we recall the standard Generalized Bicycle code construction \cite{kovalev2013quantumkronecker, Panteleev_2021}, which builds a quantum CSS code from a pair of cyclic codes defined by polynomials $f_1, f_2 \in R_\ell$.  Second, we motivate the algebraic decomposition that drives the rest of this paper: every interesting pair $(f_1, f_2)$ admits a factorization $f_1 = pf$, $f_2 = qf$, where $f = \gcd(f_1, f_2, x^\ell - 1)$ and we call $p, q$ \emph{transfer polynomials}. This separates the part of $(f_1, f_2)$ that determines how many logical qubits are encoded ($f$) from the part that determines how the two halves of the code relate to each other ($p$ and $q$).  It also rewrites every row of the parity-check matrix $H_X$ in the suggestive form $(pa, qa)$ with $a \in \langle f \rangle$, foreshadowing the cyclic submodule description of the rowspace that we formalize in Theorem~\ref{thm:cyclic_submodule}.

We begin with the standard construction: a Generalized Bicycle code is a quantum code constructed from two classical cyclic codes defined by polynomials $f_1, f_2 \in R_\ell$ for some fixed block length $\ell$. The quantum code then has $H_X, H_Z$ check matrices built from the binary circulant representations of $f_1, f_2$, mapping $f_1$ to $A$, and $f_2$ to $B$:
\bea
H_X = [A | B] \quad \quad H_Z = [B^T |A^T]
\eea
Note that by the isomorphism above for classical codes, circulant matrices necessarily commute as $R_\ell$ is a commutative ring, and as such $H_XH_Z^T = 0$ is guaranteed \footnote{Choosing to define the binary circulant matrix via row shifts as opposed to column shifts does not separate or distinguish these results on GB codes from prior work, as codes defined by row shifts of $f_1, f_2$ as opposed to column shifts are the same up to swapping the roles of $H_X, H_Z$. See Appendix~\ref{subsec:cycliccodes}.}

A standard fact about GB codes is that the number of encoded qubits $k$ is $2\deg(\gcd(f_1, f_2, x^\ell-1))$ \cite{kovalev2013quantumkronecker, Panteleev_2021}. In particular, when $f_1$ and $f_2$ share no common divisors of $x^\ell-1$, $\gcd(f_1, f_2, x^\ell-1) = 1$ and $k=0$. Thus, any interesting GB code must be built from $f_1, f_2$ that share some irreducible factor of $x^\ell-1$. As $\bbF_2[x]$ is a unique factorization domain, let
\bea
x^\ell-1 = \prod_i g_{i}^{e_i}
\eea
denote the factorization of $x^\ell-1$ over $\bbF_2[x]$, where each $g_i$ is irreducible. As we previously established, $f_1, f_2$ each on their own define a cyclic code, and as such must each generate some ideal of $R_\ell$. All ideals of $R_\ell$ are principal (generated by a single element), and each ideal of $R_\ell$ uniquely corresponds to a divisor of $x^\ell-1$, i.e., products of irreducible factors with the multiplicity of any $g_i$ at most $e_i$ \cite{huffmanpless}. Denote the ideal generated by a polynomial $p \in R_\ell$ as $\langle p \rangle$. We will routinely make use of the following fact (Lemma~\ref{lem:equivalent_ideals}): given $r \in R_\ell$, $g$ a divisor of $x^\ell-1$ 
\[
\langle r \rangle = \langle g \rangle \iff \gcd(r, x^\ell-1) = g \iff r = ug
\]
for $u$ in $R_\ell^\times$. Elements of $R_\ell^\times$ are referred to either as units, or invertible elements.

Thus, even if $f_1, f_2$ are not identically equivalent to some product of irreducible factors of $x^\ell-1$, we have that the ideal generated by $f_1$, i.e., $\langle f_1 \rangle$, is equivalent to $\langle g_{f_1}\rangle$, where $g_{f_1}$ is some product of divisors of $x^\ell-1$. This too holds for $f_2$, and we have that 
\bea
f_1 = u_{f_1}\cdot g_{f_1} \quad \quad f_2 = u_{f_2} \cdot g_{f_2}
\eea
where $u_{f_1}, u_{f_2}$ are elements in $R_\ell^\times$.

While $g_{f_1}$ and $g_{f_2}$ need not be equal, the assumption $k > 0$ forces them to share at least one irreducible factor of $x^\ell - 1$.  Let $f = \gcd(f_1, f_2, x^\ell - 1)$ denote this shared part — the product of irreducible factors of $x^\ell - 1$ that divide both $f_1$ and $f_2$.  Removing these common factors from each half, write $f_1 = pf$ and $f_2 = qf$, where $p, q \in R_\ell$, with\footnote{Be careful however to note that any pair of $p, q, x^\ell-1$ need not have gcd 1.} $\gcd(p, q, x^\ell-1) = 1$.

To see how $p$ and $q$ act on the ideals $\langle f_1 \rangle, \langle f_2 \rangle$, take any $\alpha \in \langle f_1 \rangle$ and write $\alpha = rf_1 = rpf = p(rf)$ by commutativity, with $r \in R_\ell$ and $rf \in \langle f \rangle$ by the definition of an ideal.  Thus every element of $\langle f_1 \rangle$ is the product of $p$ and an element of $\langle f \rangle$, and the analogous statement holds\footnote{Note however this does \emph{not} imply that $\langle f_1 \rangle = \langle f \rangle$.} for $\langle f_2 \rangle$.  By Lemma~\ref{lem:subsetideals},
\bea
\langle f_1 \rangle = \{ pa | a \in \langle f \rangle\} \quad \quad \langle f_2 \rangle = \{ qa | a \in \langle f \rangle \}
\eea
This description associates to each row of $H_X = [A | B]$ a pair $(\alpha, \beta)$ with $\alpha \in \langle f_1 \rangle$ and $\beta \in \langle f_2 \rangle$, linked by a common multiplier $a \in \langle f \rangle$ scaled by $p$ on the left and $q$ on the right.  Concretely, the first row of $H_X$ is the $\bbF_2^{2\ell}$ element $(pf, qf)$, the pair of classical cyclic code polynomials; the second row is $(xpf, xqf) = (p(xf), q(xf)) = (pf', qf')$ with $f' = xf \in \langle f \rangle$, a cyclic shift of each half by one position; and inductively every row of $H_X$ has the form $(pa, qa)$ for some cyclic shift $a = x^if\in \langle f \rangle$. One may similarly associate polynomial pairs to $H_Z$. This is made more concrete in Theorem~\ref{thm:cyclic_submodule}.

Given an element $g \in R_\ell$ and a pair $(\alpha, \beta) \in R_\ell^2$, the cyclic submodule of $R_\ell^2$ generated by $(\alpha, \beta)$ over $\langle g \rangle$ is the subset
\bea
    M_g(\alpha, \beta) \;=\; \{(\alpha a,\, \beta a) : a \in \langle g \rangle\} \;\subseteq\; R_\ell^2.
\eea
In words, one fixes a generator $(\alpha, \beta)$ and scales both coordinates simultaneously by elements $a$ drawn from the ideal $\langle g \rangle \subseteq R_\ell$. The set $M_g(\alpha, \beta)$ is closed under addition and under scaling both coordinates by any element of $R_\ell$, and is generated, with respect to this scaling, by the single pair $(\alpha, \beta)$, hence the term ``cyclic.''  A more general treatment of module theory can be found in, for example, \cite{dummit2003abstract}.

With these ideas in mind, we formalize the above discussion into a central theorem that frames GB codes as cyclic submodules, allowing us to invoke the vast literature on $R_\ell$ and cyclic coding theory to algebraically assess the automorphism and logical operator structure of a GB code. 

\begin{theorem}
\label{thm:cyclic_submodule}
Let $C$ be a Generalized Bicycle code of block length $\ell$ defined by polynomials $f_1, f_2 \in R_\ell = \mathbb{F}_2[x]/(x^\ell - 1)$, with parity-check matrices
\bea
H_X = [\mathrm{circ}(f_1) \mid \mathrm{circ}(f_2)], \qquad H_Z = [\mathrm{circ}(f_2)^T \mid \mathrm{circ}(f_1)^T].
\eea
and $k = 2\deg(\gcd(f_1, f_2, x^\ell - 1)) > 0$. Let $\overleftarrow{\alpha} = \alpha(x^{-1})$, that is, the polynomial $\alpha$ evaluated at $x^{-1}$.
Then there exist polynomials $p$ and $q \in R_\ell$ such that the rowspaces of $H_X, H_Z$ can be represented as a pair of cyclic submodules of $R_\ell^2$:
\begin{equation}
\begin{aligned}\label{eq:rowspaces}
    \operatorname{rs}(H_X) &= \{(pa,\, qa) \mid a \in \langle f \rangle\} = M_f(p,q)\\
    \operatorname{rs}(H_Z) &= \{(\overleftarrow{qa},\, \overleftarrow{pa}) \mid a \in \langle f \rangle\} = M_{\overleftarrow{f}}(\overleftarrow{q}, \overleftarrow{p})
\end{aligned}
\end{equation}
where
\bea
f = \gcd(f_1, f_2, x^\ell - 1), \qquad  \overleftarrow{qa} = q(x^{-1})a(x^{-1})
\,,
\eea
$f_1 = pf$, $f_2 = qf$ and $\gcd(p, q, x^\ell - 1) = 1$.  Furthermore, define
\bea
\hat{f} = \frac{x^\ell - 1}{f}, \qquad S = R_\ell / \langle \hat{f} \rangle
\eea
The transfer polynomials $p$ and $q$ are unique modulo the annihilator $\langle \hat{f} \rangle$: if $f_1 = p'f$ for some $p' \in R_\ell$, then $p' \equiv p \pmod{\hat{f}}$, and the rowspaces in \eqref{eq:rowspaces} depend on $p, q$ only through their images in $S$.
\end{theorem}

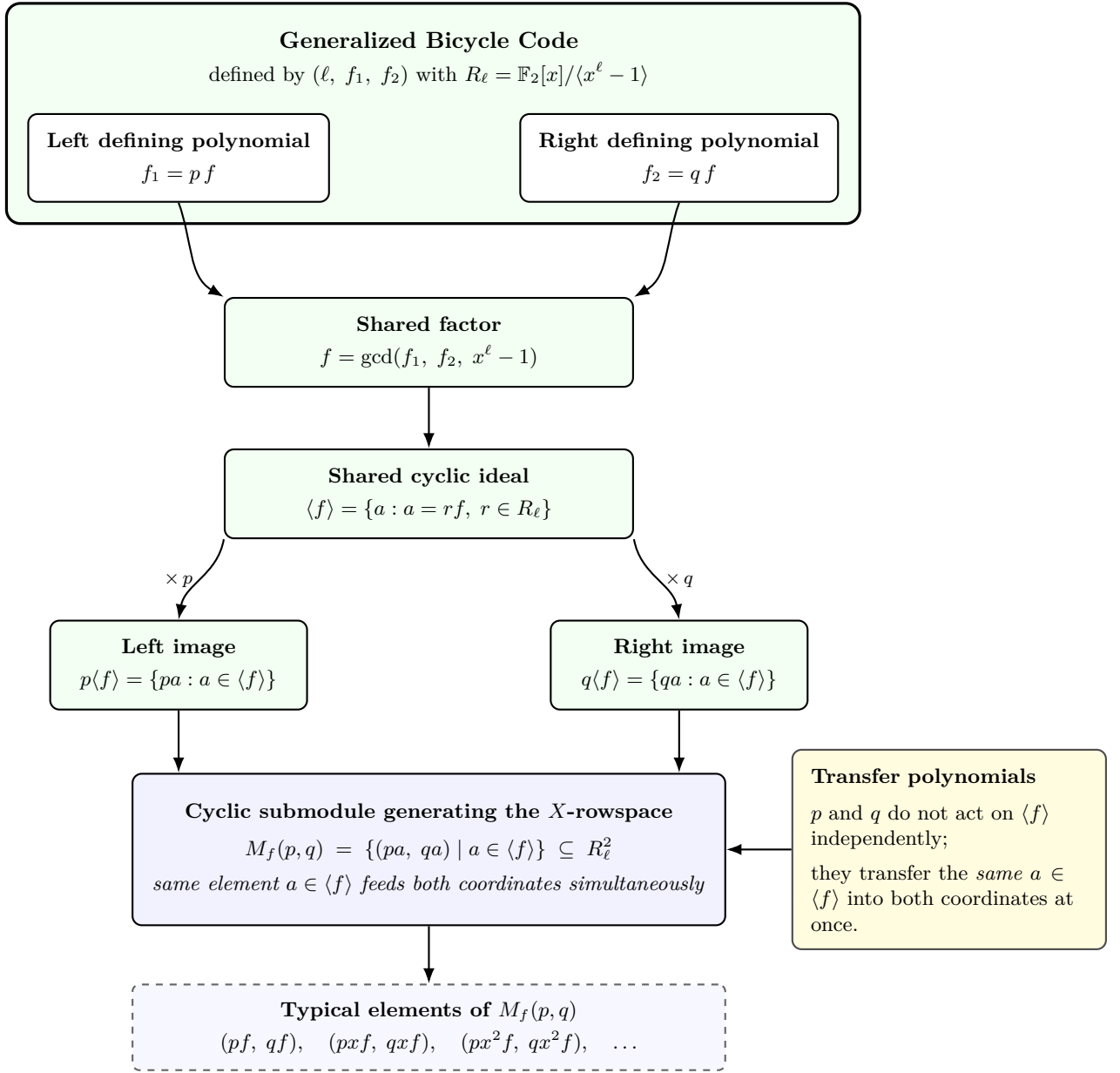
\begin{figure}[htbp]
\centering
\begin{tikzpicture}[
    font=\small,
    >=Latex,
    sbox/.style={rectangle, draw=black, thick, rounded corners,
                 align=center, inner sep=8pt},
    polybox/.style={sbox, fill=green!5},
    matbox/.style={sbox, fill=blue!5},
    subbox/.style={sbox, fill=white},
    arr/.style={->, thick},
    every node/.style={align=center}
]

\node[align=center] (GBlabel) at (0, 0.3) {
    {\normalsize\bfseries Generalized Bicycle Code}\\[3pt]
    defined by $(\ell,\; f_1,\; f_2)$ with
    $R_\ell = \F_2[x]/\langle x^\ell - 1\rangle$
};

\node[subbox, minimum width=3.8cm] (f1) at (-3.8, -1.2) {
    \textbf{Left defining polynomial}\\[3pt]
    $f_1 = p\,f$
};
\node[subbox, minimum width=3.8cm] (f2) at (3.8, -1.2) {
    \textbf{Right defining polynomial}\\[3pt]
    $f_2 = q\,f$
};

\begin{pgfonlayer}{background}
    \node[draw=black, very thick, rounded corners=6pt, fill=green!5,
          fit=(GBlabel)(f1)(f2), inner sep=9pt] (GBbox) {};
\end{pgfonlayer}

\node[polybox, minimum width=6.2cm] (f) at (0, -4.0) {
    \textbf{Shared factor}\\[3pt]
    $f = \gcd(f_1,\; f_2,\; x^\ell - 1)$
};
\draw[arr] (f1.south) to[out=-70, in=150] (f.north west);
\draw[arr] (f2.south) to[out=-110, in=30]  (f.north east);

\node[polybox, minimum width=6.2cm] (ideal) at (0, -6.3) {
    \textbf{Shared cyclic ideal}\\[3pt]
    $\langle f\rangle = \{a : a = rf,\; r\in R_\ell\}$
};
\draw[arr] (f) -- (ideal);

\node[polybox, minimum width=3.9cm] (pI) at (-3.8, -8.9) {
    \textbf{Left image}\\[3pt]
    $p\langle f\rangle = \{pa : a\in\langle f\rangle\}$
};
\node[polybox, minimum width=3.9cm] (qI) at (3.8, -8.9) {
    \textbf{Right image}\\[3pt]
    $q\langle f\rangle = \{qa : a\in\langle f\rangle\}$
};
\draw[arr] (ideal.south west) to[out=-100, in=80]
    node[left, font=\footnotesize, pos=0.5] {$\times\,p$} (pI.north);
\draw[arr] (ideal.south east) to[out=-80, in=100]
    node[right, font=\footnotesize, pos=0.5] {$\times\,q$} (qI.north);

\node[matbox, minimum width=9cm, minimum height=2.3cm] (M) at (0, -11.7) {
    \textbf{Cyclic submodule generating the $X$-rowspace}\\[5pt]
    $M_f(p,q) \;=\; \{(pa,\; qa)\mid a\in \langle f\rangle\}\;\subseteq\; R_\ell^2$\\[4pt]
    \textit{same element $a\in\langle f\rangle$ feeds both coordinates simultaneously}
};
\draw[arr] (pI.south) -- (M.north -| pI);
\draw[arr] (qI.south) -- (M.north -| qI);

\node[draw=black!60, dashed, thick, rounded corners=3pt, inner sep=7pt,
      fill=blue!3, minimum width=9cm] (ex) at (0, -14.4) {
    \textbf{Typical elements of $M_f(p,q)$}\\[3pt]
    $(pf,\;qf),\quad (pxf,\;qxf),\quad (px^2f,\;qx^2f),\quad\ldots$
};
\draw[arr] (M) -- (ex);

\node[draw=black!70, thick, rounded corners=4pt, fill=yellow!15,
      text width=4.2cm, inner sep=8pt, align=left] (note) at (7.9, -11.7) {
    \textbf{Transfer polynomials}\\[5pt]
    $p$ and $q$ do not act on 
    $\langle f \rangle$ independently;\\[4pt]
    they transfer the \emph{same}
    $a\in\langle f\rangle$ into
    both coordinates at once.
};
\draw[arr] (note.west) -- (M.east);

\end{tikzpicture}
\caption{%
    \textbf{Algebraic decomposition of a Generalized Bicycle (GB) code into
    its cyclic submodule structure.}
    A GB code is defined by parameters $(\ell, f_1, f_2)$ with $f_1, f_2\in R_\ell$
    (green, Polynomial Ring Space).
    Writing $f = \gcd(f_1,f_2,x^\ell-1)$ extracts the shared factor and yields
    transfer polynomials $p,q$ satisfying $f_1 = pf$, $f_2 = qf$.
    The shared ideal $\langle f\rangle$ is then split by $p$ and $q$ into
    left and right images, which combine into the cyclic submodule
    $M_f(p,q) = \{(pa,qa)\mid a\in\langle f\rangle\}\subseteq R_\ell^2$
    (blue, Matrix Space) that generates the $X$-rowspace of the code.
    A second cyclic submodule $M_{\overleftarrow{f}}(\overleftarrow{q}, \overleftarrow{p})$ for the $Z$-rowspace exists but is not shown here.
    Crucially, the \emph{same} element $a\in\langle f\rangle$ scales both
    coordinates simultaneously, making $M_f(p,q)$ a genuine cyclic submodule
    rather than a Cartesian product of independent ideals.}
\label{fig:cyclic_submodule}
\end{figure}

\begin{proof}
Existence and rowspace decomposition follow from the preceding discussion in Sec.~\ref{sec:gbreview}. We address uniqueness of $p, q$ in $S$.

Observe that in $R_\ell, f_1 = pf$ does not determine $p$ uniquely. Let $p' = p + r\hat{f}$. We then have that $p'f = (pf + r\hat{f}f) = pf \in R_\ell$ as $\hat{f}$ annihilates $f$. However, observe that working over $S$, $p$ defines the rowspaces uniquely.

Suppose $f_1 = p'f$ in $R_\ell$ for some $p' \in R_\ell$. Then $(p - p')f = 0$ in $R_\ell$, which means $(x^\ell - 1) \mid (p - p')f$ in $\mathbb{F}_2[x]$. Since $x^\ell - 1 = f\hat{f}$ and $\mathbb{F}_2[x]$ is an integral domain, we may cancel $f$ to obtain $\hat{f} \mid (p - p')$, i.e., $p \equiv p' \pmod{\hat{f}}$.

Conversely, suppose $p' = p + r\hat{f}$ for some $r \in R_\ell$. For any $a \in \langle f \rangle$, write $a = fb$; then
\begin{equation*}
    p'a \;=\; pa + r\hat{f} \cdot fb \;=\; pa + r(x^\ell - 1)b \;=\; pa
\end{equation*}
in $R_\ell$. Hence the generator $(pa, qa)$ is unchanged, and the rowspaces in \eqref{eq:rowspaces} depend on $p$ only through its image in $S$. The same argument applies to $q$.
\end{proof}
This theorem is illustrated in Fig.~\ref{fig:cyclic_submodule}. Note that $\psi_{x^{-1}}: x \rightarrow x^{-1}$ is always a ring automorphism of $R_\ell$ irrespective of $\ell$ as $\gcd(\ell, \ell-1)$ is always 1 (See Appendix~\ref{appen:codingtheory} for more on ring automorphisms of $R_\ell$). As such, it should be made clear that $q(x^{-1})a(x^{-1}) = (qa)(x^{-1})$, that is, the product of $qa$, evaluated at $x^{-1}$. As we will need at various points across this work to denote actual inverse polynomials, i.e., $q^{-1}(x)$, we denote the evaluation of any polynomial (or product of polynomials) at $x^{-1}$ with a backwards arrow\footnote{Be careful not to assume that $\overleftarrow{a}$ denotes the complete vector reversal of $a$. e.g., let $a = (0,1,1,0,0) = x  + x^2$. Evaluating a polynomial at $x^{-1}$ leaves the 1's coefficient unchanged, and $\overleftarrow{a} = a(x^{-1}) = x^3 + x^4 = (0,0,0,1,1) \neq (0,0,1,1,0)$.}, as in $\overleftarrow{qa}$. We choose to use $\leftarrow$ as opposed to $*$ as used in previous GB works \cite{panteleev2022quantumldpc} to denote evaluation at $x^{-1}$ to remove ambiguity about whether $qa^*$ refers to $q(x) \cdot a(x^{-1})$ or $q(x^{-1}) \cdot a(x^{-1})$ --- the notation $\overleftarrow{qa}$ makes it clear the ring automorphism is applied to the product of both elements. 

Note that $\psi_{x^{-1}}$, as a ring automorphism, is a map that preserves all algebraic structure of $R_\ell$, a fact we will use over and over in Sec.~\ref{sec:AutomorphismStructure}. Further observe that $\overleftarrow{\alpha} = \alpha(x^{-1})$ converts transpose operations in matrix space into ring operations in polynomial space; specifically, it encodes the reciprocal/reversal involution $x\mapsto x^{-1}$ on the cyclic polynomial ring, which corresponds to transposition of circulant matrices. Thus, $M_f(p,q), M_{\overleftarrow{f}}(\overleftarrow{q}, \overleftarrow{p})$ are related by a ring automorphism (and $p/q$ swapping), a fact that greatly simplifies the analysis of Sec.~\ref{sec:AutomorphismStructure}.

\begin{remark}
    We note that GB codes can be identified with index-2 quasi-cyclic codes \cite{kovalev2013quantumkronecker, Panteleev_2021}, and the study of index-2 quasi-cyclic codes in classical coding theory has previously identified such codes with module theory and as generated by a single element \cite{QCLally, Abdukhalikov_2026, dastbasteh2023polynomialrepresentationadditivecyclic}. While a few of these works have included quantum coding theory as a motivation for their work, none of these works have drawn an equivalence between cyclic submodules and Generalized Bicycle codes, or explored any of the consequences thereof. To the best of the authors' knowledge, this work marks the first time Generalized Bicycle codes have been classified as cyclic submodules or explored the resulting automorphism and logical operator structure inherited from this perspective.
\end{remark}

\subsection{Three-Space Dependency Framework for GB Codes}
\label{sec:3spaceGB}

Theorem~\ref{thm:cyclic_submodule} derives a correspondence between $H_X, H_Z$ and $M_f(p,q), M_{\overleftarrow{f}}(\overleftarrow{q}, \overleftarrow{p})$, which immediately defines a three-space perspective similar to that of cyclic codes. We may associate to every GB code of block length $\ell$ the following three spaces (illustrated by an example shown in Fig.~\ref{fig:3spacequantumexample}):
\begin{itemize}
\item \textbf{The Polynomial Ring Space.}
  The algebraic space $R_\ell^2$, where the GB code is realized as a pair of cyclic submodules.
  \begin{itemize}
  \item {\bf Operations:} Symmetries arise as permutations acting on each half, lifting to permutation automorphisms and permutation equivalences of $\langle pf \rangle, \langle qf \rangle$ from classical coding theory, such as cyclic shifts and substitution maps $x \mapsto x^j$ (where $\gcd(j, \ell) = 1$). These operations can preserve, swap, invert, and reverse the generators $pf$ and $qf$ yielding non-trivial module automorphisms, and are explored more in Sec.~\ref{subsec:bsep_auto}.
  \end{itemize}
\item \textbf{The Matrix Space.}  The space of $1 \times 2$ block matrices built from binary circulant matrices. Specifically, the stabilizer parity-check matrices ($H_X$ and $H_Z$) are formed by horizontally concatenating the circulant matrix for $pf$ with the circulant matrix for $qf$. The CSS condition ($H_X H_Z^T = 0$) is guaranteed by the commutativity of these circulant matrices.
  \begin{itemize}
  \item {\bf Operations:} Symmetries in this space manifest as physical column permutations of the matrices.
  \end{itemize}
\item \textbf{The Physical Qubit Space.}
  The vector space $\mathbb{F}_2^{2\ell} \cong \mathbb{F}_2^\ell \oplus \mathbb{F}_2^\ell$. The coordinates of this space label the $2\ell$ physical qubits, which are partitioned into a Left and Right block, each of length $\ell$.
  \begin{itemize}
  \item {\bf Operations:} Symmetries here consist of direct physical coordinate permutations, denoted by $\phi \in S_{2\ell}$.
    \end{itemize}

\end{itemize}

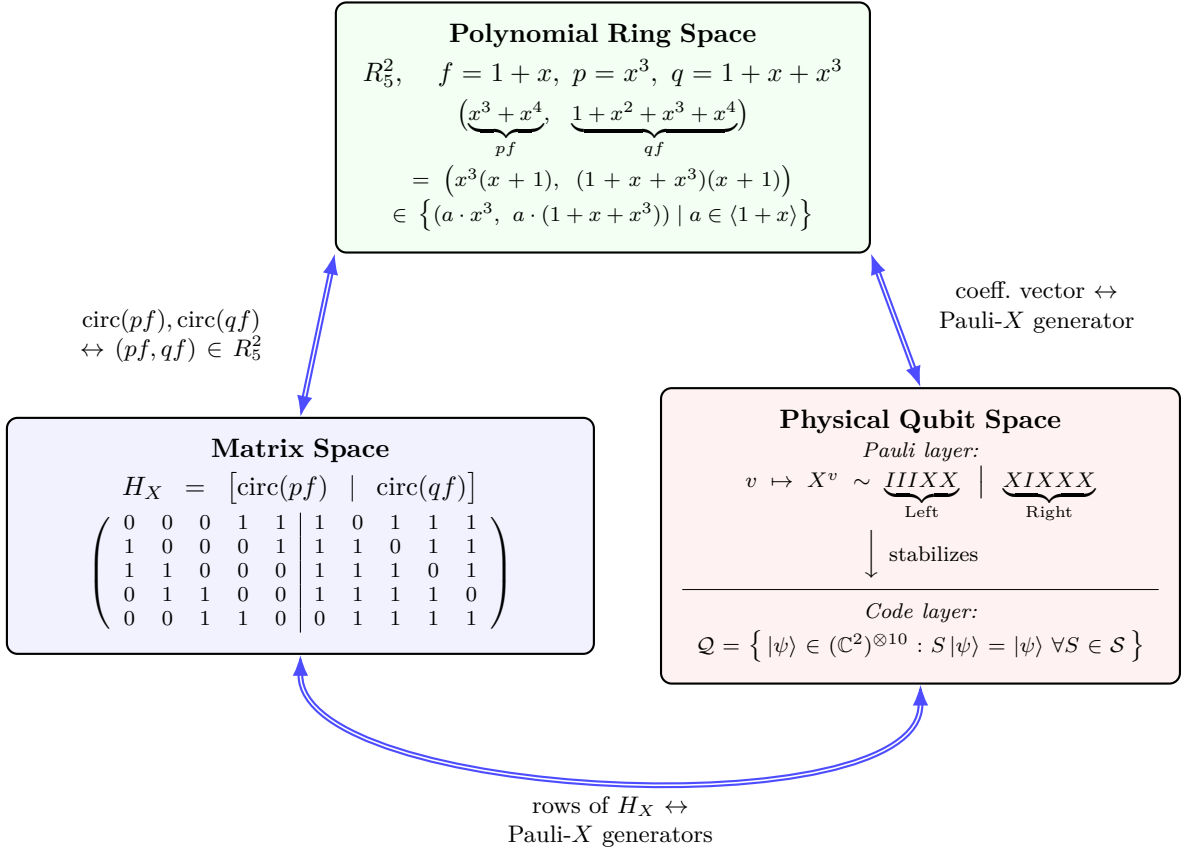
\begin{figure}[htbp]
\centering
\begin{tikzpicture}[
    >=Latex,
    spacebox/.style={
        rectangle, draw=black, thick, rounded corners,
        align=center, inner sep=8pt
    },
    iso/.style={
        <->, double, thick, draw=blue!70
    },
    action/.style={
        ->, thick, draw=black!70
    },
    labeltext/.style={
        align=center, font=\small
    }
]
\node[spacebox, fill=green!5, text width=6.5cm] (Poly) at (0, 5.4) {
    \textbf{Polynomial Ring Space} \\[0.6ex]
    $R_5^2$, \quad $f = 1+x,\ p = x^3,\ q = 1+x+x^3$ \\[0.8ex]
    \footnotesize
    $\bigl(\underbrace{x^3+x^4}_{pf},\;
           \underbrace{1+x^2+x^3+x^4}_{qf}\bigr)$\\[0.4ex]
    $= \bigl(x^3(x+1),\;(1+x+x^3)(x+1)\bigr)$ \\
    $\in \left\{(a \cdot x^3, \ a \cdot (1 + x + x^3) )\mid a \in \langle 1 + x \rangle\right\}$
};
\node[spacebox, fill=blue!5, text width=7.2cm] (Matrix) at (-4.0, 0) {
    \textbf{Matrix Space} \\[0.5ex]
    $H_X = \bigl[\mathrm{circ}(pf)\mid\mathrm{circ}(qf)\bigr]$ \\[0.5ex]
    \footnotesize
    \renewcommand{\arraystretch}{0.95}%
    $\left(\begin{array}{ccccc|ccccc}
    0 & 0 & 0 & 1 & 1 & 1 & 0 & 1 & 1 & 1 \\
    1 & 0 & 0 & 0 & 1 & 1 & 1 & 0 & 1 & 1 \\
    1 & 1 & 0 & 0 & 0 & 1 & 1 & 1 & 0 & 1 \\
    0 & 1 & 1 & 0 & 0 & 1 & 1 & 1 & 1 & 0 \\
    0 & 0 & 1 & 1 & 0 & 0 & 1 & 1 & 1 & 1
    \end{array}\right)$
};
\node[spacebox, fill=red!5, text width=6.3cm] (Qubit) at (4.2, 0) {
    \textbf{Physical Qubit Space} \\[0.4ex]
    \footnotesize
    \textit{Pauli layer:}\\[0.2ex]
    $v \mapsto X^v \sim \underbrace{IIIXX}_{\text{Left}} \;\big|\; \underbrace{XIXXX}_{\text{Right}}$\\[0.5ex]
    \normalsize $\Big\downarrow$ \footnotesize stabilizes\\[0.4ex]
    \rule{\linewidth}{0.3pt}\\[0.3ex]
    \textit{Code layer:}\\[0.2ex]
    \footnotesize
    $\mathcal{Q} = \bigl\{\,|\psi\rangle \in (\mathbb{C}^2)^{\otimes 10}  : S\,|\psi\rangle = |\psi\rangle\;\forall S\in\mathcal{S}\,\bigr\}$
};
\draw[iso] (Matrix.north) --
    node[labeltext, pos=0.5, left=0.3cm, text width=3.0cm] {
        $\mathrm{circ}(pf), \mathrm{circ}(qf)$\\
        $\leftrightarrow (pf, qf) \in R_5^2$
    }
    (Poly.south west);
\draw[iso] (Poly.south east) --
    node[labeltext, pos=0.4, right=-0.3cm, text width=4.2cm] {
        coeff.\ vector $\leftrightarrow$ \\ Pauli-$X$ generator
    }
    (Qubit.north);
\draw[iso] (Matrix.south)
    to[out=-90, in=-90, looseness=0.6]
    node[labeltext, midway, below=0.1cm, text width=4.5cm] {
        rows of $H_X$ $\leftrightarrow$ \\ Pauli-$X$ generators
    }
    (Qubit.south);
\end{tikzpicture}
\caption{The quantum three-space equivalence for a GB code with $\ell = 5$, $f = 1 + x$, $p = x^3$, $q = 1 + x + x^3$. A single stabilizer generator $(pf, qf) \in R_5^2$ maps to the first row of $H_X$ in the matrix space. In the physical qubit space $(\mathbb{C}^2)^{\otimes 10}$, that row's coefficient vector $(0,0,0,1,1\,|\,1,0,1,1,1)$ encodes a Pauli-$X$ generator $S_X = IIIXX\,|\,XIXXX$ (Pauli layer); the full stabilizer group $\mathcal{S}$ generated by all such operators defines the code space $\mathcal{Q}$ as their joint $+1$ eigenspace (Code layer). Each subsequent row of $H_X$ is a cyclic shift within each block, corresponding to multiplication by $x$ in $R_5^2$.  $H_Z$ is similarly involved, but not shown here.}
\label{fig:3spacequantumexample}
\end{figure}

This three-space perspective immediately opens up avenues to assess the automorphism structure algebraically. It is immediately obvious (and well known) that all GB codes have cyclic shift automorphisms that act on each length $\ell$ block of $H_X, H_Z$. However, the three-space dependency makes a further symmetry of the code immediately clear: all generalized bicycle codes have an $H$-type fold-transversal gate, achievable via composing full Left/Right block-swaps along the $i, i+\ell$ fold with the $\psi_{x^{-1}}$ automorphism, discussed in much greater depth in Sec.~\ref{sec:AutomorphismStructure}. As we shall see, given the three-space picture, much more will be able to be said about the prevalence of code automorphisms, and in particular, how to design codes around ensuring a rich automorphism structure.

\section{The Automorphism Structure of GB codes}
\label{sec:AutomorphismStructure}

The cyclic submodule structure established in Sec.~\ref{sec:threespaces} frames Generalized Bicycle codes as pairs of cyclic submodules of $R_\ell^2$ defined by $\ell, f_1 = pf, f_2 = qf$. In this section we exploit this algebraic structure to obtain coordinate permutations $\phi$ that map the submodule pair back to itself or exchange its two components --- i.e., code automorphisms and ZX dualities, each of which yields a potential fault-tolerant gate. This turns the search for code automorphisms, naively a brute-force search over $S_{2\ell}$, into a short list of checkable algebraic conditions on $f, p, q$ with respect to the permutation $\phi$, leveraging foundations previously laid in classical cyclic coding theory.

The section is organized as follows. Sec.~\ref{subsec:rowspace_auto_types} introduces the criterion for $\phi$ to be either a rowspace-preserving or swapping code automorphism (Corollary~\ref{cor:automorphism_criterion}), and follows with three simplifications we adopt throughout the rest of the work: block-separability, affine $\psi_L, \psi_R$, and $\psi_L = \psi_R$. These restrictions exclude genuinely interesting cases that we flag as future work, but they capture the core algebraic phenomenology of GB code automorphisms and allow us to lean on classical cyclic coding theory.

Sec.~\ref{subsec:bsep_auto} is the technical core, explicitly demonstrating non-trivial GB code automorphisms. Prop.~\ref{prop:block_automorphisms} shows that under the simplifications of Sec.~\ref{subsec:rowspace_auto_types}, the search for block-separable code automorphisms reduces to checking when a single permutation $\psi$ acts compatibly on the classical cyclic codes $\langle pf\rangle, \langle qf\rangle, \langle \overleftarrow{pf}\rangle, \langle \overleftarrow{qf}\rangle$, subject to a synchronization condition that connects the halves of each block. We then work through the three ring-level ingredients that generate all such automorphisms: cyclic shifts $\psi_i$ (Sec.~\ref{subsubsec:bsa_cyclic}), the universal block-swap involution $\Sigma = \sigma \circ (\psi_{x^{-1}} \oplus \psi_{x^{-1}})$ (Sec.~\ref{subsubsec:bsa_block}), and substitution multipliers $\psi_{x^j}$ with $\gcd(j,\ell)=1$, and the full catalog built from these is collected in Table~\ref{tab:block_sep_auts}.

Sec.~\ref{subsec:foldgates} feeds this catalog into the fold-transversal gate framework of~\cite{Breuckmann_foldtransversal}. $H$-type gates follow from rowspace-swapping automorphisms, so the new content concerns $S$- and $CX$-type gates. For $S$-type gates we apply a slight modification of \cite[Theorem~7]{Breuckmann_foldtransversal} (stated as Theorem~\ref{thm:fold_s_gate}) to the qualifying automorphisms of Table~\ref{tab:block_sep_auts}. For $CX$-type gates, where no prior framework currently exists, we derive sufficient algebraic conditions (Theorem~\ref{thm:fold_cx}) under which a physical CNOT along the $(i, i+\ell)$ fold, composed with a multiplier $\psi \oplus \psi$, extends to a code automorphism. The full set of $S$- and $CX$-type conditions is collected in Table~\ref{tab:fold_gates}.

The block-separable catalog of Sec.~\ref{subsec:bsep_auto} together with the fold-transversal gates of Sec.~\ref{subsec:foldgates} captures the primary automorphism structure of a GB code. The complementary question of logical operators and the action of any of these automorphisms is the subject of Sec.~\ref{sec:LogicalOps}.

\subsection{Categorizing Automorphisms and Simplifying Assumptions}
\label{subsec:rowspace_auto_types}

Code automorphisms come in two general flavors. \emph{Rowspace-preserving} automorphisms, which are cyclic submodule automorphisms, map $H_X$ back to $H_X$ and $H_Z$ back to $H_Z$. \emph{Rowspace-swapping} automorphisms, notably \textit{not} cyclic submodule automorphisms, exchange $H_X \leftrightarrow H_Z$, and are identified as the ZX dualities of Ref.~\cite{Breuckmann_foldtransversal}.

For our purposes, the algebraic machinery for analyzing the two cases is the same. What we are tracking, in either case, is whether the automorphism maps the cyclic submodules $M_f(p, q)$ and $M_{\overleftarrow{f}}(\overleftarrow{q}, \overleftarrow{p})$ back to themselves or exchanges them — a question entirely about the polynomial ring space, with no reference to which submodule carries $X$-type checks and which carries $Z$-type checks. This perspective matches the underlying classical code viewpoint of Section~2.2 of Ref.~\cite{Breuckmann_foldtransversal}, which treats $H_X$ and $H_Z$ symmetrically and only later assigns Pauli type to each block.

We accordingly adopt the labels $M_\sim$ for rowspace-preserving maps and $M_\leftrightarrow$ for rowspace-swapping maps. Once Pauli type is assigned in the matrix space, $M_\sim$ maps realize quantum code automorphisms and $M_\leftrightarrow$ maps realize fold-transversal gates via ZX duality. In a slight abuse of notation, we will refer to any map $\phi$ satisfying $M_\sim$ or $M_\leftrightarrow$ as a code automorphism, as the maps preserve the cyclic submodule pairs, even if one type induces a quantum code automorphism while the other is a ZX duality.

Theorem \ref{thm:cyclic_submodule} yields the following immediate corollary:
\begin{cor}\label{cor:automorphism_criterion}
A permutation $\phi \in S_{2\ell}$ induces a code automorphism of a 
Generalized Bicycle code with stabilizer rowspaces as in 
\eqref{eq:rowspaces} if and only if one of the following holds:
\begin{description}
    \item[Rowspace-preserving ($M_\sim$):] $\phi$ maps each of
    $M_f(p,q)$ and $M_{\overleftarrow{f}}(\overleftarrow{q}, \overleftarrow{p})$ back to themselves:
    \bea
    \phi\left((pa, qa)\right) = (pa', qa') \qquad \phi\left((\overleftarrow{qa}, \overleftarrow{pa})\right) = (\overleftarrow{qa''}, \overleftarrow{pa''}) \qquad a', a'' \in \langle f \rangle
    \eea
    \item[Rowspace-swapping ($M_\leftrightarrow$):] $\phi$ maps 
    $M_f(p,q)$ to 
    $M_{\overleftarrow{f}}(\overleftarrow{q}, \overleftarrow{p})$ and $M_{\overleftarrow{f}}(\overleftarrow{q}, \overleftarrow{p})$ to $M_f(p,q)$:
    \bea
    \phi\left((pa, qa)\right) = (\overleftarrow{qa'}, \overleftarrow{pa'}) \qquad \phi\left((\overleftarrow{qa}, \overleftarrow{pa})\right) = (pa'', qa'') \qquad a', a'' \in \langle f \rangle
    \eea
\end{description}
In each case, as $\phi$ is a bijection between finite sets of equal cardinality, it suffices to verify a single inclusion direction.
\end{cor}
\begin{proof}
Immediate from Theorem~\ref{thm:cyclic_submodule} and the definition of a CSS code automorphism as a coordinate permutation preserving the stabilizer group.
\end{proof}

Note that it is not required that each individual $(pa, qa)$ element contained in $M_f(p,q)$ is mapped exactly back to itself, just that $(pa, qa)$ maps to $(pa', qa')$ for $a' \in \langle f \rangle$. We use $a', a''$ above to stress that $\phi$ does not need to send $(pa, qa), (\overleftarrow{qa}, \overleftarrow{pa})$ for the same $a$ to outputs dependent on the same $a'$ - all that is required is that the same $a', a''$ appear on the LHS and RHS of a single submodule rowspace element.

Within this framework, there are three key simplifications this work takes, with permutations lying outside of these simplifications left as future work:
\begin{enumerate}
    \item \textbf{Block-separability}: All $\phi$ we consider are \emph{block-separable} with canonical form $\phi = \sigma^\epsilon \circ (\psi_L \oplus \psi_R)$. $\epsilon \in \{0, 1\}$ and $\psi_L, \psi_R \in S_{\ell}$ act on coordinate indices $L = \{1, \dots, \ell\}, R = \{1 + \ell, \dots, 2\ell\}$ respectively. We call $\sigma$ the ``block-swap'' map, swapping qubits along the $(i, i+\ell)$ fold for $i \in L$. $\phi$ acts on the two-block decomposition $\bbF_2^{2\ell} = \bbF_2^\ell \oplus \bbF_2^\ell$ either by preserving each block or by swapping the two blocks as wholes; coordinates are never mixed between blocks.  
    
    As we will see in Prop. \ref{prop:block_automorphisms}, such a canonical form transforms a question about cyclic submodule preservation into a question about classical cyclic code preservation: permutation automorphisms and equivalences. This allows us to lean on decades of cyclic coding theory research to identify and categorize GB code automorphisms. 

    We note that the \emph{block mixing} case, permutations that swap some non-empty subset of $L, R$ but not all elements, have a trivial case where block mixing permutations are guaranteed to exist: when $p = x^iq$, or vice versa. In such a case, the circulant matrices $\mathrm{circ}(pf), \mathrm{circ}(qf)$ that define $H_X, H_Z$ have the same columns up to an indexing shift by $i$, and each $j, j + i$ column pair for $0 \leq j \leq \ell-1$ are the same. Swapping columns between the left and right sides that are identical yields a trivial code automorphism with a potential non-trivial logical action. However, observe that every such code must have distance 2: a weight 2 error located exactly on positions $j$ and $j + i$ produces the exact same syndrome and cancels out at every location. Thus, these codes are not particularly interesting. 

    Block mixing automorphisms beyond the $p = x^iq$ case are not addressed in this work, and we invite a more thorough investigation of the block-separable case as future work.

    \item \textbf{$\bm \psi_L, \bm \psi_R$ affine}: As we will see, the assumption of block-separability implies $\psi_L, \psi_R$ must induce permutation automorphisms or equivalences on the classical cyclic codes defined by $\langle pf \rangle, \langle qf\rangle, \langle \overleftarrow{pf}\rangle, \langle \overleftarrow{qf}\rangle$. As Appendix~\ref{appen:codingtheory} discusses, the classical literature has studied this extensively, and cyclic codes yielding non-affine permutation structure are sporadic, and in some cases, not fully characterized. In the case that $R_\ell$ admits non-affine permutation equivalences or automorphisms, a block-separable GB code automorphism imposes further simultaneous constraints across all four defining classical cyclic codes, and the cyclic submodule requirement that $a'$ appears on both sides of a submodule element must still hold. This synchronized condition is much more restrictive than admitting a non-affine equivalence on a single ideal, making non-affine GB automorphisms exceedingly unlikely even in regimes where individual non-affine equivalences are known to exist, however not impossible. We leave such an analysis as a secondary interesting direction for future work.

    \item \textbf{$\bm \psi_L \bm = \bm \psi_R$}: Restricting $\psi_L, \psi_R$ to affine permutations, the action of each map on a monomial $x^k$ can be represented as $\psi_L: k \mapsto  kj_L + i_L, \ \psi_R: k \mapsto kj_R + i_R$. In this work, we restrict our attention to the fully symmetric case --- when $j_L = j_R, i_L = i_R$, and so $\psi_L = \psi_R$. This greatly simplifies the analysis for the $H_Z$ rowspace, and avoids a guaranteed degeneracy condition\footnote{It is standard in cyclic coding theory to treat degenerate and non-degenerate analysis separately. When $j_L \neq j_R$, the equivalence $q\psi(pa) = p\psi(qa)$ can be used to derive a joint degeneracy condition on $p, q, f$. Degenerate cyclic codes are discussed in Appendix~\ref{subsec:cycliccodes}.} when $j_L \neq j_R$.  The remaining 3 cases, when $j_L \neq j_R, \ i_L \neq i_R$, or both are interesting and potentially fruitful, and we leave these as a potential third line of future work. 
\end{enumerate}
 
Recall the canonical form of a block-separable permutation from above:
\[
\phi = \sigma^\epsilon \circ (\psi_L \oplus \psi_R)
\]
Observe that $\sigma$ and $\psi$ commute up to swapping the labels on $\psi_L, \psi_R: \ \sigma \circ (\psi_L \oplus \psi_R) = (\psi_R \oplus \psi_L) \circ \sigma $. As $\sigma$ simply swaps qubits on the $(i, i + \ell)$ fold, $\sigma$ does not change the algebraic structure contained on the $LHS$ or $RHS$ of a given cyclic submodule element. For this reason we can think of $\sigma$ as a rowspace ``fixing'' operation, and $(\psi_L \oplus \psi_R)$ as the potentially non-trivial ring automorphism operations we want to find.

We have now reduced our problem to the following: Find $\psi_L \oplus \psi_R$ such that any of the four sets of equations hold:
\bea
\psi_L(pa) = pa', \ \ \psi_R(qa) = qa', \quad \quad  \psi_L(\overleftarrow{qa}) = \overleftarrow{qa''}, \ \ \psi_R(\overleftarrow{pa}) = \overleftarrow{pa''} \label{eq:sameideals}\\
\psi_L(pa) = qa', \ \ \psi_R(qa) = pa', \quad \quad  \psi_L(\overleftarrow{qa}) = \overleftarrow{pa''}, \ \ \psi_R(\overleftarrow{pa}) = \overleftarrow{qa''} \label{eq:diffideals1}\\
\psi_L(pa) = \overleftarrow{qa'}, \ \ \psi_R(qa) = \overleftarrow{pa'}, \quad \quad  \psi_L(\overleftarrow{qa}) = pa'', \ \ \psi_R(\overleftarrow{pa}) = qa'' \label{eq:diffideals2}\\
\psi_L(pa) = \overleftarrow{pa'}, \ \ \psi_R(qa) = \overleftarrow{qa'}, \quad \quad  \psi_L(\overleftarrow{qa}) = qa'', \ \ \psi_R(\overleftarrow{pa}) = pa'' \label{eq:diffideals3}
\eea
for $a', a'' \in \langle f \rangle$. Each of these conditions is now asking about permutations $\psi_L, \psi_R$ that either preserve a cyclic code (permutation automorphisms, in the case of \eqref{eq:sameideals}), or map a given cyclic code to a different cyclic code (permutation equivalence, in the case of \eqref{eq:diffideals1}-\eqref{eq:diffideals3}) over $R_\ell$. The only time we must consider the $R_\ell^2$ structure is in the requirement that $a'$ be the same across both maps. 

Permutation automorphisms and equivalences are well studied concepts in the classical cyclic coding theory literature, and decomposing the search for GB code automorphisms into a question about permutations of classical cyclic codes makes this problem tractable. This decomposition gives rise to Prop.~\ref{prop:block_automorphisms} characterizing block-separable GB code automorphisms, illustrated in Fig.~\ref{fig:prop32_three_space}. 

\begin{prop}\label{prop:block_automorphisms}
Let $C$ be a GB code defined by $(\ell, f, p, q)$ as in Theorem~\ref{thm:cyclic_submodule}, let $\psi_L, \psi_R$ be $S_{\ell}$ permutations of $\{1, \dots, \ell\}$, and let $\sigma \in S_{2\ell}$ denote the full block-swap, sending $(i, i+\ell) \mapsto (i+\ell, i)$ for $1 \le i \le \ell$. Consider block-separable permutations of the form $\phi = \sigma^\epsilon \circ (\psi_L \oplus \psi_R)$ with $\epsilon \in \{0, 1\}$, where $\circ$ denotes function composition with the right-hand factor applied first: $(f \circ g)(v) = f(g(v))$.

The four cases below are organized along two independent axes: whether the resulting $\phi$ satisfies $M_\sim$ or $M_\leftrightarrow$, and whether the construction is \emph{block-preserving} ($\epsilon = 0$, i.e., $\phi = \psi \oplus \psi$) or \emph{block-swapping} ($\epsilon = 1$, i.e., $\phi = \sigma \circ (\psi \oplus \psi)$).  In each case the stated condition on $\psi$ is a permutation automorphism or equivalence on each of $\langle pf \rangle, \langle qf \rangle, \langle \overleftarrow{pf} \rangle, \langle \overleftarrow{qf} \rangle$, satisfying a synchronization condition on $a$.
\begin{enumerate}
    \item \textbf{$M_\sim$ block-preserving} (Eq.~\eqref{eq:sameideals}):  If $\psi_L(pa) = pa'$, $\psi_R(qa) = qa'$, $\psi_L(\overleftarrow{qa}) = \overleftarrow{qa''}$, and $\psi_R(\overleftarrow{pa}) = \overleftarrow{pa''}$ for $a, a', a'' \in \langle f \rangle$, then $\phi = \psi_L \oplus \psi_R$ is a code automorphism preserving each rowspace:
    \bea
        \phi(pa, qa) = (pa', qa') \in M_f(p,q) \qquad \phi(\overleftarrow{qa}, \overleftarrow{pa}) = (\overleftarrow{qa''}, \overleftarrow{pa''}) \in M_{\overleftarrow{f}}(\overleftarrow{q}, \overleftarrow{p})
    \eea
    
    \item \textbf{$M_\sim$ block-swapping} (Eq.~\eqref{eq:diffideals1}):  If $\psi_L(pa) = qa'$, $\psi_R(qa) = pa'$, $\psi_L(\overleftarrow{qa}) = \overleftarrow{pa''}$, and $\psi_R(\overleftarrow{pa}) = \overleftarrow{qa''}$, then $\phi = \sigma \circ (\psi_L \oplus \psi_R)$ is a code automorphism preserving each rowspace:
    \bea
        \phi(pa, qa) = (pa', qa') \in M_f(p,q) \qquad \phi(\overleftarrow{qa}, \overleftarrow{pa}) = (\overleftarrow{qa''}, \overleftarrow{pa''}) \in M_{\overleftarrow{f}}(\overleftarrow{q}, \overleftarrow{p})
    \eea
    
    \item \textbf{$M_\leftrightarrow$ block-preserving} (Eq.~\eqref{eq:diffideals2}):  If $\psi_L(pa) = \overleftarrow{qa'}$, $\psi_R(qa) = \overleftarrow{pa'}$, $\psi_L(\overleftarrow{qa}) = pa''$, and $\psi_R(\overleftarrow{pa}) = qa''$, then $\phi = \psi_L \oplus \psi_R$ is a code automorphism exchanging the two rowspaces:
    \bea
        \phi(pa, qa) = (\overleftarrow{qa'}, \overleftarrow{pa'}) \in M_{\overleftarrow{f}}(\overleftarrow{q}, \overleftarrow{p}) \qquad \phi(\overleftarrow{qa}, \overleftarrow{pa}) = (pa'', qa'') \in M_f(p,q)
    \eea
    \item \textbf{$M_\leftrightarrow$ block-swapping} (Eq.~\eqref{eq:diffideals3}):  If $\psi_L(pa) = \overleftarrow{pa'}$, $\psi_R(qa) = \overleftarrow{qa'}$, $\psi_L(\overleftarrow{qa}) = qa''$, and $\psi_R(\overleftarrow{pa}) = pa''$, then $\phi = \sigma \circ (\psi_L \oplus \psi_R)$ is a code automorphism exchanging the two rowspaces:
    \bea
        \phi(pa, qa) = (\overleftarrow{qa'}, \overleftarrow{pa'}) \in M_{\overleftarrow{f}}(\overleftarrow{q}, \overleftarrow{p}) \qquad \phi(\overleftarrow{qa}, \overleftarrow{pa}) = (pa'', qa'') \in M_f(p,q)
    \eea
\end{enumerate}
In all four cases $\psi_L, \psi_R$ preserve the stabilizer group of the GB code and may act non-trivially on the encoded logical qubits.
\end{prop}
\begin{proof}
    Follows from the preceding discussion. 
\end{proof}
A figure depicting Prop.~\ref{prop:block_automorphisms} is contained in Figure~\ref{fig:prop32_three_space}.

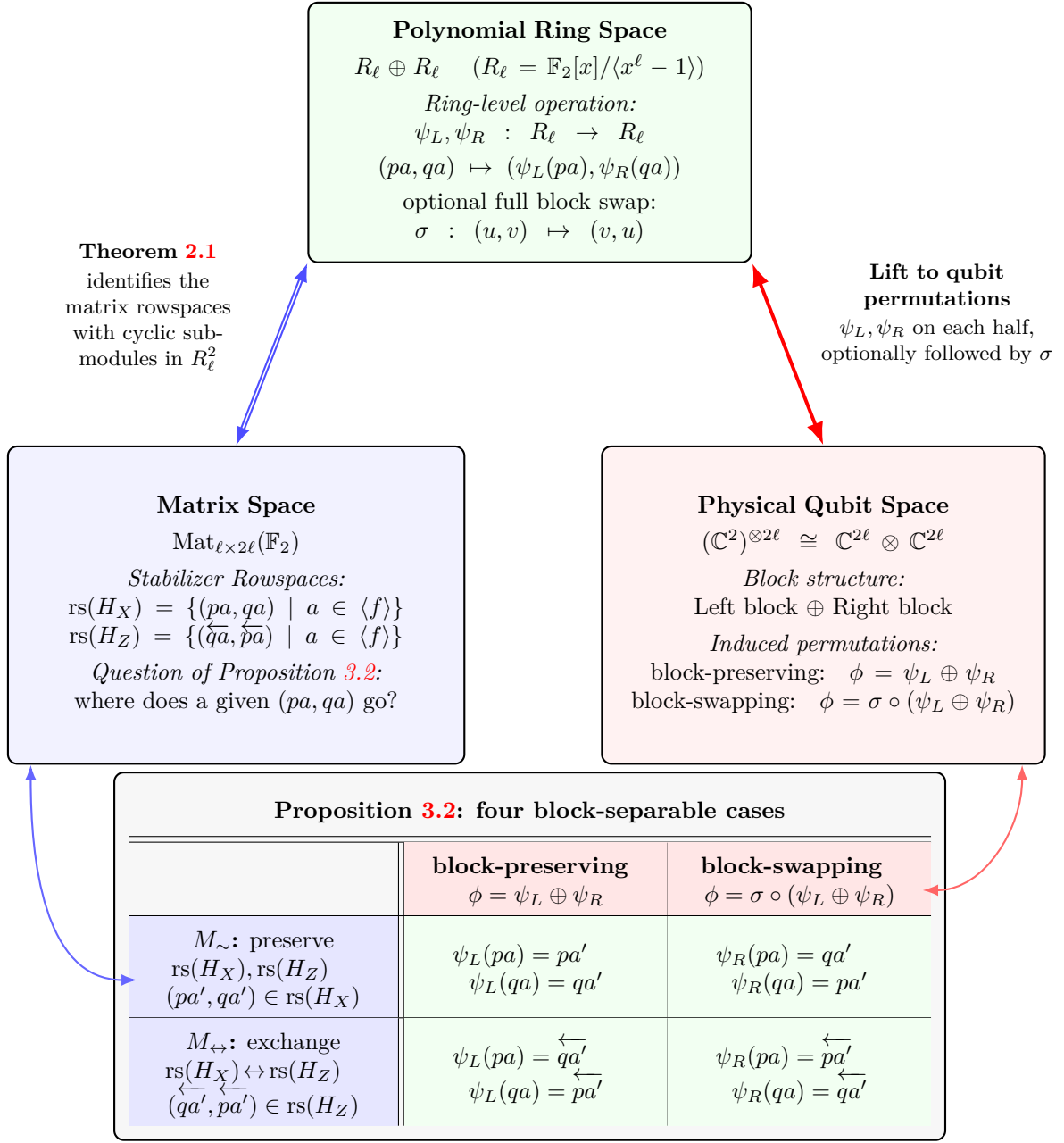
\begin{figure}[htbp]
\centering

\begin{tikzpicture}[
    >=Latex,
    spacebox/.style={
        rectangle, draw=black, thick, rounded corners,
        align=center, inner sep=8pt, minimum height=2.5cm
    },
    action/.style={
        <->, thick, draw=black!70
    },
    iso/.style={
        <->, double, thick, draw=blue!70
    },
    labeltext/.style={
        align=center, font=\small
    },
]


\node[spacebox, fill=green!5, text width=6.0cm] (Poly) at (0, 7.0) {
    \textbf{Polynomial Ring Space} \\[0.8ex]
    $R_\ell \oplus R_\ell$ \quad ($R_\ell = \mathbb{F}_2[x]/\langle x^\ell-1 \rangle$) \\[0.8ex]
    \textit{Ring-level operation:}\\
    $\psi_L, \psi_R : R_\ell \to R_\ell$ \\[0.6ex]
    $(pa,qa)\mapsto (\psi_L(pa),\psi_R(qa))$ \\[0.6ex]
    optional full block swap:\\
    $\sigma : (u,v)\mapsto (v,u)$
};

\node[spacebox, fill=blue!5, text width=6.2cm, minimum height=4.7cm] (Matrix) at (-4.35, 0) {
    \textbf{Matrix Space} \\[0.8ex]
    $\mathrm{Mat}_{\ell \times 2\ell}(\mathbb{F}_2)$ \\[0.8ex]
    \textit{Stabilizer Rowspaces:}\\
    $\operatorname{rs}(H_X) = \{(pa,qa)\mid a\in \langle f\rangle\}$\\
    $\operatorname{rs}(H_Z) = \{(\overleftarrow{qa},\overleftarrow{pa})\mid a\in \langle f\rangle\}$ \\[0.8ex]
    \textit{Question of Proposition~\ref{prop:block_automorphisms}:}\\
    where does a given $(pa,qa)$ go?
};

\node[spacebox, fill=red!5, text width=6.0cm, minimum height=4.7cm] (Physical) at (4.35, 0) {
    \textbf{Physical Qubit Space} \\[0.8ex]
    $(\mathbb{C}^2)^{\otimes 2\ell}\cong \mathbb{C}^{2\ell} \otimes \mathbb{C}^{2\ell}$ \\[0.8ex]
    \textit{Block structure:}\\
    Left block $\oplus$ Right block \\[0.8ex]
    \textit{Induced permutations:}\\
    block-preserving:\quad $\phi=\psi_L\oplus\psi_R$\\
    block-swapping:\quad $\phi=\sigma\circ(\psi_L\oplus\psi_R)$
};


\draw[iso] (Matrix.north) --
    node[labeltext, pos=0.76, left=0.45cm, text width=3.05cm] {
        \textbf{Theorem~\ref{thm:cyclic_submodule}}\\[0.3ex]
        identifies the matrix rowspaces\\
        with cyclic submodules in $R_\ell^2$
    }
    (Poly.south west);

\draw[action, draw=red, line width=1.5pt] (Poly.south east) --
    node[labeltext, pos=0.28, right=0.45cm, text width=3.7cm] {
        \textbf{Lift to qubit permutations}\\[0.3ex]
        $\psi_L, \psi_R$ on each half,\\
        optionally followed by $\sigma$
    }
    (Physical.north);



\node[draw=black, thick, rounded corners, fill=gray!6, inner sep=6pt] (PropTable) at (0,-5.2) {%
\renewcommand{\arraystretch}{1.55}%
\setlength{\tabcolsep}{7pt}%
\begin{tabular}{>{\centering\arraybackslash}m{3.5cm}||>{\centering\arraybackslash}m{3.4cm}|>{\centering\arraybackslash}m{3.4cm}}
\multicolumn{3}{c}{{\normalsize\bfseries Proposition~\ref{prop:block_automorphisms}: four block-separable cases}} \\[0.5ex]
\hline\hline
 & \cellcolor{red!10}\textbf{block-preserving}\newline $\phi = \psi_L\oplus\psi_R$
 & \cellcolor{red!10}\textbf{block-swapping}\newline $\phi = \sigma\circ(\psi_L\oplus\psi_R)$ \\[0.2ex]
\hline
\cellcolor{blue!10}\rule{0pt}{2.8ex}%
\textbf{$M_\sim$:} preserve $\operatorname{rs}(H_X),\operatorname{rs}(H_Z)$\newline
$(pa',qa') \in \operatorname{rs}(H_X)$
& \cellcolor{green!5} $\psi_L(pa)=pa'$\newline $\psi_L(qa)=qa'$
& \cellcolor{green!5} $\psi_R(pa)=qa'$\newline $\psi_R(qa)=pa'$ \\[0.6ex]
\hline
\cellcolor{blue!10}\rule{0pt}{2.8ex}%
\textbf{$M_\leftrightarrow$:} exchange $\operatorname{rs}(H_X)\!\leftrightarrow\!\operatorname{rs}(H_Z)$\newline
$(\overleftarrow{qa'},\overleftarrow{pa'}) \in \operatorname{rs}(H_Z)$
& \cellcolor{green!5} $\psi_L(pa)=\overleftarrow{qa'}$\newline $\psi_L(qa)=\overleftarrow{pa'}$
& \cellcolor{green!5} $\psi_R(pa)=\overleftarrow{pa'}$\newline $\psi_R(qa)=\overleftarrow{qa'}$ \\[0.4ex]
\end{tabular}%
};

%
\draw[<->, thick, draw=blue!60]
    ([xshift=10pt, yshift=0pt]Matrix.south west) to[out=270, in=180] ([xshift=10pt, yshift=-10pt]PropTable.west);
%
\draw[<->, thick, draw=red!60]
    ([xshift=-10pt, yshift=0pt]Physical.south east) to[out=-90, in=0] ([xshift=-10pt, yshift=-50pt]PropTable.north east);

\end{tikzpicture}
\caption{\textbf{Proposition~\ref{prop:block_automorphisms} as a three-space classification of block-separable automorphisms.}
The rowspaces $\operatorname{rs}(H_X)$ and $\operatorname{rs}(H_Z)$ belong to the \emph{Matrix Space}, while the left/right block decomposition belongs to the \emph{Physical Qubit Space}. Two ring-level maps $\psi_L, \psi_R$ on $R_\ell$, applied to each half and optionally composed with the full block swap $\sigma$, produces four cases organized by two independent binary choices: whether $\phi$ is block-preserving ($\phi=\psi_L\oplus\psi_R$) or block-swapping ($\phi=\sigma\circ(\psi_L\oplus\psi_R)$) in the \emph{Physical Qubit Space} (columns), and whether the transformed pairs remain in $\operatorname{rs}(H_X)$ (Type~I) or cross into $\operatorname{rs}(H_Z)$ (Type~II) in the \emph{Matrix Space} (rows).  Given these choices, the required conditions on $\psi_L, \psi_R$ is determined as shown in the table.}
\label{fig:prop32_three_space}
\end{figure}

Prop.~\ref{prop:block_automorphisms} and Theorem~\ref{thm:cyclic_submodule} yield the following immediate corollary:
\begin{cor}\label{cor:Hz_simplification}
    Under the conditions of Prop.~\ref{prop:block_automorphisms} and Theorem~\ref{thm:cyclic_submodule}, if $\psi = \psi_L = \psi_R$ is an affine map, then only the $M_f(p,q)$ case needs to be verified explicitly.
\end{cor}
\begin{proof}
    From Theorem~\ref{thm:cyclic_submodule}, the $H_X, H_Z$ rowspaces are related via the $\psi_{x^{-1}}$ automorphism. Both cyclic shifts and multipliers commute with the involution $\psi_{x^{-1}}$ in the relevant sense: for multipliers, $\psi_{x^j} \circ \psi_{x^{-1}} = \psi_{x^{-1}} \circ \psi_{x^j}$ as ring automorphisms of $R_\ell$; for cyclic shifts, $\psi_i \circ \psi_{x^{-1}} = \psi_{x^{-1}} \circ \psi_{-i}$, exchanging forward and backward shifts. In either case, the action of $\psi$ on the $H_X$ rowspace generators $(pa, qa)$ can be pushed through $\psi_{x^{-1}}$ to a corresponding action on the $H_Z$ rowspace generators $(\overleftarrow{qa}, \overleftarrow{pa})$. Thus, when $\psi_L = \psi_R$, the four cases of Prop.~\ref{prop:block_automorphisms} then reduce to verification on $M_f(p,q)$ alone. 
\end{proof}

As discussed previously, this work explores the case $\psi_L = \psi_R$. We accordingly write $\psi := \psi_L = \psi_R$ and consider only block-separable permutations of the form $\phi = \sigma^\epsilon \circ (\psi \oplus \psi)$ throughout.

\subsection{Block-separable $\psi_L = \psi_R$ automorphisms}
\label{subsec:bsep_auto}

The remainder of this subsection works through the three families of block-separable automorphisms admitted by the affine restriction. We denote $\psi_i$ as the cyclic shift map and $\psi_{x^j}$ as the multiplier ring automorphism map, both over $R_\ell$, with the corresponding $R_\ell^2$ operations obtained via $\phi_i = \psi_i \oplus \psi_i$ and $\phi_{x^j} = \psi_{x^j} \oplus \psi_{x^j}$. We treat cyclic shifts first; these are always $M_\sim$ and exist for every GB code.  We then identify a single combination of multiplier and full block-swap, $\Sigma = \sigma \circ \phi_{x^{-1}}$, which is always $M_\leftrightarrow$, the only multiplier-based block-swapping automorphism guaranteed to exist for every GB code, and the source of its universal $H$-type fold-transversal gate.  Finally, the remaining substitution multipliers $\phi_{x^j}$ are not guaranteed to exist; Theorem~\ref{thm:substitution_automorphisms} characterizes when they do via four explicit polynomial identities on $(p, q)$ in the quotient $S = R_\ell / \langle \hat{f} \rangle$, which categorize the relevant exponents into the sets $\mathrm{Stab}(p,q),\ \mathrm{Swap}(p,q),\ \mathrm{Inv}(p,q),$ and $\mathrm{SwapInv}(p,q)$.  The $M_\leftrightarrow$ type partner of each multiplier-based $M_\sim$ type automorphism is obtained by composition with $\Sigma$, and the full catalog is collected in Table~\ref{tab:block_sep_auts}.
 
\subsubsection{Cyclic shifts}
\label{subsubsec:bsa_cyclic}
 
The cyclic shift map $\phi_i = \psi_i \oplus \psi_i$ is the simplest automorphism class of any bicycle code.  Since each half of $H_X$ and $H_Z$ is a classical cyclic code, multiplication by any power of $x$ is absorbed into the ideal:
\begin{align}
    (pa, qa) \xrightarrow{\psi_i \oplus \psi_i} (x^i pa, x^i qa) &= (pa', qa') \in M_f(p,q) = \operatorname{rs}(H_X) \\
    (\overleftarrow{qb}, \overleftarrow{pb}) \xrightarrow{\psi_i \oplus \psi_i} (x^i \overleftarrow{qb}, x^i \overleftarrow{pb}) &= (\overleftarrow{qb'}, \overleftarrow{pb'}) \in M_{\overleftarrow{f}}(\overleftarrow{q}, \overleftarrow{p}) = \operatorname{rs}(H_Z)
\end{align}
where $a' = x^i a \in \langle f \rangle$ and $b' = x^{-i}b \in \langle f \rangle$.  These are always $M_\sim$ type.

\subsubsection{Block-swap}
\label{subsubsec:bsa_block}
 
The substitution multiplier $\psi_{x^{-1}} : f(x) \mapsto f(x^{-1})$, combined with a block-swap, always yields a GB code automorphism. As $\psi_{x^{-1}}$ is an involution ($\psi_{x^{-1}}^2 = \mathrm{id}$), applying it twice sends an element back to itself. In the notation we have defined, we have
\bea
\overleftarrow{\overleftarrow{a}} = a
\eea
Composing the block-swap $\sigma$ with blockwise application of $\psi_{x^{-1}}$ (in either order) yields a $M_\leftrightarrow$ automorphism:
\begin{align}
    (pa, qa) \xrightarrow{\sigma} (qa, pa) \xrightarrow{\psi_{x^{-1}} \oplus \psi_{x^{-1}}} \bigl(\overleftarrow{qa}, \overleftarrow{pa}\bigr) &\in  M_{\overleftarrow{f}}(\overleftarrow{q}, \overleftarrow{p}) = \operatorname{rs}(H_Z), \\[4pt]
    \bigl(\overleftarrow{qb}, \overleftarrow{pb}\bigr) \xrightarrow{\sigma} \bigl(\overleftarrow{pb}, \overleftarrow{qb}\bigr) \xrightarrow{\psi_{x^{-1}} \oplus \psi_{x^{-1}}} \bigl(pb, qb\bigr) &\in M_f(p,q) = \operatorname{rs}(H_X)
\end{align}
When considering $H_X, H_Z$ with their associated Paulis, this becomes a ZX duality and pairing the map with transversal Hadamard gates on each coordinate gives a fold-transversal H-type gate. Thus, $\Sigma = (\psi_{x^{-1}} \oplus \psi_{x^{-1}}) \circ \sigma = \sigma \circ (\psi_{x^{-1}} \oplus \psi_{x^{-1}})$ is always a $M_\leftrightarrow$ automorphism of any GB code. This code symmetry realizes the standard ZX-duality for 2BGA codes $\tau_0$ from~\cite{automorphisms_BBcodes}.

\medskip
Note that in the case of cyclic shifts and block-swaps we did not need to explicitly prove the $H_Z$ case, however it is simple and in the case of the block-swap, somewhat insightful to do so. In the case of the substitution multiplier we invoke Corollary~\ref{cor:Hz_simplification} to greatly simplify the proofs.

\subsubsection{Substitution multipliers}
 
We now turn to the first genuinely code-dependent class of block-separable automorphisms. 
Unlike the cyclic shifts and the universal block-swap, which exist for every GB code, the automorphisms in this subsection arise only when the defining polynomials of the code are compatible with non-trivial substitution symmetries of the polynomial ring $R_\ell$. We study the ring automorphisms
\bea
\psi_{x^j} : f(x) \mapsto f(x^j), \qquad j \in (\Z/\ell\Z)^\times,
\eea
which are well-defined exactly when $\gcd(j,\ell)=1$. On the qubit space, these act as coordinate 
permutations within each $\ell$-qubit block, and so are natural candidates for block-separable GB code automorphisms.

Importantly, not every such substitution preserves the GB code, nor does every permutation automorphism or permutation equivalence of $pf, qf, \overleftarrow{pf}, \overleftarrow{qf}$ yield a GB code automorphism. Recall from Theorem~\ref{thm:cyclic_submodule} that the $X$-stabilizer rowspace is the cyclic submodule
\bea
M_f(p,q)=\{(pa,qa)\mid a\in \langle f\rangle\}
\,,
\eea
where $f$ is the shared factor of the two defining polynomials and $p,q$ are the transfer polynomials. For $\psi_{x^j}$ to induce a code automorphism, two separate algebraic requirements must be met. First, the substitution must preserve the underlying ideal $\langle f\rangle$ itself --- i.e., the substitution must be a permutation automorphism of $\langle f \rangle$. This motivates the definition
\def\Pres{\mathrm{Pres}}
\bea
\Pres(f)=\{\, j\in (\Z/\ell\Z)^\times : \langle f(x^j)\rangle = \langle f\rangle \,\}
\,,
\eea
which records exactly those multipliers $j$ for which the common cyclic backbone of the GB code is unchanged (i.e., preserved). 
As established in Appendix~\ref{appen:codingtheory} (specifically Lemma~\ref{lem:equivalent_ideals}, which states that two principal ideals in $R_\ell$ are equal if and only if their generators are related by a unit), this preservation occurs if and only if $f(x^j) = uf$ for some unit $u \in R_\ell^\times$.
If $j\notin \Pres(f)$, then the substitution sends the parameter $a\in \langle f\rangle$ into a different ideal,
so the cyclic submodule structure of the rowspace is not preserved.

Second, once $j$ preserves $\langle f\rangle$, one must determine how the same substitution acts on the transfer
polynomials $p$ and $q$. Since the rowspace depends on $p$ and $q$ only through their action on $\langle f\rangle$,
this question is naturally asked in the polynomial quotient ring
\bea
S = R_\ell/\langle \hat f\rangle, \qquad \hat f = \frac{x^\ell-1}{f}.
\eea
The content of Theorem~\ref{thm:substitution_automorphisms}, presented below (and illustrated in Fig.~\ref{fig:thm33_cases}), is that there are only four algebraically relevant possibilities: under $x\mapsto x^j$, the pair $(p,q)$ may be fixed, swapped, inverted, or swapped after inversion. Each of these behaviors yields both a $M_\sim$ and a $M_\leftrightarrow$ code automorphism, composing with the full block-swap~$\sigma$ or the inversion multiplier $\psi_{x^{-1}}$.

In this way, Theorem~\ref{thm:substitution_automorphisms} converts the search for substitution-based code automorphisms when $\psi_L = \psi_R$ into a finite algebraic test. Rather than searching over coordinate permutations directly, one first determines $\Pres(f)$, and then checks which elements $j\in \Pres(f)$ act on $(p,q)$ in one of the four admissible ways in~$S$. This provides the basic mechanism by which non-trivial multiplier symmetries of the polynomial ring become explicit automorphisms of the GB code.

\begin{theorem}
\label{thm:substitution_automorphisms}
Let $C$ be a GB code defined by $(\ell, f, p, q)$ as in Theorem~\ref{thm:cyclic_submodule}, and let $j \in (\mathbb{Z}/\ell\mathbb{Z})^\times$.  Define
\bea
\operatorname{Pres}(f) = \big\{j \in (\mathbb{Z}/\ell\mathbb{Z})^\times 
: \langle f(x^j) \rangle = \langle f \rangle\big\}.
\eea
If $j \in \operatorname{Pres}(f)$ and the images of $p, q$ under $\psi_{x^j}$ satisfy either of the following conditions in $S = R_\ell / \langle \hat{f} \rangle$, with $u \in S^\times$, then $j$ induces a code automorphism:
\begin{enumerate}[label=(\arabic*)]
    \item $p(x^j) = up$ and $q(x^j) = uq$ in $S$. 
    Then $\phi = \psi_{x^j} \oplus \psi_{x^j}$ is $M_\sim$.
    \item $p(x^j) = uq$ and $q(x^j) = up$ in $S$. 
    Then $\phi = \sigma \circ (\psi_{x^j} \oplus \psi_{x^j})$ is $M_\sim$.
\end{enumerate}
Additionally, if $\gcd(p, \hat{f}) = \gcd(q, \hat{f}) = 1$ so that $p^{-1}$ and $q^{-1}$ exist in $S$, and the images of $p, q$ under $\psi_{x^j}$ satisfy either of the following conditions, then $j$ induces a code automorphism:
\begin{enumerate}[label=(\arabic*), resume]
    \item $p(x^j) = up^{-1}$ and $q(x^j) = uq^{-1}$ in $S$. 
    Then $\phi = \sigma \circ (\psi_{x^j} \oplus \psi_{x^j})$ is 
    $M_\sim$.
    \item $p(x^j) = uq^{-1}$ and $q(x^j) = up^{-1}$ in $S$. 
    Then $\phi = \psi_{x^j} \oplus \psi_{x^j}$ is $M_\sim$.
\end{enumerate}
\end{theorem}

\input{figures/substitution_mult_fig.tex}
\begin{proof}
First, observe that by Lemma \ref{lem:annihilator_sufficiency}, any of the four algebraic conditions on the transfer polynomials $p, q$ need only hold in $S$, not in $R_\ell$ itself.  The rowspace elements $(pa, qa)$ and $(\overleftarrow{qa}, \overleftarrow{pa})$ depend on $p, q, \overleftarrow{p}, \overleftarrow{q}$ only through their action on $\langle f \rangle$ and $\langle \overleftarrow{f} \rangle$, and this action is insensitive to any component of $p, q, \overleftarrow{p}, \overleftarrow{q}$ lying in $\langle \hat{f} \rangle$ or $\langle \hat{f}(x^{-1})\rangle$. Separately, since $u \in S^\times$, common-unit rescaling preserves $M_f(p,q)$ (Lemma~\ref{lem:equivalent_modules}), so we may absorb such unit factors $u$ into the free ideal element\footnote{$\overleftarrow{\hat{f}}$ is cumbersome and beginning to overload notation, and we denote this explicitly by $\hat{f}(x^{-1})$.}.

We verify cases 1 and 3 explicitly by applying $\psi_{x^j} \oplus \psi_{x^j}$ (possibly composed with $\sigma$) to a generic element $(pa, qa) \in \operatorname{rs}(H_X)$ and checking that the image lies in the appropriate rowspace. Cases 2 and 4 follow analogously. We prove the $H_X$ condition, as by Cor.~\ref{cor:Hz_simplification}, the corresponding $H_Z$ condition follows automatically. Note that any condition of the form  $p(x^j) = up$ in $S$ implies $\overleftarrow{p}(x^j) = \overleftarrow{up}$ in $\overleftarrow{S} = R_\ell/\langle \hat{f}(x^{-1}) \rangle$.

\textbf{Case 1:} $q(x^j) = uq, \ p(x^j) = up$ in $S$.
If $j$ stabilizes $p, q$ up to the common unit $u$, then $\psi_{x^j}$ applied blockwise is immediately $M_\sim$:
\begin{align}
    (pa, qa) \xrightarrow{\psi_{x^j} \oplus \psi_{x^j}} \bigl(p(x^j)a(x^j), \ q(x^j)\,a(x^j)\bigr) &= (up\,a', \, uq\,a') = (pa'', qa'') \in M_f(p,q) =  \operatorname{rs}(H_X)
\end{align}
where $a' = a(x^j) \in \langle f(x^j) \rangle = \langle f \rangle$ and $a'' = u a' \in \langle f \rangle$. Thus, $\phi_{x^{j}} = \psi_{x^j} \oplus \psi_{x^j}$ is a $M_\sim$ code automorphism.

Case 2 follows in the same way, invoking the block-swap $\sigma$ after application of $\psi_{x^j} \oplus \psi_{x^j}$ yields $(uq\,a', up\,a') = (qa'', pa'')$ to arrive back at $(pa'', qa'')$.

In cases 3 and 4 we may assume $q, p$ are invertible in $S$.  By Lemma~\ref{lem:annihilator_invertible}, this holds if and only if $\gcd(q, \hat{f}) = \gcd(p, \hat{f}) = 1$. If not, these cases do not yield code automorphisms.

\paragraph{Case 3: $q(x^j) = uq^{-1}(x)$ and $p(x^j) = up^{-1}(x)$ in $S$.}
If $j$ sends $p, q$ to their inverses up to the common unit $u$, then $\psi_{x^j}$ applied blockwise alone does not preserve either rowspace. However, composing with the block-swap $\sigma$ yields a $M_\sim$ automorphism:
\begin{align}
    (pa, qa) \xrightarrow{\psi_{x^j} \oplus \psi_{x^j}} \bigl(p(x^j)a(x^j), \ q(x^j)\,a(x^j)\bigr) &= (up^{-1}a', uq^{-1}a')
\end{align}
where $a' = a(x^j) \in \langle f(x^j) \rangle = \langle f \rangle$. Setting $\tilde{a} = u a' \in \langle f \rangle$, this is $(p^{-1}\tilde{a}, q^{-1}\tilde{a})$. Observe the following\footnote{
Since $\hat{f} \cdot \tilde{a} = 0$ in $R_\ell$ for any $\tilde{a} \in \langle f \rangle$ (as $\hat{f} \cdot f = x^\ell - 1 = 0$ in $R_\ell$), the action of $q^{-1}p^{-1}$ on elements of $\langle f \rangle$ is well-defined in $R_\ell$ itself: any two lifts of $q^{-1}p^{-1}$ from $R_\ell/\langle \hat{f}\rangle$ to $R_\ell$ differ by a multiple of $\hat{f}$, which annihilates $\langle f \rangle$. Thus $a'' = q^{-1}p^{-1}\tilde{a} \in \langle f \rangle$ since $\langle f \rangle$ is an ideal.}:
\bea
a'' = q^{-1}p^{-1}\tilde{a} \in \langle f \rangle \quad \quad \tilde{a} = a''qp \in \langle f \rangle
\eea
as we are working in a commutative ring, we have
\bea
p^{-1}\tilde{a} = p^{-1}qpa'' = qa'' \quad \quad q^{-1}\tilde{a} = q^{-1}qpa'' = pa'' 
\eea
and so
\begin{align}
     (p^{-1}\tilde{a}, q^{-1}\tilde{a}) &= (qa'', pa'')
\end{align}
for $a'' \in \langle f \rangle$. Once again, applying $\sigma$, we have:
\begin{align}
    (qa'', pa'') \xrightarrow{\sigma} &(pa'', qa'') \in M_f(p,q) = \operatorname{rs}(H_X)
\end{align}
Thus, $\phi_{x^j} = \sigma \circ (\psi_{x^j} \oplus \psi_{x^j})$ is a $M_\sim$ automorphism.

Case 4 follows in the same way, setting $\tilde{a} = u a(x^j) \in \langle f \rangle$ and $a'' = q^{-1}p^{-1}\tilde{a} \in \langle f \rangle$.

\end{proof}

Thus, we have shown conditions on $f, \ell, p, q$ that yield a much richer set of code automorphisms than cyclic shifts and block-swaps alone. Observe that the four cases from Theorem \ref{thm:substitution_automorphisms} categorize the useful substitution multipliers into the sets
\begin{align}
\mathrm{Stab}(p,q) &= \{j \in \mathrm{Pres}(f) : p(x^j) = up,\  q(x^j) = uq \}\label{eq:stab_set}\\
\mathrm{Swap}(p,q) &= \{j \in \mathrm{Pres}(f) : p(x^j) = uq, \ q(x^j) = up \}\label{eq:swap_set} \\
\mathrm{Inv}(p,q) &= \{j \in \mathrm{Pres}(f)  : p(x^j) = up^{-1}, \ q(x^j) = uq^{-1} \}\label{eq:inv_set} \\
\mathrm{SwapInv}(p,q) &= \{j \in \mathrm{Pres}(f)  : p(x^j) = uq^{-1}, \ q(x^j) = up^{-1} \}\label{eq:swapinv_set}
\end{align}
Before we catalog the full set of block-separable automorphisms, we end with a corollary:
\begin{cor}
    Composing any automorphism with the block-swap involution $\Sigma =(\psi_{x^{-1}} \oplus \psi_{x^{-1}}) \circ \sigma$ yields a code automorphism of the opposite type.
\end{cor}
\begin{proof}
    Composition of two automorphisms yields an automorphism. A $M_\sim$ automorphism preserves each rowspace; a $M_\leftrightarrow$ automorphism exchanges them. Since $\Sigma$ is $M_\leftrightarrow$, composing it with a $M_\sim$ automorphism exchanges the rowspaces ($M_\leftrightarrow$), while composing it with a $M_\leftrightarrow$ automorphism preserves them ($M_\leftrightarrow \circ M_\leftrightarrow = M_\sim$).
\end{proof}

We catalog the full set of block-separable automorphisms generated by the affine group in Table \ref{tab:block_sep_auts}. Writing $\Lambda = \mathrm{Stab}(p,q) \cup \mathrm{Inv}(p,q) \cup \mathrm{Swap}(p,q) \cup \mathrm{SwapInv}(p,q)$, the total number of block-separable automorphisms (including shifts) obtained from $\psi_L = \psi_R$ is
\bea
|\mathcal{G}| = 2\ell \cdot |\Lambda|.
\eea

\begin{table*}[tb]
\centering
\renewcommand{\arraystretch}{1.6}
\small
\begin{tabular}{c p{3.8cm} c p{6.2cm}}
\hline
Source & Automorphism & Type & Conditions \\
\hline
$i \in \mathbb{Z}/\ell\mathbb{Z}$ & $\phi_i$ & $M_\sim$ & None (Cyclic shifts) \\[4pt]
$j \in \mathrm{Stab}(p,q)$ & $\Phi_j^{\mathrm{stab}} = \phi_{x^j}$ & $M_\sim$ & $\langle f(x^j)\rangle = \langle f \rangle$; \newline $p(x^j) = up$, $q(x^j) = uq$ in $S$ \\[4pt]
$j \in \mathrm{Swap}(p,q)$ & $\Phi_j^{\mathrm{swap}} = \sigma \circ \phi_{x^j}$ & $M_\sim$ & $\langle f(x^j)\rangle = \langle f \rangle$; \newline $p(x^j) = uq$, $q(x^j) = up$ in $S$ \\[4pt]
$j \in \mathrm{Inv}(p,q)$ & $\Phi_j^{\mathrm{inv}} = \sigma \circ \phi_{x^j}$ & $M_\sim$ & $\langle f(x^j)\rangle = \langle f \rangle$; \newline $\gcd(p, \hat{f}) = \gcd(q, \hat{f}) = 1$; \newline $p(x^j) = up^{-1}$, $q(x^j) = uq^{-1}$ in $S$ \\[4pt]
$j \in \mathrm{SwapInv}(p,q)$ & $\Phi_j^{\mathrm{si}} = \phi_{x^j}$ & $M_\sim$ & $\langle f(x^j)\rangle = \langle f \rangle$; \newline $\gcd(p, \hat{f}) = \gcd(q, \hat{f}) = 1$; \newline $p(x^j) = uq^{-1}$, $q(x^j) = up^{-1}$ in $S$ \\[4pt]
\hline
--- & $\Sigma = \sigma \circ \phi_{x^{-1}}$ & $M_\leftrightarrow$ & None \\[4pt]
$j \in \mathrm{Stab}(p,q)$ & $\Sigma\circ \Phi_j^{\mathrm{stab}} = \sigma \circ \phi_{x^{-j}}$ & $M_\leftrightarrow$ & Same as Stab \\[4pt]
$j \in \mathrm{Swap}(p,q)$ & $\Sigma \circ \Phi_j^{\mathrm{swap}} = \phi_{x^{-j}}$ & $M_\leftrightarrow$ & Same as Swap \\[4pt]
$j \in \mathrm{Inv}(p,q)$ & $\Sigma \circ \Phi_j^{\mathrm{inv}} = \phi_{x^{-j}}$ & $M_\leftrightarrow$ & Same as Inv \\[4pt]
$j \in \mathrm{SwapInv}(p,q)$ & $\Sigma \circ \Phi_j^{\mathrm{si}} = \sigma \circ \phi_{x^{-j}}$ & $M_\leftrightarrow$ & Same as SwapInv \\
\hline
\end{tabular}
\caption{Block-separable automorphisms for general GB codes. $\phi_{x^j} = \psi_{x^j} \oplus \psi_{x^j}$ denotes the multiplier ring automorphism $x^j$ and $\phi_i = \psi_i \oplus \psi_i$ denote the cyclic shift by $i$-indices automorphism, each acting on each block separately. Additionally, $\phi_{x^{a}} \circ \phi_{x^b} = \phi_{x^{ab}}$, and in particular, we combine sequential $\phi_{x^{-1}}, \phi_{x^j}$ automorphisms into $\phi_{x^{-j}}$. $u$ denotes an arbitrary element of $S^\times$. Note that $\sigma$ commutes with $\phi_{x^j}, \phi_i$, yielding the $M_\leftrightarrow$ automorphisms that exchange the $H_X$ and $H_Z$ rowspaces. Here $\hat{f} = (x^\ell - 1)/f$ denotes the annihilator of $f$, and all equalities involving $p, q$ are in $S = R_\ell / \langle \hat{f} \rangle$. The $M_\sim$ code automorphisms yield quantum code automorphisms when assigning $X, Z$ Paulis to $M_f(p,q)$ and $M_{\overleftarrow{f}}(\overleftarrow{q}, \overleftarrow{p})$, while $M_\leftrightarrow$ yields ZX dualities.}
\label{tab:block_sep_auts}
\end{table*}
 
\subsection{Fold-Transversal Gates}
\label{subsec:foldgates}

The block-separable catalog of Sec.~\ref{subsec:bsep_auto} produces an explicit list of $M_\sim$ and $M_\leftrightarrow$ code automorphisms of any GB code when $\psi_L = \psi_R$.  This subsection converts that list into fault-tolerant logical Clifford gates via the \emph{fold-transversal} gate framework of~\cite{Breuckmann_foldtransversal}.

A fold-transversal gate is a logical gate implemented by combining a geometric pairing (''folding'' of qubits) with a semi-uniform local physical operation. In this work we extend the fold-transversal gate framework by considering logical gates ``up to automorphism'' --- that is, we allow some permutation of physical qubits after implementing the physical transversal gates in order to preserve the rowspaces. We present the new fold-transversal gates found, supplemented with conditions for fold-transversal CNOTs. 

Note that the $H$-type case is subsumed by Sec.~\ref{subsec:bsep_auto}: pairing any $M_\leftrightarrow$ automorphism with transversal Hadamards realizes an $H$-type fold-transversal gate.  The remaining Clifford generators, $S$ and $\mathrm{CX}$, require additional algebraic conditions, and are the focus of this subsection.

For \textbf{$S$-type gates} (Sec.~\ref{sec:phase_type}), we apply a slightly modified version of Theorem~7 of~\cite{Breuckmann_foldtransversal}, which lifts a ZX-duality $\tau$ to a logical $S$-type Clifford whenever $\tau$ is an involution and satisfies an even-overlap condition on its fixed-point set. For \textbf{$\mathrm{CX}$-type gates} (Sec.~\ref{sec:cx_type}), no prior framework is presently known, and we derive sufficient algebraic conditions under which a fold-transversal CNOT along the $(i, i+\ell)$ fold composed with a ring automorphism $\psi \oplus \psi$ extends to a code automorphism (Theorem~\ref{thm:fold_cx}); the construction has four variants, indexed by the CNOT direction (block~1 $\to$ block~2 or vice versa) and by the rowspace behavior of the resulting automorphism ($M_\sim$ or $M_\leftrightarrow$).

All resulting conditions, both $S$-type (necessary and sufficient) and $\mathrm{CX}$-type (sufficient), are collected in Table~\ref{tab:fold_gates}.

\subsubsection{Phase-type gates}
\label{sec:phase_type}

The full set of $M_\leftrightarrow$ code automorphisms yielding ZX dualities in Table~\ref{tab:block_sep_auts} naturally splits into two families along the block-action axis of Prop.~\ref{prop:block_automorphisms}: the \emph{block-swapping} family (those involving $\sigma$) and the \emph{block-preserving} family (those not involving $\sigma$). As the criterion of Theorem~7 of Ref.~\cite{Breuckmann_foldtransversal} is dependent on properties of the underlying fold, we find handling the two families separately to be the easiest choice.  We have the following:
\begin{enumerate}
    \item \textbf{$M_\leftrightarrow$, block-swapping:} $\mathrm{Stab}(p,q)\  \cup \ \mathrm{SwapInv}(p,q)$

    For each $j \in \mathrm{Stab}(p,q) \ \cup \ \mathrm{SwapInv}(p,q)$, the permutation
    \bea
    \tau_j^{(1)} \;=\; \Sigma\circ \Phi_j^{\mathrm{stab} \ \cup \ \mathrm{si}} = \sigma \circ \phi_{x^{-j}}
    \eea
    is a $M_\leftrightarrow$ automorphism/ZX-duality. Since $1 \in \mathrm{Stab}(p,q)$ always, this family is never empty. The $j = 1$ case recovers the $CZ_{\tau_0\omega}$ fold transversal gate of~\cite{automorphisms_BBcodes}\footnote{$\tau_0$ is notation from ~\cite{automorphisms_BBcodes} and should not be mistaken for $\tau_j^{(1)}$}.

    \item \textbf{$M_\leftrightarrow$, block-preserving:} $\mathrm{Swap}(p,q) \ \cup \ \mathrm{Inv}(p,q)$

    For each $j \in \mathrm{Swap}(p,q) \ \cup \ \mathrm{Inv}(p,q)$, the permutation
    \bea
    \tau_j^{(2)} = \Sigma \circ \Phi_j^{\mathrm{swap} \ \cup \ \mathrm{inv}} = \phi_{x^{-j}}
    \eea
    is a $M_\leftrightarrow$ automorphism/ZX-duality. This sends qubit $i$ to $(-ji \bmod \ell)$ within each block; no physical swap of the two halves occurs. This family exists only when at least one of Swap, Inv is nonempty. The conditions on these sets are non-trivial, and often this family is empty.
\end{enumerate}

To simplify the analysis, we opt to use an altered version of the framework introduced for $S$-type fold-transversal gates in \cite{Breuckmann_foldtransversal}. As our statement is slightly modified, we provide a proof of the theorem. 

\begin{theorem}[Slight modification of Theorem 7 of \cite{Breuckmann_foldtransversal}]\label{thm:fold_s_gate}
   Let $\tau$ be a $ZX$ duality with fixed point set $\mathrm{Fix}(\tau) = \{i \in \{1, \dots, n\} : \tau(i) = i\}$ such that
   \begin{enumerate}[label=({\bf S\arabic*})]
    \item $\tau^2 = \mathrm{id}$ 
    \item $|\mathrm{Fix}(\tau)| = 0 \bmod 2$, and for each $X$-check $X^{\otimes s}$ for $s \subset \{1, \dots, n\}$, $|s \cap \mathrm{Fix}(\tau)| = 0\bmod 2$.
\end{enumerate}
Then, there exists a Pauli operator $P$ such that 
\begin{equation}
    S_{\tau} = P \bigotimes_{i = \tau(i)}
S_i \bigotimes_{i < \tau(i)} \mathrm{CZ}_{i,\tau(i)}
\end{equation}
is a logical gate.
\end{theorem}
We note that Theorem 7 of \cite{Breuckmann_foldtransversal} further requires that for each \(X\) check \(X^{s}\), \(s\) contains an even number of two element orbits, in which case \(P = I\). 

\begin{proof}
    We will show that for every element \(M\) of the stabilizer, \(S_{\tau}MS_{\tau}^{\dagger}\) is also an element of the stabilizer. First, observe that if \(M = Z^s\) for some \(s \subset \{1, \dots, n\}\), \(S_{\tau} M S_{\tau}^{\dagger} = M\). Let \(M^X_1, \dots, M^X_m\) be an independent, generating set of \(X\) check operators, with \(M^X_j = X^{\otimes s_j}\). We see that
    \begin{equation}
     \left ( \bigotimes_{i = \tau(i)}
S_i \bigotimes_{i < \tau(i)} \mathrm{CZ}_{i,\tau(i)} \right ) X^{\otimes s_j} \left( \bigotimes_{i = \tau(i)}
S_i \bigotimes_{i < \tau(i)} \mathrm{CZ}_{i,\tau(i)}\right )^{\dagger} = (-1)^{c_j} X^{\otimes s_j} Z^{\otimes \tau(s_j)}
    \end{equation}
    where \(c_j = (\frac{1}{2} |\mathrm{Fix}(\tau) \cap s_j | + w_j) \bmod 2\), with \(w_j\) being the number of two element orbits of \(\tau\) contained in \(s_j\). 

    Since \(M_1^X, \dots, M_m^X\) are independent stabilizer generators, we may extend them to a symplectic basis of the Pauli group containing the destabilizers \(D_1^Z, \dots, D_m^Z\), such that
    \[D_j^Z M_{j'}^X = (-1)^{\delta_{j, j'}} M_{j'}^X D_j^Z\]
    and such that each \(D_j^Z\) commutes with all of the \(Z\) stabilizer generators. Let
    \[P = \prod_{j: c_j = 1} D_j^Z.\]
     Then, since \(\tau\) is a \(ZX\) duality, we have that
    \[S_{\tau} X^{\otimes s_j} S_{\tau}^{\dagger} = X^{\otimes s_j} Z^{\otimes \tau(s_j)}\]
    which is in the stabilizer, as desired.
\end{proof}

We additionally include a lemma that will become useful in the block-preserving case:

\begin{lem}\label{lem:support_overlap}
Let $\alpha(x) \in R_\ell$ and let $F \subseteq \mathbb{Z}/\ell\mathbb{Z}$ be a subset closed under negation, with indicator polynomial
$F(x) = \sum_{i \in F} x^i$.  Then
\bea
[\alpha(x) \cdot F(x)]_0 \;=\; \sum_{i \in F} \alpha_i \;=\; |\mathrm{supp}(\alpha) \cap F| \pmod{2}.
\eea
\end{lem}
\begin{proof}
Expanding the product,
\bea
[\alpha(x) \cdot F(x)]_0 = \sum_i \alpha_i \cdot \mathbf{1}[-i \in F]
\eea
Since $F$ is closed under negation, $-i \in F$ if and only if $i \in F$, giving $\sum_{i \in F} \alpha_i$.  Over $\mathbb{F}_2$ this equals $|\mathrm{supp}(\alpha) \cap F| \bmod 2$.
\end{proof}

We now evaluate these conditions for each family of $M_\leftrightarrow$ automorphisms.

\begin{prop}\label{prop:S_type_swapping}
Let $C$ be a GB code defined by $(\ell, f, p, q)$, and let 
$j \in \operatorname{Stab}(p,q) \cup \operatorname{SwapInv}(p,q)$. Then, $\tau =  \sigma \circ \phi_{x^{-j}}$ admits an S-type fold-transversal gate
\bea
S_\tau = \bigotimes_{i < \tau(i)} \mathrm{CZ}_{i,\,\tau(i)}
\eea
if and only if $j^2 \equiv 1 \pmod{\ell}$.
\end{prop}
\begin{proof}

Since $\tau = \sigma \circ \phi_{x^{-j}}$ maps every qubit in block~1 to a qubit in block~2 and vice versa, $\tau$ has no fixed points. This has two immediate consequences: condition (S2) is trivially satisfied, and the $S^{(\dagger)}$ factors in the gate $S_\tau$ are absent, leaving only CZ (plus Pauli correction) gates on two-element orbits. It remains to verify $j^2 \equiv 1 \pmod{\ell}$.

Applying $\tau$ twice cancels out $\sigma$, and we may assess the action of $\tau^2$ within a single block. In the absence of $\sigma$, $\tau$ acts as $i \mapsto -ji \bmod \ell$. Applying $\tau$ twice gives $i \mapsto -j(-ji) = j^2 i \bmod \ell$. Thus $\tau^2 = \mathrm{id}$ if and only if $j^2 \equiv 1 \pmod{\ell}$.
\end{proof}

\begin{prop}\label{prop:S_type_preserving}
Let $C$ be a GB code defined by $(\ell, f, p, q)$, and let $j \in \operatorname{Swap}(p,q) \cup \operatorname{Inv}(p,q)$. Then, $\tau = \phi_{x^{-j}}$ admits an S-type fold-transversal gate
\bea
S_\tau = \bigotimes_{i = \tau(i)} S^{(\dagger)}_i 
\;\;\bigotimes_{i < \tau(i)} \mathrm{CZ}_{i,\,\tau(i)}
\eea
if and only if the following conditions are satisfied:
\begin{enumerate}
    \item $j^2 \equiv 1 \pmod{\ell}$.
    \item $\gcd(\hat{f},\, x^{\ell/d} - 1)$ divides $(p + q)$, where $d = \gcd(1 + j,\, \ell)$
\end{enumerate}
\end{prop}
\begin{proof}
Unlike the block-swapping case, $\tau = \phi_{x^{-j}}$ acts within each block, sending $i \mapsto -ji \bmod \ell$ in both blocks independently. In particular, $\tau$ can have fixed points, and the second condition from Theorem~\ref{thm:fold_s_gate} requires verification.

\medskip
\noindent\textbf{(S1).}\quad
Since $\tau$ acts blockwise, $\tau^2 = \phi_{x^{j^2}}$. This follows immediately from the proof of Prop.~\ref{prop:S_type_swapping}, and is the identity if and only if $j^2 \equiv 1 \pmod{\ell}$.

\medskip
\noindent\textbf{(S2).}  The fixed points of $\tau$ within a single block satisfy $-ji \equiv i \pmod{\ell}$, i.e., $(1 + j)i \equiv 0 \pmod{\ell}$.  Let
\bea
F = \{i \in \mathbb{Z}/\ell\mathbb{Z} : (1 + j)i \equiv 0 \pmod{\ell}\}
\eea
$F$ is a subgroup of size $d = \gcd(1 + j, \ell)$ with elements all multiples of $\ell/d \pmod{\ell}$.  Since $\tau$ acts identically on
both blocks, $\mathrm{Fix}(\tau) = F \cup (F + \ell)$ and
$|\mathrm{Fix}(\tau)| = 2d$.

Condition (S2) requires that each X-check has even overlap with $\mathrm{Fix}(\tau)$.  An X-check $(pa, qa)$ has support $\mathrm{supp}(pa)$ in block~1 and $\mathrm{supp}(qa)$ in block~2, so the overlap is
\bea
|\mathrm{supp}(pa) \cap F| \;+\; |\mathrm{supp}(qa) \cap F|.
\eea
Let $F(x) = \sum_{i \in F} x^i \in R_\ell$ be the polynomial representation of $F$. Since $F$ is a subgroup of $\mathbb{Z}/\ell\mathbb{Z}$ it is closed under negation, so $F(x^{-1}) = F(x)$.  Using Lemma \ref{lem:support_overlap}, the overlap becomes
\bea
[p(x)a(x) \cdot F(x)]_0 + [q(x)\,a(x) \cdot F(x)]_0 = [(p(x) + q(x))\, a(x) \cdot F(x)]_0.
\eea
As we are working over $\bbF_2$, for (S2) to be satisfied this must be 0 for every $a \in \langle f \rangle$. As $\langle f \rangle = \{rf:r \in R_\ell \}$, we have
\bea
[rf(p+q)F]_0 = 0 \in \bbF_2
\eea
Taking $r = x^k$ to cyclically shift $f$, we have that every coefficient of $f(p+q)F$ must be 0, and so
\bea
(p + q) \cdot F(x) \in \langle \hat{f} \rangle
\eea
which is satisfied exactly when
\bea
\hat{f} \mid (p + q)F(x)
\eea
We have that 
\bea
F(x) = \sum_{k=0}^{d-1} x^{k\ell/d} = \frac{x^\ell - 1}{x^{\ell/d} - 1}.
\eea
where recall $d = \gcd(j+1, \ell)$. Both $F(x)$ and $\hat{f}$ divide $x^\ell - 1$, so we may analyze the divisibility condition in $\overline{\mathbb{F}}_2$, the algebraic closure of $\bbF_2$. The roots of $\hat{f}$ are the $\ell$-th roots of unity that are not roots of $f$, while the roots of $F(x)$ are the $\ell$-th roots of unity that are not roots of $x^{\ell/d} - 1$, i.e., those whose order does not divide $\ell/d$. Thus, a root $\zeta$ of $\hat{f}$ fails to be a root of $F$ precisely when $\zeta^{\ell/d} = 1$, i.e., when its order divides $\ell/d$.

For $\hat{f} \mid (p+q)F(x)$, every root $\zeta$ of $\hat{f}$ must be a root of either $F$ or $p+q$. Therefore, for every root $\zeta$ of $\hat{f}$ such that $\zeta^{\ell/d} = 1$, we must have
\bea
(p+q)(\zeta) = 0,
\eea
i.e., $q(\zeta) = p(\zeta)$. These are precisely the roots of
\bea
g = \gcd(\hat{f},\, x^{\ell/d} - 1),
\eea
and the condition that $p+q$ vanishes on all such roots is equivalent to
\bea
g \mid (p+q).
\eea
Thus, (S2) holds for a given $j$ when $\gcd(\hat{f}, x^{\ell/d}-1)$ divides $(p + q)$, where $d = \gcd(1+j, \ell)$.

\end{proof}

\subsubsection{CX-type gates}
\label{sec:cx_type}

The $S$-type analysis of Sec.~\ref{sec:phase_type} relied on an existing framework. No analogous framework currently exists for fold-transversal CNOTs, and we do not attempt to develop a comprehensive framework here.  Instead, we ask a more focused question: \emph{when can we guarantee that a fold-transversal CNOT, composed with a multiplier ring automorphism $\psi \oplus \psi$, extends to a code automorphism?}  Theorem~\ref{thm:fold_cx} below answers this with \emph{sufficient} (but not necessary) algebraic conditions on $(\ell, f, p, q, \psi)$.  GB codes that fail these conditions may still admit fold-transversal CNOTs by some other construction; we make no claim of completeness, only of constructive utility, and the conditions are concrete enough to be searched over directly when designing codes.

The underlying idea is the following.  When accounting for the Pauli type assigned to each rowspace, a physical CNOT along the $(i, i+\ell)$ fold acts on, for example, a rowspace element $(pa, qa) \in \operatorname{rs}(H_X)$ by adding one block into the other,
\bea\label{eq:cnot_action}
(pa, qa) \xrightarrow{\mathrm{CNOT}_{2 \to 1}} ((p + q)a,\, qa), \qquad (pa, qa) \xrightarrow{\mathrm{CNOT}_{1 \to 2}} (pa,\, (p + q)a),
\eea
and the image is generally \emph{not} itself a rowspace element: the modified coordinate $(p+q)a$ need not lie in $p \cdot \langle f \rangle$ (or $q \cdot \langle f \rangle$), so the cyclic submodule structure of Theorem~\ref{thm:cyclic_submodule} is destroyed.  Following the CNOT by a multiplier ring automorphism $\psi \oplus \psi$ gives us one degree of algebraic freedom with which to repair the structure.  Concretely, we ask $\psi$ to be chosen so that the composed image lands back in either $M_f(p,q)$ (a $M_\sim$ outcome) or $M_{\overleftarrow{f}}(\overleftarrow{q}, \overleftarrow{p})$ (a $M_\leftrightarrow$ outcome), i.e., so that the CNOT $+ \psi$ composition is itself a code automorphism.  The required identity on $\psi, p, q$ takes a M\"obius-like form
\bea
\psi(p) \;\equiv\; \frac{p+q}{q}\,\psi(q) \quad\text{(or one of three variants)}
\eea
in the appropriate quotient ring, with $\psi(f) = u\cdot f$ (or $\psi(f) = u \cdot \overleftarrow{f}$) with $u \in R_\ell^\times$  so $\langle \psi(f) \rangle = \langle f \rangle$ (or $\langle \psi(f) \rangle = \langle \overleftarrow{f} \rangle)$ and $\psi(a) \in \langle f \rangle$ (or $\langle \overleftarrow{f}\rangle$) and the cyclic submodule structure over $f$ (or $\overleftarrow{f}$) is preserved. 

This produces the four variants of Theorem~\ref{thm:fold_cx} (illustrated in Fig.~\ref{fig:thm38_mobius_fixed}), indexed independently by the CNOT direction and the rowspace behavior of the resulting automorphism.  The two CNOT directions --- block~$2 \to$ block~$1$ (controlled on $i + \ell$, targeting $i$) and block~$1 \to$ block~$2$ (controlled on $i$, targeting $i + \ell$) --- differ by which block absorbs the addition in~\eqref{eq:cnot_action}, and consequently swap the roles of $p$ and $q$ in the M\"obius identity.  Within each direction, $M_\sim$ outcomes yield conditions in $S = R_\ell / \langle \hat{f} \rangle$, while $M_\leftrightarrow$ outcomes yield conditions in $\overleftarrow{S} = R_\ell / \langle \hat{f}(x^{-1}) \rangle$ involving transfer polynomials $\overleftarrow{p}, \overleftarrow{q}$.

Two technical points are worth flagging in advance.  First, each variant requires one of $p$ or $q$ to be invertible modulo the relevant annihilator (the $\gcd(p,\hat{f}) = 1$ or $\gcd(q,\hat{f}) = 1$ hypotheses in Theorem~\ref{thm:fold_cx} below); without this invertibility, the M\"obius-like solution for $a'$ does not exist, and the whole construction fails to be well-defined.  This is the central reason the conditions are merely sufficient: the construction is built around an explicit ansatz $a' = q^{-1}\psi(qa)$ (or its analog), and codes whose $(p, q)$ make this ansatz unavailable may still admit fold-transversal CNOTs through entirely different constructions.  Second, as in the substitution-multiplier setting (Theorem~\ref{thm:substitution_automorphisms}) and the $S$-type analysis, as we restrict to $\psi_L = \psi_R$ the $H_Z$ verification is the image of the $H_X$ verification under the ring automorphism $x \mapsto x^{-1}$, so each of the four cases reduces to checking a single algebraic identity.

\begin{theorem}\label{thm:fold_cx}
Let $C$ be a GB code defined by $(\ell, f, p, q)$ as in Theorem~\ref{thm:cyclic_submodule}. Let $\psi_{x^j}$ be any multiplier ring automorphism of $R_\ell$, and $u \in R_\ell^\times$ an arbitrary invertible element. Let $S = R_\ell/\langle \hat{f} \rangle, \ \overleftarrow{S} = R_\ell/\langle \hat{f}(x^{-1})\rangle$. Then, if $p, q, \psi_{x^j}$ satisfy either of the following conditions in the indicated quotient ring, transversal CNOTs along the $(i, i +\ell)$ fold controlled on block 2 ($i + \ell$), targeting block 1 ($i$) followed by $\psi_{x^j} \oplus \psi_{x^j}$ yield a code automorphism.
\begin{enumerate}[label=(\arabic*)]
    \item $\gcd(q, \hat{f}) = 1,\  \psi_{x^j}(f) =  uf,\ $ and $\ \psi_{x^j}(p) \equiv \frac{p+q}{q}\psi_{x^j}(q)$ in $S$ ($M_\sim$)
    
    \item $\gcd(p, \hat{f})= 1, \ \psi_{x^j}(f) = u\overleftarrow{f} $ and $ \ \psi_{x^j}(p) \equiv \frac{\overleftarrow{p} + \overleftarrow{q}}{\overleftarrow{p}}\psi_{x^j}(q) $ in $\overleftarrow{S}$ ($M_\leftrightarrow$)
\end{enumerate}
Alternatively, if $p, q, \psi_{x^j}$ satisfy either of the following conditions in the indicated quotient ring, then transversal CNOTs along the $(i, i +\ell)$ fold controlled on block 1 ($i$), targeting block 2 ($i + \ell$) followed by $\psi_{x^j} \oplus \psi_{x^j}$ yields a code automorphism.
\begin{enumerate}[label=(\arabic*), resume]
    \item $\gcd(p, \hat{f}) = 1, \ \psi_{x^j}(f) = uf$ and $ \ \psi_{x^j}(q) \equiv \frac{p+q}{p}\psi_{x^j}(p) $ in $S$. ($M_\sim$)
    
    \item $\gcd(q, \hat{f}) = 1, \ \psi_{x^j}(f) = u\overleftarrow{f} $ and $ \ \psi_{x^j}(q) \equiv \frac{\overleftarrow{p}+\overleftarrow{q}}{\overleftarrow{q}}\psi_{x^j}(p) $ in $\overleftarrow{S}$. ($M_\leftrightarrow$)
\end{enumerate}
The conditions above are stated as ratio identities in $S$ (resp. $\overleftarrow{S}$) and are therefore invariant under the rescaling $(p,q) \mapsto v(p,q), \ v \in S^\times$ under which $M_f(p,q)$ is unchanged (Lemma~\ref{lem:equivalent_modules}).
\end{theorem}

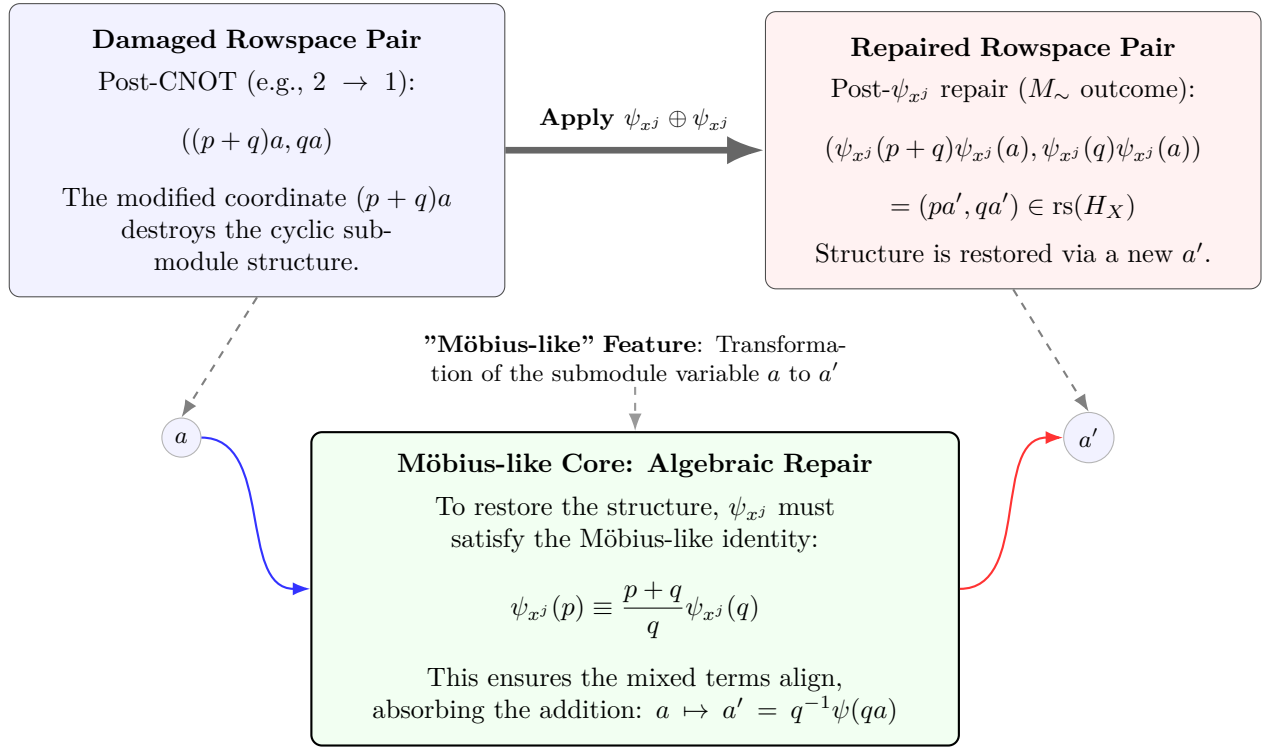
\begin{figure}[htbp]
\centering
\begin{tikzpicture}[
    >=Latex,
    labeltext/.style={align=center, font=\small},
    casebox/.style={
        rectangle, draw=black!70, rounded corners,
        align=center, inner sep=10pt, text width=5.85cm, fill=white
    },
    titlebox/.style={
        rectangle, draw=black, thick, rounded corners, fill=black!5,
        align=center, inner sep=8pt, text width=16cm
    },
    invertbox/.style={
        rectangle, draw=black, thick, rounded corners, fill=green!5,
        align=center, inner sep=8pt, text width=8.0cm
    },
    varnode/.style={
        circle, draw=gray!60, fill=blue!5, inner sep=3pt, font=\small
    }
]

\node[titlebox] (Title) at (0, 8.5) {
    \textbf{Theorem \ref{thm:fold_cx}: Möbius-Like Mechanism for Fold-Transversal CNOTs} \\[0.6ex]
    A physical CNOT adds one block to the other, breaking the cyclic submodule structure. \\
    Applying a multiplier $\psi_{x^j} \oplus \psi_{x^j}$ repairs this if $\psi_{x^j}$ satisfies a Möbius-like identity. \\
    Pre-requisites: $q^{-1}$ exists in $S = R_\ell/\langle \hat{f} \rangle$ and $\psi(a) \in \langle f \rangle$\\
};

\node[casebox, fill=blue!5] (Input) at (-5, 5) {
    \textbf{Damaged Rowspace Pair} \\[0.8ex]
    Post-CNOT (e.g., $2 \to 1$):
    \[ ((p+q)a, qa) \]
    The modified coordinate $(p+q)a$ \\
    destroys the cyclic submodule structure.
};

\node[casebox, fill=red!5] (Output) at (5, 5) {
    \textbf{Repaired Rowspace Pair} \\[0.8ex]
    Post-$\psi_{x^j}$ repair ($M_\sim$ outcome):
    \[ (\psi_{x^j}(p+q)\psi_{x^j}(a), \psi_{x^j}(q)\psi_{x^j}(a)) \]
    \[ = (p a', q a') \in \operatorname{rs}(H_X) \]
    Structure is restored via a new $a'$.
};

\draw[->, ultra thick, draw=black!60, line width=2.5pt] (Input.east) --
    node[labeltext, above=0.1cm] {\textbf{Apply $\psi_{x^j} \oplus \psi_{x^j}$}}
    (Output.west);

\node[invertbox] (Inversion) at (0, -0.8) {
    \textbf{Möbius-like Core:} \textbf{Algebraic Repair} \\[1ex]
    To restore the structure, $\psi_{x^j}$ must \\
    satisfy the Möbius-like identity:
    \[ \psi_{x^j}(p) \equiv \frac{p+q}{q}\psi_{x^j}(q) \]
    This ensures the mixed terms align, \\
    absorbing the addition: $a \mapsto a' = q^{-1}\psi(qa)$
};

\node[labeltext, text width=10cm, fill=white, inner sep=2pt] (MobiusLabel) at (0, 2.25) {
    \textbf{"Möbius-like" Feature}: 
    Transformation of the 
    submodule variable $a$ to $a'$
};

\node[varnode] (InVar) at (-6, 1.2) {$a$};
\node[varnode] (OutVar) at (6, 1.2) {$a'$};

\draw[->, thick, dashed, draw=black!50] (Input.south) -- (InVar.north);
\draw[->, thick, dashed, draw=black!50] (Output.south) -- (OutVar.north);

\draw[->, thick, draw=blue!80] (InVar.east) to[out=0, in=180] (Inversion.west);
\draw[->, thick, draw=red!80] (Inversion.east) to[out=0, in=180] (OutVar.west);

\draw[->, thick, dashed, draw=gray!80] (MobiusLabel.south) -- (Inversion.north);

\end{tikzpicture}
\caption{\textbf{Visualization of Theorem \ref{thm:fold_cx}: Fold-Transversal CNOTs and Möbius-like Repair} ($\mathrm{CNOT}_{2\rightarrow1}$, $M_\sim$ case) A physical CNOT disrupts the cyclic submodule structure by modifying one of the coordinates out of the target ideal. Following the CNOT with a ring multiplier permutation $\psi_{x^j} \oplus \psi_{x^j}$ can restore the structure (mapping the pair back into $\operatorname{rs}(H_X)$ or $\operatorname{rs}(H_Z)$) if the transfer polynomials satisfy a specific Möbius-like identity. This identity allows the transformed mixed term to be absorbed into a new valid submodule variable $a'$. }
\label{fig:thm38_mobius_fixed}
\end{figure}

\begin{proof}

To assess the action of physical CNOTs along the $i, i +\ell$ fold, we now consider the Pauli operator assigned to the $H_X, H_Z$ rowspaces and in particular how an $H_X$ vs $H_Z$ rowspace element is conjugated under the action of a CNOT. First we observe the $2 \rightarrow 1$ ($i + \ell, i)$ control, target direction:
    \begin{align}\label{eq:foldcnot_21}
        \operatorname{rs}(H_X):(pa, qa) \xrightarrow{\mathrm{CNOT}_{2 \rightarrow 1}} (pa + qa, qa) \quad \quad \operatorname{rs}(H_Z): (\overleftarrow{qa}, \overleftarrow{pa}) \xrightarrow{\mathrm{CNOT}_{2 \rightarrow 1}}  (\overleftarrow{qa}, \overleftarrow{pa} + \overleftarrow{qa})
    \end{align}
    and second in the $1 \rightarrow 2$ $(i, i + \ell)$ control, target direction:
    \begin{align}\label{eq:foldcnot_12}
        \operatorname{rs}(H_X): (pa, qa) \xrightarrow{\mathrm{CNOT}_{1 \rightarrow 2}} (pa, qa + pa) \quad \quad \operatorname{rs}(H_Z):(\overleftarrow{qa}, \overleftarrow{pa}) \xrightarrow{\mathrm{CNOT}_{1 \rightarrow 2}} (\overleftarrow{qa} + \overleftarrow{pa}, \overleftarrow{pa})
    \end{align}
    Proving each of the 4 cases follows the exact same template.
    \begin{itemize}
        \item For the $H_X$ argument, using the appropriate form of \eqref{eq:foldcnot_21} or \eqref{eq:foldcnot_12}, isolate the single term $v \in \{\psi(qa), \psi(pa)\}$ and build an $a'$ as a function of the appropriate $ p^{-1}, q^{-1}, \psi(a), \psi(p), \psi(q)$ such that $v = wa'$ for $w \in \{p, q\}$.
        \item Derive the resulting constraint that $\psi(pa) + \psi(qa) = wa'$ for the appropriate $w$
        \item Observe that the $H_Z$ derives the exact same constraint as the $H_X$ case, with the addition of the application the $x^{-1}$ ring automorphism. Thus, if the $H_X$ case is satisfied, the $H_Z$ case follows.
    \end{itemize} 
    We prove the first case explicitly, and then identify the appropriate $a'$ choices for the remaining conditions. 
    
    For the $H_X$ rowspace to be mapped back to itself upon the application of $\mathrm{CNOT}_{2\rightarrow 1}$ followed by $\psi$, from \eqref{eq:foldcnot_21} we must have
    \begin{align}
    \psi(pa + qa) = pa' \quad \quad \psi(qa) = qa'
    \end{align}
    Observe that as we require $\psi(f) = uf$ for $u \in R_\ell^\times$, we have that $\psi(a) \in \langle f \rangle$, and we can assign $a' = q^{-1}\psi(qa) \in \langle f \rangle$ and obtain\footnote{Recall, as a ring automorphism, $\psi(ab) = \psi(a)\psi(b)$.}
    \bea
        \psi(qa) = qq^{-1}\psi(qa) = qa' 
    \eea
    It then follows that we must have
    \bea\label{eq:cx_mobius_eq}
        \psi(pa + qa) = (\psi(p) + \psi(q))\psi(a) = pq^{-1}\psi(qa) = pa'
    \eea
    For this to be true the following identity must hold:
    \bea
    \psi(p) + \psi(q) \equiv pq^{-1}\psi(q) \pmod{\langle \hat{f} \rangle} \implies \psi(p) \equiv \frac{p + q}{q}\psi(q) \pmod{\langle \hat{f} \rangle}
    \eea
    where we cancel $\psi(a)$ on each side of \ref{eq:cx_mobius_eq}, implying we now work mod $\hat{f}$, i.e., in $S$. For the $H_Z$ rowspace, we analogously must have
    \bea
    \psi(\overleftarrow{qa}) = \overleftarrow{qa'} \quad \quad \psi(\overleftarrow{pa} + \overleftarrow{qa}) = \overleftarrow{pa'}
    \eea
    Choosing the same $a'=q^{-1}\psi(qa)$ as above, we obtain $\psi(\overleftarrow{qa}) = \overleftarrow{qa'}$, yielding
    \begin{align}
    \left[\psi(\overleftarrow{p}) + \psi(\overleftarrow{q})\right]\psi(\overleftarrow{a}) = \overleftarrow{pq^{-1}\psi(qa)}
    \end{align}
    which is satisfied when
    \begin{align}
    \psi(\overleftarrow{p}) = \frac{\overleftarrow{p} + \overleftarrow{q}}{\overleftarrow{q}}\psi(\overleftarrow{q}) \pmod {\langle \hat{f}(x^{-1}) \rangle}
    \end{align}
    Note that similar to Theorem~\ref{thm:substitution_automorphisms}, $\psi(p), \psi(q)$ are invariant up to scaling both by the same unit $v \in S^\times$. Observe that the $H_Z$ condition is the image of the $H_X$ condition under the $x^{-1}$ ring automorphism, and so if the $H_X$ condition is satisfied, the $H_Z$ condition is automatically satisfied.

    The remaining cases are satisfied via the following assignments, under the assumptions that $\psi(a)$ lives in the appropriate ideal as in the theorem statement, and $p^{-1}$ or $q^{-1}$ exists in the appropriate quotient ring $S$ or $\overleftarrow{S}$:
    \begin{itemize}
        \item $CNOT_{2 \rightarrow 1}, \ M_\leftrightarrow$: Assign $a' = p^{-1}\overleftarrow{\psi(qa)} \in \langle f \rangle$, yields the requirement that 
        \bea
        \psi(p) \equiv \frac{\overleftarrow{p} + \overleftarrow{q}}{\overleftarrow{p}}\psi(q) \pmod{\langle \hat{f}(x^{-1}) \rangle}
        \eea

        \item $CNOT_{1 \rightarrow 2}, \ M_\sim$: Assign $a' = p^{-1}\psi(pa) \in \langle f \rangle$, yields the requirement that 
        \bea
        \psi(q) \equiv \frac{p + q}{p}\psi(p) \pmod{\langle \hat{f} \rangle}
        \eea

        \item $CNOT_{1 \rightarrow 2}, \ M_\leftrightarrow$: Assign $a' = q^{-1}\overleftarrow{\psi(pa)} \in \langle f \rangle$, yields the requirement that 
        \bea
        \psi(q) \equiv \frac{\overleftarrow{p} + \overleftarrow{q}}{\overleftarrow{q}}\psi(p) \pmod{\langle \hat{f}(x^{-1}) \rangle}
        \eea
    \end{itemize}

\end{proof}
Be careful to note the seemingly swapped inclusions of $\psi(a) \in \langle \overleftarrow{f} \rangle$ and $\overleftarrow{\psi(qa)}, \overleftarrow{\psi(pa)} \in \langle f \rangle$ in the $M_\leftrightarrow$ cases. This is intentional and required, and one has to take care in expanding these proofs to ensure the correct elements land in the correct spaces. Further, note that the condition $\gcd(q, \hat{f}) = 1$ implies that $\gcd(\overleftarrow{q}, \hat{f}(x^{-1})) = 1$ as $\psi_{x^{-1}}$ is a ring automorphism. The same holds for $p$, and the requisite conditions for (2) and (4) are met.

We summarize the conditions on $S$- and CX-type fold-transversal gates in Table~\ref{tab:fold_gates}.

\begin{table*}[tb]
\centering
\renewcommand{\arraystretch}{1.8}
\small
\begin{tabular}{p{1.8cm} p{2.2cm} p{2.8cm} p{6.2cm}}
\hline
Gate & Source & Construction & Conditions \\
\hline

S-Type
& $j \in \mathrm{Stab} \cup \mathrm{SwapInv}$
& $\bigotimes \mathrm{CZ}_{i,\tau(i)}$, \newline
  $\tau = \sigma \circ \phi_{x^{-j}}$
& $\bullet\ j^2 \equiv 1 \pmod{\ell}$ \\[4pt]

S-Type
& $j \in \mathrm{Swap} \cup \mathrm{Inv}$
& $S/S^\dagger$ on $\mathrm{Fix}(\tau)$, \newline
  $\mathrm{CZ}$ on 2-orbits, \newline
  $\tau = \phi_{x^{-j}}$
& $\bullet\ j^2 \equiv 1 \pmod{\ell}$ \newline
  $\bullet\ \gcd(\hat{f},\, x^{\ell/d}-1) \mid (p+q)$,\newline \phantom{$\bullet$} $d = \gcd(1+j, \ell)$ \\[4pt]
\hline

CX-Type \newline ($M_\sim$)
& $j \in \mathrm{Pres}(f)$
& $\phi_{x^j} \circ \mathrm{CX}_{i + \ell \to i}$
& $\gcd(q, \hat{f}) = 1$ \newline $p(x^j) \equiv \dfrac{p+q}{q}\, q(x^j)$ in $S$ \\[4pt]
CX-Type \newline ($M_\leftrightarrow$)
& $j \in \mathrm{Pres}(\overleftarrow{f})$
& $\phi_{x^j} \circ \mathrm{CX}_{i + \ell \to i}$
& $\gcd(p, \hat{f}) = 1$ \newline $p(x^j) \equiv \dfrac{\overleftarrow{p}+\overleftarrow{q}}{\overleftarrow{p}}\, q(x^j) $ in $\overleftarrow{S}$ \\[4pt]
CX-Type \newline ($M_\sim$)
& $j \in \mathrm{Pres}(f)$
& $\phi_{x^j} \circ \mathrm{CX}_{i \to i + \ell}$
& $\gcd(p, \hat{f}) = 1$ \newline $q(x^j) \equiv \dfrac{p+q}{p}\, p(x^j)$ in $S$ \\[4pt]
CX-Type \newline ($M_\leftrightarrow$)
& $j \in \mathrm{Pres}(\overleftarrow{f})$
& $\phi_{x^j} \circ \mathrm{CX}_{i \to i + \ell}$
& $\gcd(q, \hat{f}) = 1$ \newline $q(x^j) \equiv \dfrac{\overleftarrow{p}+\overleftarrow{q}}{\overleftarrow{q}}\, p(x^j)$ in $\overleftarrow{S}$ \\[4pt]

\hline
\end{tabular}
\caption{Conditions for S and CX type fold-transversal gates. All parities are mod 2 and all polynomial identities are up to units in the indicated quotient ring. Here $\hat{f} = (x^\ell - 1)/f$, and $\phi_{x^j} = \psi_{x^j} \oplus \psi_{x^j}$. Pres$(\overleftarrow{f})$ is the analogue of Pres$(f)$, all $j$ such that $\langle f(x^j) \rangle = \langle \overleftarrow{f} \rangle$. S-type conditions follow from Theorem 7 of \cite{Breuckmann_foldtransversal}, altered as in Theorem~\ref{thm:fold_s_gate}. CX-type conditions are sufficient (Theorem~\ref{thm:fold_cx}) but not necessary. All results hold over $\mathbb{F}_2$ only --- analysis for $\bbF_q$ is left as future work.}
\label{tab:fold_gates}
\end{table*}

\noindent
The conditions presented in Tables~\ref{tab:block_sep_auts} and \ref{tab:fold_gates} enable the main operational result of this paper, which is an efficient search algorithm to find GB codes with desired automorphism-based gates, as summarized in Fig.~\ref{fig:search_alg}.

\input{figures/search_alg_fig.tex}

\section{Logical Operators and Actions of Automorphisms}
\label{sec:LogicalOps}

Sec.~\ref{sec:3SpaceIso} and \ref{sec:AutomorphismStructure} developed the algebraic framework that underlies the main result of this work: a systematic methodology for determining whether or not a given GB code has automorphisms beyond the trivial cyclic shifts, with a full treatment of the $\psi_L = \psi_R$ case. However, Sec.~\ref{sec:AutomorphismStructure} remains completely independent of the logical action that any code automorphism implements. What remains is to pass from the existence of these symmetries to their logical action: when a given automorphism or fold-transversal gate preserves the cyclic submodule structure, how can we determine what operation it induces on the encoded qubits?

Assessing this question requires a concrete set of logical operator representatives. For a CSS code, the logical $X$- and $Z$-operators are defined as kernel elements modulo the appropriate stabilizer rowspaces. In the present setting however, those kernels inherit additional structure from the cyclic-submodule description of Theorem~\ref{thm:cyclic_submodule}, and this makes it possible to analyze them algebraically rather than by direct linear algebra on large matrices. The central goal of this section is therefore twofold: first, to construct canonical and computationally useful representatives for the logical operators of a GB code; and second, to use those representatives to assess how automorphisms map the canonical representative set to determine the logical action of the automorphisms identified in Sec.~\ref{sec:AutomorphismStructure}.

The section is organized as follows. Following a characterization of the $H_X, H_Z$ kernels in Sec.~\ref{subsec:lop_prelim}, in Sec.~\ref{subsec:lop_crt} we describe the logical operator spaces using the Chinese Remainder Theorem decomposition of $R_\ell$ and the induced decomposition of the kernel and stabilizer quotients into local components. This yields a componentwise picture of the logical space that we then use to construct explicit logical representatives. Sec.~\ref{subsubsec:odd_ell} discusses a particularly useful simplification of the logical operators when $\gcd(f, \hat{f}) = 1$. With these representatives in hand, Sec.~\ref{subsec:lop_automorph} returns to the automorphisms catalogued in Sec.~\ref{sec:AutomorphismStructure} and discusses their induced action on the logical space. The main point is that, once a basis of logical representatives has been fixed, each code automorphism determines an $\bbF_2$-linear transformation on the encoded qubits, and the CRT description provides a natural setting in which that action can be computed explicitly. This turns the algebraic classification of automorphisms from Sec.~\ref{sec:AutomorphismStructure} into a practical method for determining the logical gate implemented by a given permutation symmetry of the code.

The supplementary materials of \url{https://github.com/ajdav136/GBAutomorphisms} provide a script for building the logical operators of a GB code given a valid $(\ell, f, p, q)$ set, determining automorphisms and fold-transversal gates, and assessing the logical action of each automorphism. We note that the methods and results of this section are consistent with the algebraic scaffolding set up in \cite{automorphisms_BBcodes}

\subsection{GB code kernel characterization}
\label{subsec:lop_prelim}

As a CSS stabilizer code, the logical $X$-operators are elements of $\ker H_Z$ that do not themselves lie in $\operatorname{rs}(H_X)$, and analogously for logical $Z$. We therefore begin by characterizing each kernel within the framework of the cyclic submodule characterization in Theorem \ref{thm:cyclic_submodule}. We start with a helpful lemma:
\begin{lem}\label{lem:dotproduct_zerocoeff}
    The inner product of two vectors $\langle \alpha, \beta \rangle$ over $\bbF_2$ is equal to $[\overleftarrow{\alpha}\beta]_0$, i.e., the zeroth coefficient of the polynomial product of the $R_\ell$ representations where one polynomial is evaluated at $x^{-1}$.
\end{lem}
\begin{proof}
Let $\alpha = \sum_{i=0}^{\ell-1} \alpha_i x^i$ and $\beta = \sum_{j=0}^{\ell-1} \beta_j x^j$ in $R_\ell$. In $R_\ell$, we have $x^{-i} = x^{\ell - i}$, so $\overleftarrow{\alpha} = \alpha(x^{-1}) = \sum_{i=0}^{\ell-1} \alpha_i x^{\ell - i}$. Then
\bea
\overleftarrow{\alpha} \cdot \beta = \sum_{i,j} \alpha_i \beta_j \, x^{\ell - i + j}
\eea
and the $0$-th coefficient collects all terms with $\ell - i + j \equiv 0 \pmod{\ell}$, i.e., $j \equiv i \pmod{\ell}$. Since $0 \leq i, j \leq \ell - 1$, this forces $j = i$, giving
\bea
[\overleftarrow{\alpha}\,\beta]_0 = \sum_{i=0}^{\ell-1} \alpha_i \beta_i = \langle \alpha, \beta \rangle.
\eea
\end{proof}
Note that as we are working over a commutative ring, we have $[\overleftarrow{\alpha}\beta]_0 = [\overleftarrow{\beta}\alpha]_0 = [\beta\overleftarrow{\alpha}]_0 = [\alpha\overleftarrow{\beta}]_0$.

\begin{lem}\label{lem:kernels-characterization}
With notation as above,
\begin{equation}\label{eq:kernels}
\begin{aligned}
\ker H_Z &= \{(u, w) \in R_\ell^2 \mid q u + p w \in \langle \hat{f} \rangle\},\\
\ker H_X &= \{(u, w) \in R_\ell^2 \mid p\overleftarrow{u} + q\overleftarrow{w} \in \langle \hat{f} \rangle\}
\end{aligned}
\end{equation}

where products are taken in $R_\ell$.
\end{lem}
 Be very careful to note that $p\overleftarrow{u}$ (and $q\overleftarrow{w}$) is \textit{not} a typographical error: $p$ here is not evaluated at $x^{-1}$ while $u$ is, and as such the correct expression is $p(x) \cdot u(x^{-1}) = p\overleftarrow{u}$.
\begin{proof}
We show the $H_Z$ kernel case --- the argument for $H_X$ is analogous. 

A vector $v\in \bbF_2^{2\ell}$ lies in $\ker H_Z$ if and only if it is orthogonal to every row of $H_Z$. That is, $\langle v , (\overleftarrow{qa}, \overleftarrow{pa})\rangle = 0 \ \forall a \in \langle f \rangle$. Let $v = (u, w)$, and we have $\langle u , \overleftarrow{qa}\rangle + \langle w, \overleftarrow{pa}\rangle$. Using Lemma \ref{lem:dotproduct_zerocoeff}, we have the following:
\begin{align}
    \langle (u,w), (\overleftarrow{qa}, \overleftarrow{pa}) \rangle &= \langle u, \overleftarrow{qa}\rangle +     \langle w, \overleftarrow{pa} \rangle \\
    &= [u\cdot q\cdot a]_0 + [w \cdot p\cdot a]_0 \\
    &= [a(qu + pw)]_0 = 0
\end{align}
where $\psi_{x^{-1}}(\overleftarrow{qa}) = qa, \ \psi_{x^{-1}}(\overleftarrow{pa}) = pa$ under 
the $x \rightarrow x^{-1}$ ring automorphism, as $\psi_{x^{-1}}$ is an involution. 

Let $s = qu + pw$. We need $[as]_0 = 0$ for all $a \in \langle f \rangle$. 
Since $\langle f \rangle$ is an ideal of $R_\ell$, for any $a \in \langle f \rangle$ 
we also have $x^i a \in \langle f \rangle$ for all $i$, and so
\bea
[x^i a \cdot s]_0 = [as]_{\ell - i \bmod \ell} = 0 \quad \forall\, i \in \{0, \dots, \ell-1\}.
\eea
Ranging $i$ over $0, \dots, \ell-1$, every coefficient of $as$ vanishes, 
hence $as = 0$ in $R_\ell$ for all $a \in \langle f \rangle$. That is, $s = qu + pw \in \langle \hat{f} \rangle$.
\end{proof}

Note that the definition of $\ker H_X $ can equivalently be expressed as 
\bea
\ker H_X = \{(u, w) \in R_\ell^2 \mid \overleftarrow{p}u + \overleftarrow{q}w \in \langle \hat{f}(x^{-1}) \rangle\}
\eea
by applying the $x \rightarrow x^{-1}$ map to $(pa, qa)$ as opposed to $(u, w)$ when extracting the $0$-th coefficient. We stick with the expression in equation $\eqref{eq:kernels}$ for consistency.

With the kernels in hand, we can identify explicit logical representatives. 

\subsection{Logical Operators via the CRT decomposition}
\label{subsec:lop_crt}

The kernel description in Lemma~\ref{lem:kernels-characterization} expresses $\ker H_Z$ and $\ker H_X$ globally, as conditions on pairs $(u,w) \in R_\ell^2$. To obtain an algebraically tractable picture of the logical operator space, we now decompose these kernels, and the corresponding stabilizer images, component-by-component using the Chinese Remainder Theorem. This achieves three simplifications at once: the kernel and the stabilizer image split as direct products of local pieces at each component; within each local component $i$, divisibility, ideals, and dimensions are all governed by the $f_i$-valuations $\zeta_i = v_i(f),\ \rho_i = v_i(p),\ \sigma_i = v_i(q)$ of the defining polynomials; and the global logical quotient $\mathcal{L}_X$ is reassembled from these local quotients by lifting through canonical CRT indicators $h_i$. The end result is a closed-form basis for $\mathcal{L}_X$, presented componentwise in terms of $\ell, f, p, q$.

The principal output of this subsection is Theorem~\ref{thm:explicit-basis}, which exhibits, at each CRT component~$i$, a basis of $2\zeta_i d_i$ logical representatives organized into two families: a \emph{single-slot} family, supported (non-zero) in only one of the two blocks of $R_\ell^2$ with no dependence on $p$ or $q$, and a \emph{two-slot} family, in which the kernel relation forces a coupling between the two blocks and the transfer polynomials $p, q$ enter explicitly. Summing the contributions over all components recovers the dimension count $k = 2\deg(f)$ established in Sec.~\ref{sec:3SpaceIso}, confirming that this construction is complete.

We work throughout with $X$-type logicals, as they admit the cleanest algebraic description; the $Z$-logical space is isomorphic via the coordinate change $u \mapsto u(x^{-1}) = \overleftarrow{u}$, $w \mapsto w(x^{-1}) = \overleftarrow{w}$ and is recovered in Sec.~\ref{subsubsec:Ztype} by transporting the $X$-type result through this isomorphism.

\subsubsection{The CRT decomposition and local-ring setup}
\label{subsubsec:crt_setup}

Write $\ell = 2^s m$ with $m$ odd and $s \geq 0$. When $\ell$ is odd, $s = 0$. By the Frobenius automorphism over $\bbF_2$,
\bea
    x^\ell - 1 = (x^m - 1)^{2^s} = \prod_{i=1}^{r} f_i^{2^s}
\eea
where $x^m - 1 = \prod_{i=1}^r f_i$ is the squarefree factorization into distinct irreducibles\footnote{Note that this is saying that when $\ell$ is even, every irreducible factor of $x^m-1$ is raised to the same power, $2^s$, to recover $x^\ell-1$ --- there are no ``asymmetric'' powers in the decomposition of $x^\ell-1$.}. By the Chinese Remainder Theorem,
\bea
    R_\ell \;\cong\; \prod_{i=1}^{r} L_i,
    \qquad
    L_i = \mathbb{F}_2[x]\big/\!\langle f_i^{2^s}\rangle
\eea
Let $d_i = \deg(f_i)$. Each $L_i$ is a vector space over $\bbF_2$ of dimension $\deg(f_i^{2^s}) = d_i \cdot 2^s$, and has a maximal ideal $\langle f_i \rangle$, where we are working $\pmod{f_i^{2^s}}$ and so $f_i^{2^s} \equiv 0$. $L_i$ is a chain ring, as its ideals form a totally ordered chain:
\bea
    0= \langle f_i^{2^s}\rangle \subset \langle f_i^{2^s - 1}\rangle \subset \cdots \subset \langle f_i^1\rangle \subset \langle f_i^0\rangle = L_i .
\eea
By the third isomorphism theorem, we have
\bea
L_i / \langle f_i \rangle = \left(\bbF_2[x]/\langle f_i^{2^s} \rangle \right)/ \left( \langle f_i \rangle / \langle f_i^{2^s} \rangle \right) \cong \bbF_2[x]/\langle f_i  \rangle \cong  \bbF_{2^{d_i}}
\eea

The study of cyclic codes over finite chain rings is well studied in the classical literature \cite{martinez2017codesaffinealgebrasfinite, norton2000structure, cyclic_negacyclic_chainrings}. In particular, while the standard monomial basis for a degree $d_i2^s$ polynomial space would be $\{1, x, x^2, \dots, x^{d_i\cdot 2^s - 1} \}$, it is customary and advantageous in the case of finite chain rings to use an alternative basis for $L_i$ via $\{x^zf_i^t \}$ ranging over $z, t$ as follows:
\begin{equation}\label{eq:local-basis}
    \mathcal{B}_i
    = \left\{ x^z \cdot f_i^t \;\middle|\; 0 \leq z < d_i,\;\; 0 \leq t < 2^s \right\},
\end{equation}
This basis can be visualized in Table~\ref{tab:basis_visualization}, where the exponent~$z$ on $x$ is the horizontal direction and the exponent $t$ on $f_i$ is the vertical direction.
\begin{table}[]
    \centering
    $\begin{array}{c|cccc}
        t & z=0 & z=1 & \cdots & z=d_i{-}1 \\
        \hline
        0       & 1           & x           & \cdots & x^{d_i-1} \\
        1       & f_i       & xf_i      & \cdots & x^{d_i-1}f_i \\
        \vdots  & \vdots      & \vdots      & \ddots & \vdots \\
        2^s - 1 & f_i^{2^s-1} & xf_i^{2^s-1} & \cdots
                                             & x^{d_i-1}f_i^{2^s-1}
    \end{array}$
    \caption{A visualization of the $L_i$ basis at component $i$ generated by elements $x^zf_i^t$ for $0 \leq z \leq d_i-1$ where $d_i = \deg(f_i)$, $0 \leq t \leq 2^s-1 $. This table makes assessing the dimension of expressions like $\langle f_i^\alpha \rangle$ or $\langle f_i^\alpha \rangle/ \langle f_i ^\beta \rangle$ easy by turning the calculation into a counting argument on the number of rows and columns of this table included based on the values of $\alpha, \beta$.}
    \label{tab:basis_visualization}
\end{table}
Table~\ref{tab:basis_visualization} and its structure will become useful when assessing the dimension of quotient structures that appear in the description of logical operator spaces. The following two examples demonstrate how to use this table to assess the dimension of an ideal $\langle f^{\alpha}\rangle$ for some expression $\alpha$:
\begin{itemize}
    \item Consider the ideal $\langle f_i^a\rangle$, consisting of everything divisible by $f_i^a$, which has a basis $\{x^z f_i^t : 0 \le z < d_i,\; a \le t < 2^s\}$. A basis for this space corresponds to the rows of Table~\ref{tab:basis_visualization} starting from the row $t = a$, omitting everything above that row (rows $a-1, \dots, 0$), and including everything below it (rows $a, \dots, 2^s-1$). $\langle f_i^a \rangle $ has $\mathbb{F}_2$-dimension $d_i(2^s - a)$.
    \item  For $a < b$, consider the quotient $\langle f_i^a\rangle /\langle f_i^b\rangle $. Expanding, we have $\{x^z f_i^t + \langle f_i^b\rangle  : 0 \le z < d_i,\; a \le t < b\}$. From Table~\ref{tab:basis_visualization}, a basis for this space starts at row $a$, as in the example above, however does not extend all the way down to row $2^s-1$. Instead, we include elements of all rows through and including row $b$, and we stop there. This space has dimension $d_i(b - a)$.
\end{itemize}
Assessing the dimension of these expressions becomes central to the argument of Sec.~\ref{subsec:lop_crt_comp} in determining the number of linearly independent logical operators each component $i$ has, and in proving the set of logical operators has size $k = 2\deg(f)$, as required for a GB code.

Finally, while the analysis to follow is all done at the component level, we eventually need elements of $R_\ell$ and their $\bbF_2^{2\ell}$ representations to actually use as logical operators in a GB code. To lift elements from a single local component $L_i$ back to the ring $R_\ell$, we use $h_i = (x^\ell - 1)/f_i^{2^s}$. $h_i$ is divisible by $f_j^{2^s}$ for every $j \neq i$, and hence vanishes in each $L_j$. However at component $i$, all factors of $f_i$ have been divided out, so $h_i$ is coprime to $f_i$ and is therefore a unit in $L_i$. Thus, $h_i$ ``selects'' component $i$ and annihilates all others. Let $u_i = h_i \bmod f_i^{2^s}$ denote the (unit) image of $h_i$ in $L_i$. Then
\bea
    a\cdot u_i^{-1} \cdot h_i\bmod (x^\ell - 1)
\eea
is an element of $R_\ell$ that projects to $a$ at component $i$ and to zero at every other component. 

\subsubsection{Logical Representatives via a CRT component-level analysis}
\label{subsec:lop_crt_comp}

Denote the set of $X$-logical operators as $\mathcal{L}_X$. From Lemma~\ref{lem:kernels-characterization}, we have
\begin{equation}\label{eq:LX-def}
    \mathcal{L}_X
    \;=\;
    \frac{\ker(H_Z)}{\operatorname{rs}(H_X)}
    \;=\;
    \frac{\{(u, w) \in R_\ell^2 \mid q u + p w \in \langle \hat{f} \rangle\}}
         {\{(p a,\; q a) \mid a \in \langle f \rangle\}}.
\end{equation}

We analyze $\mathcal{L}_X$ at a single CRT component~$i$. Let $v(\cdot )$ denote the $f_i$-adic valuation, i.e., the number of times $f_i$ divides a given polynomial\footnote{i.e. $v(x) = j$ means that $x = y\cdot f_i^j$ where $f_i \nmid y$.}. At each CRT component~$i$, define the valuations:
\begin{align}
    \zeta_i &= v_i(f), &
    \rho_i &= v_i(p), &
    \sigma_i &= v_i(q).
\end{align}
Write $p = f_i^{\rho_i} p'_i$ and $q = f_i^{\sigma_i} q'_i$, where $p'_i, q'_i$ are units in~$L_i$\footnote{In a local ring, every element is either a unit, or belongs to the maximal ideal. Removing all factors of $f_i$, we know that $f_i\nmid p'$, and so $p' \notin \langle f_i\rangle$ and therefore must be a unit.}. The annihilator $\langle \hat{f}\rangle _i$ of $\langle f \rangle_i = \langle f_i^{\zeta_i} \rangle$ in $L_i$ is $\langle f_i^{2^s - \zeta_i} \rangle$. Any $\alpha \in \langle f_i^{2^s - \zeta_i} \rangle$ satisfies $\alpha \cdot f_i^{\zeta_i} = \alpha'f_i^{2^s-\zeta_i}\cdot f_i^{\zeta_i} \in \langle f_i^{2^s} \rangle$ for $\alpha'$ a unit, so $\alpha$ annihilates $\langle f_i^{\zeta_i}\rangle$. For the reverse, suppose $\beta \cdot f_i^{\zeta_i} = 0$ in $L_i$. Since $L_i$ is a chain ring, $\beta = f_i^a \cdot u$ for some $a \geq 0$ and unit $u$, and the condition $f_i^{a + \zeta_i} = 0$ forces $a + \zeta_i \geq 2^s$, i.e., $a \geq 2^s - \zeta_i$, so $\beta \in \langle f_i^{2^s - \zeta_i} \rangle$.

The condition that $\gcd(p, q, x^\ell - 1) = 1$, i.e., that $p, q$ do not share any common divisors of $x^\ell-1$ from the code construction guarantees that at each CRT component~$i$, at least one of $p, q$ is not divisible by any non-zero power of $f_i$, and therefore must be a unit in $L_i$. That is, at every component $i$,
\begin{equation}\label{eq:min-val-zero}
    \min(\rho_i, \sigma_i) = 0,
\end{equation}

Let $K_i$, $S_i$ denote the kernel of $H_Z$ and rowspace of $H_X$ at the $i$-th CRT component, respectively. We assess the case $\sigma_i = 0$ (i.e., $q_i$ is a unit); the case $\rho_i = 0$ is analogous.

The $H_Z$ kernel condition from Lemma \ref{lem:kernels-characterization} at component~$i$ is $q_i u_i + p_i w_i \in \langle \hat{f} \rangle_i = \langle f_i^{2^s - \zeta_i}\rangle$. Since $q_i$ is a unit, we can solve for $u_i$:
\begin{equation}\label{eq:CRTkernelcondition}
    u_i = q_i^{-1}p_i w_i + g_i,
\end{equation}
over $\bbF_2$ where $w_i \in L_i$ is a free parameter and $g_i \in \langle f_i^{2^s - \zeta_i}\rangle$ captures the annihilator contribution. The kernel is thus parametrized by pairs $(w_i, g_i)$:
\bea
K_i \cong \{(w_i, g_i)  \mid  w_i \in L_i, g_i \in \langle f_i^{2^s - \zeta_i} \rangle\}
\eea
where there exists an $\bbF_2$-linear isomorphism from the $(w_i, g_i)$ to the $(u_i, w_i)$ parameterization:
\bea\label{eq:kernel_wg_param}
(w_i, g_i) \mapsto (q_i^{-1}p_iw_i + g_i, \ w_i)
\eea
Observe that the logical operators at component $i$ can be described by their ``slot'' structure, i.e., their left-slot and right-slot structure. For example, a $(w_i, g_i)$ pair with $w_i = 0$ yields the logical operator $(g_i, \ 0)$, which has first slot $g_i$ and second slot empty. On the other hand, $(w_i, g_i)$ with $w_i$ non-zero yields a logical operator with both slots non-zero. 

While $K_i$ itself does not have a direct product structure, its parameterization via $(w_i, g_i)$ pairs does, and as such we can assess the dimension of $K_i$ via the direct product $L_i \times \langle f_i^{2^s - \zeta_i} \rangle$.
The component $i$ kernel then has $\dim(K_i) = \dim(L_i) + \dim(\langle f_i^{2^s-\zeta_i}\rangle) = 2^sd_i + \zeta_id_i = (2^s + \zeta_i)\, d_i$ by Table~\ref{tab:basis_visualization}. 

$S_i$ while being fully contained in $K_i$, is additionally constrained at component $i$ by the stabilizer condition $(p_i\, a,\; q_i\, a)$ with $a \in \langle f \rangle_i = \langle f_i^{\zeta_i}\rangle$. From the kernel discussion above and equation \eqref{eq:kernel_wg_param}, we have $w_i = q_i a$ with $q_i$ a unit, and so $w_i$ ranges over $\langle f_i^{\zeta_i}\rangle$. To identify the corresponding~$g_i$, satisfying the stabilizer condition, substitute $u_i = p_i a$ and
$w_i = q_i a$ into~\eqref{eq:CRTkernelcondition}:
\bea
    p_i a = q_i^{-1} p_i (q_i a) + g_i = p_i a + g_i \implies 2p_ia = g_i \equiv 0
\eea
over $\bbF_2$, and so we have that $g_i = 0$, i.e., $g_i \in \langle f_i^{2^s}\rangle = \{0\}$. Therefore the stabilizer has a $(w_i, g_i)$ parametrization as
\bea
    S_i \cong \bigl\{(w_i, g_i) \mid w_i \in \langle f_i^{\zeta_i}\rangle,\;\; g_i \in \langle f_i^{2^s}\rangle\bigr\},
\eea
which decomposes into the direct product $\langle f_i^{\zeta_i}\rangle \times \langle f_i^{2^s}\rangle$, with $\dim(S_i) = \dim(\langle f_i^{\zeta_i}\rangle) = (2^s - \zeta_i)\, d_i$.

Each $(w_i, g_i)$ pair represents a logical operator via \eqref{eq:kernel_wg_param}, and the logical space at component~$i$ is then parameterized by
\bea
    K_i / S_i = \frac{\{(w_i, g_i)  \mid  w_i \in L_i, g_i \in \langle f_i^{2^s - \zeta_i} \rangle\}}{\bigl\{(w_i, g_i) \mid w_i \in \langle f_i^{\zeta_i}\rangle,\;\; g_i \in \langle f_i^{2^s}\rangle\bigr\}}
\eea
The quotient therefore decomposes coordinate-wise:
\bea
    K_i / S_i  = \frac{L_i \times \langle f_i^{2^s - \zeta_i}\rangle}{\langle f_i^{\zeta_i} \rangle \times \langle f_i^{2^s} \rangle}\;\cong\; \frac{L_i}{\langle f_i^{\zeta_i} \rangle} \times \frac{\langle f_i^{2^s - \zeta_i} \rangle}{\langle f_i^{2^s} \rangle} \;=\; L_i / \langle f_i^{\zeta_i} \rangle \times \langle f_i^{2^s - \zeta_i} \rangle
\eea
as $\langle f_i^{2^s} \rangle = \{ 0 \}$. By Table \ref{tab:basis_visualization}, this direct product has dimension
\bea
    \zeta_i \cdot d_i \;+\; \zeta_i \cdot d_i \;=\; 2\zeta_i\, d_i.
\eea
Having established the component level structure of $K_i/S_i$, we can now write down explicit logical operator representatives.
\begin{theorem}\label{thm:explicit-basis}
Let $C$ be a GB code as in Theorem~\ref{thm:cyclic_submodule}. Let $r$ be the number of unique irreducible polynomials in the factorization of $f$, and let $\ell = 2^s m$ for $m$ odd. Let $\zeta_i = v_i(f), \ \rho_i = v_i(p), \ \sigma_i = v_i(q)$ denote the $f_i$-adic valuation of $f, p$, and $q$. Let $ q_i^{-1} \ (p_i^{-1})$ denote the inverse of $q \ (p)$ in the $i$-th CRT component $L_i$ when $\sigma_i \ (\rho_i)$ is 0 in $L_i$. Let $h_i = (x^\ell-1)/f_i^{2^s}$. For each irreducible factor $f_i$ of $x^m - 1$ with $\zeta_i > 0$, a basis for the $X$-type logical operators from component~$i$ consists of $2\zeta_i d_i$ representatives, divided into two families, with each family contributing $\zeta_id_i$ instances. Across all components $i$, these logical operators form a canonical basis of logical operator representatives of $C$.

\medskip
\noindent Case 1: $\sigma_i = 0$ (i.e., $q_i$ is a unit at component~$i$).

\smallskip
\noindent\emph{$(0, g_i)$: Single-slot representatives $(u_i, w_i)$ built via the $(w_i, g_i)$ parameterization with $w_i = 0, g_i$ free: $(u_i, w_i) = (q_i^{-1}p_iw_i + g_i, w_i) = (g_i, 0)$}:
\begin{equation}\label{eq:typeB-q-unit}
    \left\{ \Bigl(\, h_i f_i^tx^z ,\;\; 0 \,\Bigr) \;\middle|\; 2^s - \zeta_i \leq t < 2^s,\;\; 0 \leq z < d_i \right\}.
\end{equation}
These are $\zeta_id_i$ elements living in $\langle \hat{f} \rangle_i =\langle f_i^{2^s - \zeta_i}\rangle$.

\smallskip
\noindent\emph{$(w_i, 0)$: Two-slot representatives $(u_i, w_i)$ built via the $(w_i, g_i)$ parameterization with $w_i$ free, $g_i = 0$: $(u_i, w_i) = (q_i^{-1}p_iw_i + g_i, w_i) = (q_i^{-1}p_iw_i, w_i)$}:
\begin{equation}\label{eq:typeA-q-unit}
    \left\{ \Bigl(\, h_i q_i^{-1} p_i  f_i^t  x^z ,\;\; h_i  f_i^t x^z \,\Bigr)
    \;\middle|\; 0 \leq t < \zeta_i,\;\;0 \leq z < d_i \right\}.
\end{equation}
These are $\zeta_i d_i$ elements whose second slot ranges over a basis for $L_i / \langle f_i^{\zeta_i}\rangle$, with the first slot forced by the kernel relation~\eqref{eq:CRTkernelcondition}.

Case 2, when $\rho_i = 0$ (i.e., $p_i$ is a unit at component $i$), is symmetric, and the roles of the two slots swap. Concretely, swapping the left and right coordinates of the pairs above obtains the remaining logical operator representatives. Note that at each component $i$, at least one of $\rho_i, \sigma_i$ is zero, and $2\zeta_id_i$ logical operators are contributed per $p, q$ unit at component $i$.

\noindent Combining each case, the total number of $X$-type logical operators is
\begin{equation}\label{eq:k-formula-even}
    \sum_{i=1}^{r} 2\zeta_i\, d_i \;=\;  2\sum_{i=1}^{r} \zeta_i\deg(f_i) \;=\; 2\deg(f) = k
\end{equation}
\end{theorem}

The conceptual flow of this theorem is illustrated in Fig.~\ref{fig:explicit-basis-thm-concepts}.

\input{figures/theorem_4_3_fig.tex}

\begin{proof}
We verify the kernel condition, non-membership in the stabilizer, independence, and coset uniqueness for each family, treating Case~1 ($\sigma_i = 0$). Case~2 follows by symmetry.

\medskip
\noindent\textbf{Kernel/Stabilizer condition for $(0, g_i)$.}
A $(0, g_i)$ representative has $u_i = h_if_i^t  x^z$ and $w_i = 0$ for $t \geq 2^s - \zeta_i$. Let $t' = t - (2^s - \zeta_i) \geq 0$. The kernel requires $q_i u_i + p_i \cdot 0 \in \langle \hat{f}_i\rangle = \langle f_i^{2^s-\zeta_i}\rangle$. Observe,
\bea
q_i \cdot h_if_i^t  x^z = q_i \cdot h_if_i^{t'}f_i^{2^s-\zeta_i} x^z = (q_i \cdot h_if_i^{t'} x^z)f_i^{2^s-\zeta_i} \in \langle f_i^{2^s - \zeta_i}\rangle
\eea
Thus, this element is in the kernel, and at all other components the representative is zero (via the CRT lifting by $h_i$). 

A $(0, g_i)$ element is a stabilizer element if $(u_i, w_i) = (h_if_i^t  x^z, 0)= (p_i c,\; q_i c)$ for $c \in \langle f \rangle_i = \langle f_i^{\zeta_i}\rangle$. But note that $q_i$ is a unit in the $i$-th component, and $w_i = q_ic = 0$ and so it must be that $c = 0$. For this to be true, it must be that 
\bea
h_if_i^tx^z =  p_i \cdot 0 = 0 \
\eea
However, observe that $h_i$ and $x^z$ are units, and $f_i \neq 0$ for $t < 2^s$, so this cannot be satisfied for any non-zero $(0, g_i)$ element. Thus, all non-zero type $(0, g_i)$ elements are logical operators of $C$. Note that single-slot operators have no dependence on $p$ or $q$.

\medskip
\noindent\textbf{Kernel/Stabilizer condition for $(w_i, 0)$.}
A $(w_i, 0)$ representative has $u_i =  h_i q_i^{-1} p_i  f_i^t x^z , \ w_i =  h_i f_i^t x^z$. The kernel requires $q_i u_i + p_i w_i \in \langle f_i^{2^s - \zeta_i}\rangle$. Observe,
\begin{align}
q_i \cdot h_i q_i^{-1} p_i  f_i^t x^z \ + p_i \cdot h_i f_i^t x^z = h_i p_i f_i^t x^z +  p_i h_i f_i^t x^z = 0 \in \langle f_i^{2^s - \zeta_i}\rangle
\end{align}
As we are working mod 2. Thus, this element is in the kernel, and at all other components the representative is zero (via the CRT indicator~$h_i$). 

A $(w_i, 0)$ element is a stabilizer element if $(u_i, w_i) = (h_i q_i^{-1} p_i  f_i^t x^z, \ h_i f_i^t x^z) = (p_i c,\; q_i c)$ for $c \in \langle f \rangle_i = \langle f_i^{\zeta_i}\rangle$. We must then have
\begin{align}
    h_i q_i^{-1} p_i  f_i^t x^z =  p_i \left( h_i q_i^{-1} f_i^t x^z\right) &\implies h_i q_i^{-1} f_i^t x^z \in \langle f_i^{\zeta_i} \rangle\\
    h_i f_i^t x^z = q_ic &\implies q_i^{-1}h_if_i^tx^z \in \langle f_i^{\zeta_i} \rangle 
\end{align}
However, observe that in the $(w_i,0)$ case, since $t < \zeta_i$ and $h_i, q_i$, and $x^z$ are each units, we have $h_i q_i^{-1} f_i^t x^z \notin \langle f_i^{\zeta_i} \rangle$, so this element cannot possibly be a stabilizer. Thus, all $(w_i, 0)$ operators are logical operators of $C$. Note that the $(w_i, 0)$ operators do have a dependence on $p$ and $q$.

\medskip
\noindent\textbf{Independence and unique cosets.}\quad
We now show that the $2\zeta_i d_i$ representatives at each component are linearly independent modulo $S_i$, i.e., that no non-trivial $\mathbb{F}_2$-linear combination lies in the stabilizer. Recall that the stabilizer $S_i$ consists of pairs $(w_i, g_i)$ with $w_i \in \langle f_i^{\zeta_i} \rangle$ and $g_i \in \langle f_i^{2^s} \rangle = \{0\}$.

Consider a linear combination of $(0, g_i), (w_i, 0)$ elements:
\bea
\sum_\alpha c_\alpha (0, g_i)_\alpha + \sum_\beta c_\beta (w_i, 0)_\beta \in S_i.
\eea
This sum has $g_i$-component $\sum_\alpha c_\alpha g_{i,\alpha}$ and $w_i$-component $\sum_\beta c_\beta w_{i,\beta}$. Membership in $S_i$ requires $g_i = 0$ and $w_i \in \langle f_i^{\zeta_i} \rangle$. We handle each constraint in turn.

The condition $g_i = 0$ gives $\sum_\alpha c_\alpha g_{i,\alpha} = 0$. The elements 
$\{h_i f_i^t x^z : 2^s - \zeta_i \leq t < 2^s,\; 0 \leq z < d_i\}$ are linearly independent in $\langle f_i^{2^s - \zeta_i} \rangle$: 
since $h_i$ is a unit in $L_i$, the set 
$\{h_i f_i^t x^z\}$ spans the same space as the standard basis $\{f_i^t x^z\}$ of $\langle f_i^{2^s - \zeta_i} \rangle$ 
(cf.\ Table~\ref{tab:basis_visualization}, rows $t \geq 2^s - \zeta_i$). 
Thus no linear combination can obtain $0$ in the $g_i$ slot, and the only solution that satisfies the stabilizer conditions is all $c_\alpha = 0$.

The condition $w_i \in \langle f_i^{\zeta_i} \rangle$ gives $\sum_\beta c_\beta w_{i,\beta} \in \langle f_i^{\zeta_i} \rangle$, 
or equivalently, the image of $\sum_\beta c_\beta w_{i,\beta}$ in $L_i / \langle f_i^{\zeta_i} \rangle$ is zero. The elements 
$\{h_i f_i^t x^z : 0 \leq t < \zeta_i,\; 0 \leq z < d_i\}$ project to linearly independent elements of $L_i / \langle f_i^{\zeta_i} \rangle$: again, $h_i$ is a unit 
in $L_i$, so $\{h_i f_i^t x^z + \langle f_i^{\zeta_i} \rangle\}$ is a scalar multiple of the standard basis $\{f_i^t x^z + \langle f_i^{\zeta_i} \rangle\}$ for the quotient 
(cf.\ Table~\ref{tab:basis_visualization}, rows $0 \leq t < \zeta_i$). Their images being independent in the quotient forces 
all $c_\beta = 0$.

Since the only solution is the trivial one, the $2\zeta_i d_i$ 
representatives define distinct cosets in $K_i / S_i$. Representatives from distinct components $i \neq i'$ cannot be equivalent modulo the stabilizer, since the CRT lifting by $h_i$ ensures that a component $i$ representative is zero at all other components; any linear relation would therefore decompose into independent relations at each component, which we have 
already ruled out. Thus, together with the dimension count~\eqref{eq:k-formula-even}, this set 
forms a complete basis.
\end{proof}

\begin{remark}
    The CRT decomposition of $R_\ell$ can equivalently be seen through the application of Schur's lemma (specifically Artin-Wedderburn) to the group algebra $\mathbb{F}_2[\mathbb{Z}/\ell\mathbb{Z}]$. 
\end{remark}

\subsubsection{$Z$-type logical operators}
\label{subsubsec:Ztype}

For the $Z$- type logical operators, observe that $Z$ logicals live in $\ker(H_X)/\operatorname{rs}(H_Z)$.  From Lemma~\ref{lem:kernels-characterization}:
\bea
    \mathcal{L}_Z \;=\; 
    \frac{\{(u, w) \mid p\overleftarrow{u} + q\overleftarrow{w} 
    \in \langle \hat{f} \rangle\}}
    {\{(\overleftarrow{qa},\; \overleftarrow{pa}) \mid 
    a \in \langle f \rangle\}}.
\eea
The bijection $\Psi: (u, w) \mapsto (\overleftarrow{u}, \overleftarrow{w})$ on $R_\ell^2$ maps the kernel of $\mathcal{L}_Z$ to $\{(u, w) \mid pu + qw \in 
\langle \hat{f} \rangle\}$ and the rowspace to $\{(qa, pa) \mid a \in \langle f \rangle\}$, inducing an isomorphism
\bea
    \mathcal{L}_Z \;\cong\;
    \frac{\{(u, w) \mid pu + 
    qw \in \langle \hat{f} \rangle\}}
    {\{(qa, pa) \mid 
    a \in \langle f \rangle\}}.
\eea
This is precisely the $X$-logical quotient~\eqref{eq:LX-def} with $p,q$ exchanged. The CRT analysis of Theorem~\ref{thm:explicit-basis} therefore applies with the roles of $\rho_i$ and $\sigma_i$ exchanged, producing the same dimension $2\zeta_i d_i$ at each component and the same total $k = 2\deg(f)$. 

\subsubsection{Simplification when $\langle f \rangle \cap \langle \hat{f} \rangle = \{ 0 \}$}
\label{subsubsec:odd_ell}

Theorem \ref{thm:explicit-basis} constructs a canonical logical operator set in full generality for GB codes; however, we observe here that if $\langle f \rangle \cap \langle \hat{f} \rangle = \{ 0 \}$ the analysis simplifies greatly, and a much simpler set of canonical logical operators exists. 
\begin{prop}\label{prop:disjoint_f_fhat}
    When $\langle f \rangle \cap \langle \hat{f}\rangle = \{0\}$, the following is a convenient choice of logical operator representatives, where $d_f = \deg(f)$
    \begin{align}
        X \ \mathrm{logicals: } \{ (x^i\hat{f}, 0)\} \cup \{(0, x^i\hat{f}) \} \qquad 0 \leq i < d_f
        \end{align}
        \begin{align}
        Z \ \mathrm{logicals: } \{ (x^i\hat{f}(x^{-1}), 0)\} \cup \{(0, x^i\hat{f}(x^{-1})) \} \qquad 0 \leq i < d_f
    \end{align}
    which has no dependence on $p, q$. 
\end{prop}
\begin{proof}
When $\langle f \rangle \cap \langle \hat{f} \rangle = \{0\}$, each irreducible factor $f_i^{2^s}$ of $x^\ell - 1$ goes wholly to either $f$ or $\hat{f}$, so $\zeta_i \in \{0, 2^s\}$ at each CRT component. Applying Theorem~\ref{thm:explicit-basis}, the $(0, g_i)$ family contributes $d_f$ pure first-slot elements that live in $\langle \hat{f} \rangle$, and since $\{x^z \hat{f} : 0 \leq z < d_f\}$ is a basis for $\langle \hat{f} \rangle$ (Lemma~\ref{lem:annihilatordimension}), we may take $\{(x^z \hat{f}, 0)\}$ as the $(0, g_i)$ representatives.

We claim the $(w_i, 0)$ elements, mixed two-slot elements, can be replaced by the pure second-slot elements $(0, x^z \hat{f})$. Observe that $(0, x^z\hat{f})$ satisfies the kernel condition: $q \cdot 0 + p \cdot x^z\hat{f} = px^z\hat{f} \in \langle \hat{f} \rangle$, and it is not a stabilizer element: membership in $\operatorname{rs}(H_X) = \{(pa, qa) \mid a \in \langle f \rangle\}$ would require $pa = 0$ and $qa = x^z\hat{f}$, forcing $qa \in \langle f \rangle \cap \langle \hat{f} \rangle = \{0\}$; but $x^z \hat{f} \neq 0$, a contradiction. The $d_f$ elements $\{(0, x^z\hat{f})\}$ are independent by Lemma~\ref{lem:annihilatordimension}, and together with the $d_f$ first-slot representatives they give $k = 2d_f$ independent logical operators, matching the dimension of $\mathcal{L}_X$. Since none depend on $p, q$, neither does the basis.

The $Z$-type representatives follow via the $\mathcal{L}_Z \cong \mathcal{L}_X$ via $p \leftrightarrow q$ isomorphism from above: the $X$-type basis is independent of $p, q$, so swapping $p \leftrightarrow q$ leaves it unchanged, and applying $\psi_{x^{-1}}$ to return to the original coordinates sends $\hat{f} \mapsto \hat{f}(x^{-1})$.
\end{proof}
\begin{cor}
    Prop.~\ref{prop:disjoint_f_fhat} yields a canonical logical operator set for all GB codes with $\ell$ odd. When $\ell$ is even, Prop.~\ref{prop:disjoint_f_fhat} yields a canonical logical operator set if $\gcd(f, \hat{f}) = 1$.
\end{cor}
\begin{proof}
    If $\gcd(f, \hat{f}) = 1$, $f, \hat{f}$ share no common divisors of $x^\ell-1$ and as such generate disjoint ideals: $\langle f \rangle \cap \langle \hat{f} \rangle = \{0\}$. When $\ell$ is odd, $x^\ell-1$ is squarefree and $\gcd(f, \hat{f}) = 1$. When $\ell$ is even, $x^\ell-1$ is not squarefree, however if $\gcd(f, \hat{f}) = 1$ the ideals must still be disjoint and the statement of the proposition applies.
\end{proof}

Note that when $\gcd(f, \hat{f}) \neq 1$, the single slot candidates $(x^z\hat{f}, 0)$ and 
$(0, x^z\hat{f})$ are still valid logical operators, but they are no longer linearly independent modulo the stabilizer: non-trivial relations in $\mathcal{L}_X$ arise from elements of $\langle f \rangle \cap \langle \hat{f} \rangle$. The full CRT-based basis of Theorem~\ref{thm:explicit-basis}, which couples the two slots via $p$ and $q$, is required in this case.

Under the framework of Prop.~\ref{prop:disjoint_f_fhat}, as we have a set of logical $\overline{X}, \overline{Z}$ operators that span the logical space, an algorithm for a symplectic pairing is particularly simple. Recall that $k = 2\deg(f) = 2d_f$. Denote the ``trivial assignment'' of the logical operators listed above to be $\overline{X}_1 = (\hat{f}, 0),\ \overline{X}_2 = (x\hat{f}, 0),\ \dots$, with the support switching slots at $\overline{X}_{d_f+1}$, and similarly for the $Z$-logicals, using $\hat{f}(x^{-1})$. It will become useful to define the alternative notation $\overline{X}_i^{(L)}, \overline{X}_i^{(R)}$ to refer to the left and right slot logical operators with the same $x^i$ offset. 

Each logical operator is supported on a single slot, so the symplectic form
\bea
    \Omega\bigl((u_1, w_1),\, (u_2, w_2)\bigr) \;=\; 
    \langle u_1, w_2 \rangle + \langle w_1, u_2 \rangle
\eea
yields a block-diagonal commutation matrix
\bea
    \mathcal{C} \;=\; \begin{pmatrix} \mathcal{C}_L & 0 \\ 0 & \mathcal{C}_R \end{pmatrix},
    \qquad
    \mathcal{C}_{ij} \;=\; \Omega(\overline{X}_i,\, \overline{Z}_j),
\eea
where $\mathcal{C}_L$ and $\mathcal{C}_R$ are the blocks corresponding to logical operators supported on the left and right slots, respectively. Concretely, for $\overline{X}_i = (u_i, 0)$ and $\overline{Z}_j = (w_j, 0)$ both supported in the left slot, $\mathcal{C}_{ij} = \langle u_i, w_j \rangle$ reduces to the ordinary $\mathbb{F}_2^\ell$ inner product, computable via Lemma~\ref{lem:dotproduct_zerocoeff}; $\mathcal{C}_L = \mathcal{C}_R$, and cross-slot pairings vanish since neither operator has support in both slots. We will refer to $\mathcal{C}_\alpha$ when discussing a block matrix of $\mathcal{C}$ where the left vs right block distinction is not important.

The matrix $\mathcal{C}$ is invertible: the stabilizer $S$ is an isotropic subspace of $\mathbb{F}_2^{2\ell}$ under $\Omega$, so the quotient $N(S)/S$ inherits a non-degenerate symplectic form~\cite{gottesman_thesis}. Any matrix of $\Omega$ on a basis of a non-degenerate symplectic space is therefore invertible. A symplectically-paired basis is obtained by inverting $\mathcal{C}$ to obtain $\mathcal{C}^{-1}\mathcal{C} = I$, and setting $\overline{X}_i' = \overline{X}_i$ and replacing the $\overline{Z}_j$ representatives with
\bea\label{eq:Zprime}
    \overline{Z}_j^{(\alpha)\prime}
    \;=\; \sum_{h=0}^{d_f-1} (\mathcal{C}_\alpha^{-1})_{hj}\, \overline{Z}_h^{(\alpha)},
    \qquad 0 \le j < k/2,\;\; \alpha \in \{L, R\}.
\eea
Then
\bea
    \langle\overline{X}_i^{(\alpha)}, \overline{Z}_j^{(\beta)\prime}\rangle
    \;=\; \delta_{\alpha\beta} \sum_h (\mathcal{C}_\alpha^{-1})_{hj}\,[\mathcal{C}_\alpha]_{ih}
    \;=\; \delta_{\alpha\beta}\,[I]_{ij}
    \;=\; \delta_{\alpha\beta}\,\delta_{ij},
\eea

Note further that, by Lemma~\ref{lem:dotproduct_zerocoeff} and the Frobenius automorphism for $\bbF_2$, we have
\bea\label{eq:ij_symplecticpairing}
    \langle x^i \hat{f}, x^j \hat{f}(x^{-1}) \rangle = \left[ x^i \hat{f} x^{-j} \hat{f}\right]_0 = \left[ x^{i-j \pmod{\ell}} \hat{f}(x^2)\right]_0 = \left[ \hat{f}(x^2)\right]_{j-i \pmod{\ell}}
\eea
$[\mathcal{C}_\alpha]_{ij} \;=\; [\hat{f}(x^2)]_{j-i\,\bmod\,\ell}$ and the entries depend only on the difference $j - i \pmod{\ell}$. Geometrically, $\mathcal{C}_\alpha$ is the $d_f \times d_f$ leading principal submatrix of the $\ell \times \ell$ circulant with first row $\hat{f}(x^2)$. The entire commutation structure is determined by a single polynomial $\hat{f}(x^2) \in R_\ell$ rather than by $k^2$ independent entries.

This yields the following deterministic procedure for the canonical logical operator
assignment:
\begin{enumerate}
    \item Compute $\hat{f}(x^2) \in R_\ell$ from the defining polynomial $\hat{f}$.
    \item Form $\mathcal{C}_\alpha$ from $\mathrm{circ}(\hat{f}(x^2))$, the $d_f \times d_f$ leading principal submatrix (upper left hand corner)
    \item Compute $\mathcal{C}_\alpha^{-1}$ over $\bbF_2$ by Gaussian elimination.
    \item Define the canonical $X$- and $Z$-logicals by \eqref{eq:Zprime} together with the trivial $X$-side
          assignment.
\end{enumerate}
The $X$-representatives remain individually polynomial --- single-monomial multiples of $\hat{f}$ --- while the $Z$-representatives are $\bbF_2$-linear combinations of the templates $\{x^k\hat{f}(x^{-1})\}_{k=0}^{d_f-1}$, supported entirely within a single slot.

\subsection{Actions of Automorphisms and Fold-Transversal Gates on Logical Operators}
\label{subsec:lop_automorph}
 
We now connect the logical operators derived in this section to the automorphisms cataloged in Sec.~\ref{sec:AutomorphismStructure}. Given a code automorphism $\phi$, its action on the logical space is determined by how it maps logical operator representatives modulo the stabilizer rowspaces. The specific logical gate assigned to a given automorphism is dependent on parameters $\ell, f, p, q$, however, the underlying structure of a logical operator yields well-defined rules for how it evolves under any of the automorphisms discussed in this work. 

Note that in the case where $\ell$ is even with $\langle f \rangle \cap \langle \hat{f} \rangle \neq \{0\}$ case, the representatives of Theorem~\ref{thm:explicit-basis} form an $\mathbb{F}_2$-basis for the logical space but are not in general symplectically paired, i.e., the commutation relations need not hold. This does not diminish their utility for computing logical actions: since automorphisms act $\mathbb{F}_2$-linearly on the logical space, knowing how logical operators are mapped on any basis determines the action on every other basis. The CRT-derived basis provides the algebraically natural setting for deriving logical operator mappings in closed form; the symplectic adjustment is linear algebra that can be performed for any specific code instance. 

Generally, the action of an automorphism on a logical operator can do one of two things: send it back to a logical operator in the original representative set, or send it to some $\bbF_2^{2\ell}$ vector outside of that set. In the case that a logical operator is sent to something outside the set, we wish to re-express the logical operator as a linear combination of logical operators from the representative set, up to stabilizers, to determine the logical action. Re-expressing logical operators follows a few simple rules:

\begin{itemize}
    \item $x$-exponent out of bounds

    Both the logical operators of Theorem~\ref{thm:explicit-basis} and the simplification in Subsec.~\ref{subsubsec:odd_ell} are valid when the contents of a single slot include $x^k$ with $0 \leq k < d_i$ with $d_i = \deg{f_i}$ in the CRT picture, and $0 \leq k < d_f$ in the simplification. If an automorphism applies $\psi$ mapping an exponent $k$ to $kj + i$ produces a slot with $x^{kj + i}$ where $kj + i = d_i + \epsilon$ with $\epsilon \geq 0$, this logical operator is no longer represented in the canonical basis. To determine the action of the automorphism on this logical operator, we can re-express it as a linear combination of our original canonical representatives. As $f_i$ (resp $f$) has degree $d_i$ ($d_f$), we can write $f_i = x^{d_i} + \sum_{z=0}^{d_i - 1} a_z x^z$ for $a_z \in \bbF_2$. We then have
    \bea\label{eq:xd_rewrite}
        x^{d_i} = f_i + \sum_z a_z x^z
    \eea
    Substituting this expression in for $x^{d_i}$, obtain $(f_i + \sum_z a_z x^z)x^\epsilon$. Applying this substitution recursively, we eventually arrive at a linear combination of canonical logical operator representatives.

    \item $f_i$-exponent out of bounds

    This case only applies to the logical operators of Theorem~\ref{thm:explicit-basis}. If $t$ is bounded above by $2^s-1$, $t$ moving out of bounds implies the contribution of $f_i^t$ vanishes as $h_if_i^{2^s} = x^\ell - 1 = 0$ in $R_\ell$. When $t$ is bounded above by $\zeta_i-1$, $t$ moving out of bounds implies the contribution lies in $\langle f_i^{\zeta_i}\rangle = \langle f \rangle_i$, exactly the stabilizer at component $i$. There exists a stabilizer element that is exactly the contribution of $f_i^{t}$ and the remaining logical operator components can be assessed up to stabilizers.

    \item Introducing a unit under ring automorphism

    A predicate for a ring automorphism to be a valid GB code automorphism is $f(x^j) = uf$ for $u$ a unit, and so $\hat{f}(x^j) = v\hat{f}$ for $v$ a unit. Logical operators are then sent to linear combinations of valid representatives, determined by the 1's in the coefficient representation of $v$, with the $x$-exponent reduction applied as needed.
\end{itemize}

Any logical action can be determined from applying the above rules to the logical operators of Theorem~\ref{thm:explicit-basis} or Sec.~\ref{subsubsec:odd_ell}.

\section{The Maximal Cube Root Construction}\label{sec:MaxCubeRoot}

As a proof of concept, we introduce the \emph{Maximal Cube Root} (MCR) family of GB codes, whose defining polynomials are chosen by algebraic design to guarantee a large automorphism group together with a rich set of fold-transversal CX gates. The construction is governed by a single structural requirement: the existence of the transfer \emph{ratio} $r = pq^{-1}$ satisfying the identity $r^2 + r + 1 = 0$ in the quotient ring $S$, which as we will see, has roots at exactly the non-trivial cube roots of unity. We show that this single identity simultaneously collapses the four CX-type gate conditions into a single binary condition on multipliers, and that the resulting collapse determines the automorphism group and the fold-CX set in one stroke.

The construction is best read as a deliberate choice of $(\ell, f, p, q)$ that simultaneously activates the algebraic machinery developed in Sec.~\ref{sec:3SpaceIso}--\ref{sec:LogicalOps}. Theorem~\ref{thm:cyclic_submodule} expresses the GB stabilizer rowspaces as cyclic submodules of $R_\ell^2$ generated by $(p,q)$ over $\langle f\rangle$, so that every code-preservation condition becomes an algebraic identity on $(f, p, q)$ in the quotient $S = R_\ell / \langle \hat f \rangle$. Theorem~\ref{thm:substitution_automorphisms} characterizes the block-separable substitution-multiplier automorphisms by four conditions on $(p, q)$ in $S$ --- the families $\mathrm{Stab}, \mathrm{Swap}, \mathrm{Inv}, \mathrm{SwapInv}$, indexed by $j \in \mathrm{Pres}(f)$. Finally, Theorem~\ref{thm:fold_cx} characterizes the fold-transversal CX gates compatible with a multiplier $\psi_{x^j}$ by four M\"obius-type identities on $(p, q)$. The MCR design replaces all of these polynomial conditions on $(p, q)$ with a single scalar identity on the ratio $r = pq^{-1}\in S$, namely $r^2 + r + 1 = 0$. The requirement that $\ell$ be odd places the construction in the simplified regime of Sec.~\ref{subsubsec:odd_ell}, in which the canonical logical operators of Theorem~\ref{thm:explicit-basis} take a particularly clean, $p, q$-independent form---making the logical action of any candidate symmetry computable without re-solving the logical operator kernel for each new choice of $(p, q)$. The MCR construction is therefore the natural place where the cyclic-submodule description, the substitution and fold-CX classifications, and the simplified logical operator basis of the previous sections all meet on a single algebraic identity.

\begin{definition}[Maximal Cube Root code]\label{def:MCR}
A \emph{Maximal Cube Root (MCR) code} is a Generalized Bicycle code (Theorem~\ref{thm:cyclic_submodule}) specified by a tuple $(\ell, f, p, q)$ with $f \mid x^\ell - 1$ and $\gcd(p, q, x^\ell - 1) = 1$, subject to the following four conditions:
\begin{enumerate}
    \item $\ell$ is odd
    \item every irreducible factor of $\hat{f} = (x^\ell - 1)/f$ has even degree over $\mathbb{F}_2$
    \item $q$ is invertible in $S = R_\ell / \langle \hat f \rangle$
    \item the transfer ratio $r = p q^{-1} \in S$ satisfies $r^2 + r + 1 = 0$ in $S$
\end{enumerate}
\null\hfill\ensuremath{\square}
\end{definition}

Condition~(1) makes $x^\ell - 1$ squarefree, so $R_\ell$ and $S$ decompose under the CRT as a product of finite fields. Condition~(3) ensures the ratio $r$ is well-defined, and condition~(4) is the central algebraic constraint whose consequences occupy the rest of this section. Condition~(2) is an existence requirement for~(4), which we address first. Note that $x + 1$ always divides $x^\ell - 1$ and has degree one, so condition~(2) forces $x + 1 \mid f$ (rather than $x + 1 \mid \hat{f}$).

Writing $S \cong \prod_i \mathbb{F}_{2^{d_i}}$ (as $\ell$ is odd), the condition $r^2 + r + 1 = 0$ in $S$ is equivalent to $r_i^2 + r_i + 1 = 0$ in $\mathbb{F}_{2^{d_i}}$ for every $i$. We therefore need $r^2 + r + 1 \in \mathbb{F}_2[r]$ to have a root in every component field.

Roots of $r^2 + r + 1$ are exactly the primitive cube roots of unity, since $r^3 + 1 = (r + 1)(r^2 + r + 1)$ over $\mathbb{F}_2$. The multiplicative group $\mathbb{F}_{2^d}^\times$ is cyclic of order $2^d - 1$ and contains an element of order~3 if and only if $3 \mid 2^d - 1$. Computing modulo~3,
\bea
    2^d \equiv \begin{cases} 1 \pmod 3 & \text{if $d$ is even,} \\ 2 \pmod 3 & \text{if $d$ is odd,} \end{cases}
\eea
so $3 \mid 2^d - 1$ precisely when $d$ is even. Hence $r^2 + r + 1$ has a root in every CRT component of $S$ if and only if every $d_i$ is even, which is condition~(2). When this holds, each component contributes exactly two roots---some $\omega_i \in \mathbb{F}_{2^{d_i}}$ and its conjugate $\omega_i^2 = \omega_i + 1$---and roots can be chosen independently across components, yielding $2^t$ valid ratios $r \in S$ where $t$ is the number of CRT components.

Finally, condition~(4) yields $r^2 + r + 1 = 0$ in $S$, from which three identities immediately follow:
\begin{equation}\label{eq:mcr-identities}
    r^{-1} = r^2 = r + 1, \qquad (r+1)^{-1} = r, \qquad \frac{r}{r+1} = r^2 = r + 1.
\end{equation}
The first follows directly from $r^2 + r + 1 = 0$, the second from $r \cdot (r + 1) = r^2 + r = 1$, and the third from combining the previous two. As $\psi_{x^{-1}}$ is a ring automorphism, the same identities hold for $\overleftarrow{r}$ over $\overleftarrow{S} = R_\ell/\langle \hat{f}(x^{-1}) \rangle$. Two structural consequences are immediate:
\begin{itemize}
    \item \emph{Invertibility of $p, q, r$ is automatic.} $q$ is invertible by design, $r$ is invertible as $r\cdot r^2 = 1$, and $p$ is invertible as $p = qr$ where both $q, r$ are invertible. All invertibility hypotheses required by the CX-type gate conditions of Theorem~\ref{thm:fold_cx} (the $\gcd(p, \hat f) = 1$ and $\gcd(q, \hat f) = 1$ premises) are therefore automatically satisfied.

    \item \emph{Four conditions collapse to two.} The identities~\eqref{eq:mcr-identities} force
    \bea
        r^{-1} \;=\; r + 1 \;=\; \tfrac{r}{r+1} \qquad \text{in } S,
    \eea
    so the three non-trivial conditions coincide. Under the cube root identity, the global image $r(x^j)$ can equal $r$, $r + 1$, or neither. 
\end{itemize}

The preceding frameworks yield a substantial collapse of the substitution-multiplier conditions of Theorem~\ref{thm:substitution_automorphisms}, and the fold-CX conditions of Theorem~\ref{thm:fold_cx} collapse into a single equation on the images of $r(x^j) \in S$ and $\overleftarrow{S}$. Rather than checking each of these conditions on $p, q$ separately under $\phi_{x^j}$, one need only check a single scalar equation on the transfer ratio $r = p q^{-1} \in S$ (or, symmetrically, on $p^{-1}q$). We formalize this in Prop.~\ref{prop:ratio-reduction}.
 
Throughout this argument we assume $q \in S^\times$, so that $r = p q^{-1} \in S$ is defined; the case where $p \in S^\times$ (using $r' = qp^{-1}$) is analogous. When neither $p, q$ is invertible in $S$, the ratio-based framework is not available and the code is not an MCR code; however, the polynomial conditions of $\mathrm{Stab}$ and $\mathrm{Swap}$ remain directly checkable, and interesting code automorphisms may remain.

\begin{prop}\label{prop:ratio-reduction}
Let $(\ell, f, p, q)$ define a GB code with $\ell$ odd, $\gcd(p, q, x^\ell - 1) = 1$ and $q \in S^\times$, and set $r = p q^{-1} \in S$. For any multiplier $j \in \mathrm{Pres}(f)$:
\begin{enumerate}[label=\textup{(\alph*)}]
    \item $\phi_{x^j}$ is a $M_\sim$ GB code automorphism if and only if $r(x^j) = r$ in $S$.
    \item $\sigma \circ \phi_{x^j}$ is a $M_\sim$ GB code automorphism if and only if $r(x^j) =  r^{-1}$ in $S$.
\end{enumerate}
\end{prop}
 
\begin{proof}
We prove~(a); part~(b) is analogous.

As $j$ preserves $\langle f \rangle$, the map $\phi_{x^j}$ is a permutation that sends $(pa, qa)$ to $(p(x^j)a',q(x^j)a')$ for some $a' \in \langle f \rangle$ in $S$. Let $M(\alpha, \beta)$ denote the cyclic submodule generated by $\alpha, \beta$. We have
\bea
    \phi_{x^j}\bigl(M(p, q)\bigr) \;=\; M\bigl(p(x^j),\, q(x^j)\bigr).
\eea
As $\phi_{x^j}$ is a ring automorphism applied blockwise, it is a $\mathbb{F}_2$-linear bijection of $R_\ell^2$ and preserves $\bbF_2$-dimension. As these two submodules have equal $\mathbb{F}_2$-dimension, their equality is equivalent to the single containment
\bea
    M\bigl(p(x^j), q(x^j)\bigr) \;\subseteq\; M(p, q).
\eea
This containment holds if and only if, for every $a' \in \langle f \rangle$, there exists $a \in \langle f \rangle$ with
\bea
    p a \;=\; p(x^j)\, a' \qquad \text{and} \qquad q a \;=\; q(x^j)\, a' \qquad \text{in } R_\ell.
\eea
Since $q \in S^\times$, the second equation determines $a$ uniquely in $S$: $a \equiv q^{-1} q(x^j) a' \pmod{\hat f}$. Substituting into the first gives
\bea
    r \cdot q(x^j) - p(x^j)\;=\; 0 \qquad \text{in } S.
\eea
As $a'$ ranges over $\langle f \rangle$ and the image of $\langle f \rangle$ in $S$ is all of $S$ (Lemma \ref{lem:f_unit_in_S}), the identity holds for every $a' \in S$, forcing $r \cdot q(x^j) = p(x^j)$ in $S$. Dividing by $q(x^j) \in S^\times$ then yields $r(x^j) = r \in S^\times$.
 
Conversely, if $r(x^j) = r$, then $p(x^j) = r \cdot q(x^j)$, and setting $a = q^{-1} q(x^j) a'$ for each $a'$ produces the required element.
\end{proof}

The fold-CX preserving conditions of Theorem~\ref{thm:fold_cx} admit analogous reductions by condition 4, which we formalize in Prop.~\ref{prop:mcr_cx_reduction}.
\begin{prop}\label{prop:mcr_cx_reduction}
    Let the conditions of Prop.~\ref{prop:ratio-reduction} hold, and additionally, restrict to $r$ such that $r^2 + r + 1 = 0$ in $S$. Then, the $M_\sim$ fold-CX gates of Theorem~\ref{thm:fold_cx} hold if $r(x^j) = r^{-1}$ in $S$ and the $M_\leftrightarrow$ gates hold if $r(x^{-j}) = r$ in $S$.
\end{prop}
\begin{proof}
    From the previous discussion we have that if $q$ is invertible and $r^2 + r + 1 =0$, $p$ must also be invertible. Using the identities on $r^2 + r + 1 = 0$ from above, observe the $M_\sim$ conditions for fold CNOTs collapse:
    \begin{gather}
    p(x^j) = \frac{p + q}{q}q(x^j) \implies r(x^j) = r + 1 = r^{-1} \textrm{ in } S \\
    q(x^j) = \frac{p + q}{p}p(x^j) \implies r^{-1}(x^j) = r^{-1} + 1 \implies r(x^j) = r + 1 = r^{-1}
    \textrm{ in } S
    \end{gather}
    Similarly, for the $M_\leftrightarrow$ conditions, over $\overleftarrow{S}$:
    \begin{gather}
    p(x^j) = \frac{\overleftarrow{p} + \overleftarrow{q}}{\overleftarrow{p}}q(x^j) \implies r(x^j) = \overleftarrow{r^{-1}} + 1 = \overleftarrow{r} \textrm{ in } \overleftarrow{S} \\
    q(x^j) = \frac{\overleftarrow{p} + \overleftarrow{q}}{\overleftarrow{q}}p(x^j) \implies r^{-1}(x^j) = \overleftarrow{r} + 1 \implies r(x^j) = \overleftarrow{r^{-1}} + 1 = \overleftarrow{r} \textrm{ in } \overleftarrow{S}
    \end{gather}
    As $\psi_{x^{-1}}$ is a ring automorphism, $r(x^j) = \overleftarrow{r}$ in $\overleftarrow{S}$ is the same condition as $r(x^{-j}) = r$ in $S$.
\end{proof}

Collecting Prop.~\ref{prop:ratio-reduction} and \ref{prop:mcr_cx_reduction} yields the following simplified conditions:
\medskip
\begin{center}
\renewcommand{\arraystretch}{1.2}
\begin{tabular}{l l}
\hline
Gate condition & Required action on $r$ \\
\hline
$\mathrm{Stab}(p, q),\; \mathrm{SwapInv}(p, q)$     & $r(x^j) = r$ in $S$ \\
$\mathrm{Swap}(p, q),\; \mathrm{Inv}(p, q)$         & $r(x^j) = r^{-1}$ in $S$ \\
Fold-CX ($2 \to 1$, $M_\sim$)                     & $r(x^j) = r^{-1}$ in $S$ \\
Fold-CX ($1 \to 2$, $M_\sim$)                     & $r(x^j) = r^{-1}$ in $S$\\
Fold-CX ($2 \to 1$, $M_\leftrightarrow$)          & $r(x^j) =r$ in $S$\\
Fold-CX ($1 \to 2$, $M_\leftrightarrow$)          & $r(x^j) = r$ in $S$ \\
\hline
\end{tabular}
\end{center}
\medskip

Note that the condition $r(x^j) = r$ in $S$ is used for the $M_\leftrightarrow$ fold-CX's, which on first glance, is not consistent with what was proved in Prop.~\ref{prop:mcr_cx_reduction}. Observe that if $r(x^j) = r$ in $S$, as exponents are taken $\pmod{\ell}$, every $j$ has a corresponding $k$ such that $-k = j \pmod{\ell}$. Thus, $r(x^j) = r(x^{-k}) = r$ for some $k$, yielding a $M_\leftrightarrow$ fold-CX solution. We phrase the condition this way to draw attention to the collapse of $\mathrm{Stab}(p, q),\; \mathrm{SwapInv}(p, q)$ and $M_\leftrightarrow$ Fold-CX into one algebraically checkable condition.

The collapse immediately yields the defining structural property of MCR codes:
\begin{itemize}
    \item If $r(x^j) = r$, then $j$ \emph{simultaneously}: 
    
    \begin{enumerate}[label=(\alph*)]
        \item induces a $M_\sim$ multiplier automorphism via $\phi_{x^j}$ for $j \in \mathrm{Stab} \cup \mathrm{SwapInv}$, and
        \item provides a fold-transversal CX gate in each of the two $M_\leftrightarrow$ directions of Theorem~\ref{thm:fold_cx} using $\phi_{x^{-j}}$.
    \end{enumerate}
    
    \item If $r(x^j) = r + 1$, then $j$ \emph{simultaneously}:
    \begin{enumerate}[label=(\alph*)]
        \item induces a $M_\sim$ multiplier automorphism via $\sigma \circ \phi_{x^j}$ for $j \in \mathrm{Swap} \cup \mathrm{Inv}$, and,
        \item provides a fold-transversal CX gate in each of the two $M_\sim$ directions of Theorem~\ref{thm:fold_cx} using $\phi_{x^j}$.
    \end{enumerate}
\end{itemize}

We refer to the collapse of $r(x^j)$ into conditions that simultaneously yield an automorphism and two Fold-CX gates as \emph{maximal coupling}: a single algebraic identity ensures that every multiplier whose action on $r$ is fixing or inverting simultaneously furnishes a block-transposing code automorphism and two fold-transversal CX gates. The maximality is precisely this coincidence — three a-priori-distinct ratio conditions forced to a common target by the cube-root identity.

Tables~\ref{tab:k2-mcr} and~\ref{tab:k-large-mcr} list representative MCR codes for a range of $\ell$, along with their parameters $[[n, k, d]]$, stabilizer weights~$w$, and automorphism group sizes. Tables~\ref{tab:k2-logical-gates} and~\ref{tab:k6-logical-gates} demonstrate the full logical gate set achievable for a selected $k = 2$ code and $k = 6$ code using the canonical logical assignment from Sec.~\ref{subsubsec:odd_ell}. Note that different logical assignments will yield different logical gate sets. 

The generating set of automorphisms for the $k=2$ case makes the following Proposition (Prop.~\ref{prop:addressable_hs}) and Corollary (Corollary~\ref{cor:2qubitcliffordgroup}) relevant:
\begin{prop}\label{prop:addressable_hs}
    Addressable $H$ and $S$ can be obtained from the set $\{H^{\otimes2},S^{\otimes2}, CZ, CX_{1\rightarrow2} \}$ on 2 qubits, up to global phase.
\end{prop}
\begin{proof}
    First, we show that 
    \bea
        S_2 = CX_{1\rightarrow2} \cdot S^{\otimes 2} \cdot CX_{1\rightarrow2} \cdot CZ
    \eea
    This maps $X_1, X_2, Z_1, Z_2$ as follows:
    \begin{gather}
        X_1 \xrightarrow{CZ} X_1Z_2 \xrightarrow{CX_{1\rightarrow2}} Y_1Y_2 \xrightarrow{S^{\otimes 2}} X_1X_2 \xrightarrow{CX_{1\rightarrow2}} X_1 \\
        X_2 \xrightarrow{CZ} Z_1X_2 \xrightarrow{CX_{1\rightarrow2}} Z_1X_2  \xrightarrow{S^{\otimes 2}} Z_1Y_2  \xrightarrow{CX_{1\rightarrow2}} Y_2 \\
        Z_1 \xrightarrow{CZ} Z_1 \xrightarrow{CX_{1\rightarrow2}} Z_1 \xrightarrow{S^{\otimes 2}} Z_1 \xrightarrow{CX_{1\rightarrow2}} Z_1\\
        Z_2 \xrightarrow{CZ} Z_2 \xrightarrow{CX_{1\rightarrow2}} Z_1Z_2 \xrightarrow{S^{\otimes 2}} Z_1Z_2 \xrightarrow{CX_{1\rightarrow2}} Z_2
    \end{gather}
    yielding the action of $S_2$. Given access to $S_2$ and $S^{\otimes 2}$, $S_1$ is additionally accessible.

    Next, we show that
    \bea
        H_1 = S^{\otimes 2} \cdot CX_{1\rightarrow2} \cdot H^{\otimes 2} \cdot CZ \cdot S^{\otimes 2} \cdot H^{\otimes 2} \cdot CX_{1\rightarrow2} \cdot S^{\otimes 2}
    \eea
    This maps $X_1, X_2, Z_1, Z_2$ as follows:
    \begin{align}
    X_1 &\xrightarrow{S^{\otimes 2}} Y_1 \xrightarrow{CX_{1\rightarrow2}} Y_1X_2 \xrightarrow{H^{\otimes 2}} -Y_1Z_2 \xrightarrow{S^{\otimes 2}} X_1Z_2 \xrightarrow{CZ} X_1 \xrightarrow{H^{\otimes 2}} Z_1 \xrightarrow{CX_{1\rightarrow2}} Z_1 \xrightarrow{S^{\otimes 2}} Z_1 \\
    Z_1 &\xrightarrow{S^{\otimes 2}} Z_1 \xrightarrow{CX_{1\rightarrow2}} Z_1 \xrightarrow{H^{\otimes 2}} X_1 \xrightarrow{S^{\otimes 2}} Y_1 \xrightarrow{CZ} Y_1Z_2 \xrightarrow{H^{\otimes 2}} -Y_1X_2 \xrightarrow{CX_{1\rightarrow2}} -Y_1 \xrightarrow{S^{\otimes 2}} X_1 \\
    X_2 &\xrightarrow{S^{\otimes 2}} Y_2 \xrightarrow{CX_{1\rightarrow2}} Z_1Y_2 \xrightarrow{H^{\otimes 2}} -X_1Y_2 \xrightarrow{S^{\otimes 2}} Y_1X_2 \xrightarrow{CZ} X_1Y_2 \xrightarrow{H^{\otimes 2}} -Z_1Y_2 \xrightarrow{CX_{1\rightarrow2}} -Y_2 \xrightarrow{S^{\otimes 2}} X_2 \\
    Z_2 &\xrightarrow{S^{\otimes 2}} Z_2 \xrightarrow{CX_{1\rightarrow2}} Z_1Z_2 \xrightarrow{H^{\otimes 2}} X_1X_2 \xrightarrow{S^{\otimes 2}} Y_1Y_2 \xrightarrow{CZ} X_1X_2 \xrightarrow{H^{\otimes 2}} Z_1Z_2 \xrightarrow{CX_{1\rightarrow2}} Z_2 \xrightarrow{S^{\otimes 2}} Z_2
    \end{align}
    yielding the action of $H_1$. Given access to $H_1$ and $H^{\otimes 2}$, $H_2$ is additionally accessible.
\end{proof}
Note there is nothing special about choosing $CX_{1\rightarrow 2}$ as opposed to $CX_{2\rightarrow1}$, and $CX_{2\rightarrow1}$ can equivalently yield the full Clifford group.
\begin{cor}\label{cor:2qubitcliffordgroup}
    The set $\{H^{\otimes2},S^{\otimes2}, CZ, CX_{1\rightarrow 2} \}$ generates the 2-qubit Clifford group.
\end{cor}
\begin{proof}
    Immediate from Prop.~\ref{prop:addressable_hs}.
\end{proof}

We end this section with a brief discussion of the contents of Tables~\ref{tab:k2-mcr}--\ref{tab:k6-logical-gates}. 
\begin{itemize}
    \item Entries in Table~\ref{tab:k2-mcr} with $|\mathrm{Gates}| = 10$ generate the 2-qubit Clifford group via automorphism and fold-transversal gate implementation. Entries in Table~\ref{tab:k2-mcr} with $|\mathrm{Gates}| = 9$ are all identically missing the $S^{\otimes 2}$ gate and do not generate the 2-qubit Clifford group. This is because all $S$-type automorphisms for the $|\mathrm{Gates}| = 9$ codes are block-swapping, whereas codes with $|\mathrm{Gates}| = 10$ that admit the $S^{\otimes 2}$ gate possess at least one block-preserving $S$-type automorphism.
    
    \item As $k$ scales, generating the full Clifford group becomes more difficult, and more addressable gates are required. Since we restrict to $\psi_L = \psi_R$, the same permutation is being carried out on both halves of the code, and it is reasonable to expect a ``symmetric'' logical action. Fold-CX gates, in which the logical action is not completely symmetric, are the exception to this rule. However, fold-CX-type gates only seem to break symmetry between the two halves of the codes, and when each half supports more than 1 qubit each, it remains unclear if this can be leveraged to obtain addressability for $k > 2$.
    
    An interesting follow-up question is to explore the asymmetric regime (where $\psi_L \neq \psi_R$) and/or block-mixing automorphisms to determine if relaxing the symmetry constraints allows one to access such addressable gates for larger $k$.

    \item In each case, the naive stabilizer weight resulting from $\text{wt}(pf) + \text{wt}(qf)$ can be, and often is, quite high. By Lemma~\ref{lem:equivalent_modules}, we may use unit scaling of $(pf, qf)$ to obtain a generator of $M_f(p,q)$ with potentially lower weight. Finding such a $u$ that minimizes the stabilizer weight is a variant of the minimum weight codeword search for linear codes and is, in general, NP-hard. Simple methods using lattices \cite{conway_sloane_1999} and information-set decoding (ISD) \cite{prange_1962} are sufficient in practice for the code parameters considered here: the former provides a fast feasible upper bound, and the latter certifies optimality when $\dim\langle f \rangle$ is small, and otherwise refines the bound for larger instances. Lattice and ISD-style methods are both available in the accompanying code. ILP based methods using gurobi \cite{gurobi} or other solvers are sufficient for exact minima.

    Note that one has to take care not to choose any arbitrary element of $M_f(p,q)$, as there often exist many low-weight, easily obtainable vectors that scale $(pf, qf)$ by non-unit elements of $R_\ell$ and do not generate the entirety of $M_f(p,q)$. As such, cyclic shifts of such a vector cannot be used as stabilizers for these codes.

    \item In the $k=2$ case, the generating set of Corollary~\ref{cor:2qubitcliffordgroup} is quite small, and MCR codes are ``over-automorphismed'', even at $\ell = 18$. That is, the size of the automorphism group is vastly larger than the number of obtainable logical actions, $|\mathrm{Gates}|$. Simultaneously, the stabilizer weight of these $k=2$ codes appears to grow with $\ell$. It would be interesting to see if relaxing any of the MCR conditions can yield a $k=2$ family that generates the Clifford group with weight-8 stabilizers.
\end{itemize}

\begin{sidewaystable}[ht]
\centering
\renewcommand{\arraystretch}{1.4}
\begin{tabular}{cccc>{\centering\arraybackslash$}p{8cm}<{$}ccc}
\hline
$[[n, k, d]]$ & $\ell$ & $f$ & $p$ & \multicolumn{1}{c}{$q$}  & $w$ & $|\mathrm{Aut}|$ & $|\mathrm{Gates}|$ \\
\hline
$[[18, 2, 5]]$  & 9 & $x + 1$ & $1$ & x^7 + x^4 + x^3 + x & 8 & 108 & 10 \\
$[[22, 2, 6]]$  & 11 & $x + 1$ & $1$ & x^9 + x^5 + x^4 + x^3 + x & 8 & 220 & 10 \\
$[[30, 2, 7]]$  & 15 & $x + 1$ & $1$ & x^{13} + x^9 + x^7 + x^6 + x^5 + x^4 + x + 1 & 8 & 120 & 9 \\
$[[50, 2, 9]]$  & 25 & $x + 1$ & $1$ & x^{23} + x^{22} + x^{18} + x^{17} + x^{15} + x^{13} + x^{12} + x^{10} + x^8 + x^7 + x^3 + x^2 + 1 &  12 & 1000 & 9 \\
$[[54, 2, 10]]$ & 27 & $x + 1$ & 1 & x^{25} + x^{22} + x^{21} + x^{19} + x^{18} + x^{16} + x^{13} + x^{12} + x^{10} + x^7 + x^4 + x^3 + x + 1 & 16 & 972 & 10 \\
$[[58, 2, 11]]$ & 29 & $x + 1$ & 1 & x^{27} + x^{26} + x^{21} + x^{19} + x^{18} + x^{17} + x^{15} + x^{14} + x^{12} + x^{11} + x^{10} + x^8 + x^3 + x^2 &  12  & 1624 & 9 \\
$[[66, 2, 13]]$ & 33 & $x + 1$ & 1 & x^{31} + x^{30} + x^{28} + x^{25} + x^{24} + x^{21} + x^{19} + x^{18} + x^{16} + x^{13} + x^{11} + x^{10} + x^7 + x^6 + x^4 + x &  12 & 660 & 10 \\
\hline
\end{tabular}
\caption{MCR codes with $k = 2$. All codes are defined by $f = x + 1$. A representative $(p, q)$ is used for concreteness, alternative choices exist for each row. With only 2 logicals, the set of obtainable Clifford gates via automorphism is limited. Entries with $|\mathrm{Gates}| = 10$ gates have the logical gates of \ref{tab:k2-logical-gates} and generate the 2-qubit Clifford group. Codes with $|\mathrm{Gates}| = 9$ are identically missing $S^{\otimes 2}$ and do not generate the 2-qubit Clifford group. $|\mathrm{Gates}|$ does not include gates achievable from composing automorphisms, and only lists gates achievable from Tables~\ref{tab:block_sep_auts} and~\ref{tab:fold_gates}. Distance was found using the AB reduction methods \cite{wang2022distanceboundsgeneralizedbicycle} and the Gurobi optimization package \cite{gurobi}. $w$ refers to the stabilizer weight that is achievable for a given code. Stabilizer weights are obtainable using the lattice/ISD-style methods found in \texttt{low\_weight\_generator.sage}, located in the associated code repo, and the Gurobi optimization package was used to confirm minimality\cite{gurobi}. The collection of codes displayed here was chosen to demonstrate increasing distances while keeping $k = 2$. Each instance included here is the smallest $n$ for which the respective $d$ was achieved. Codes here were collected via \texttt{MCR\_code\_search.sage}, and analyzed with \texttt{gb\_code\_analysis.sage} found in \url{https://github.com/ajdav136/GBAutomorphisms} }
\label{tab:k2-mcr}
\end{sidewaystable}

\begin{sidewaystable}[ht]
\centering
\renewcommand{\arraystretch}{1.4}
\begin{tabular}{cccc>{\centering\arraybackslash$}p{8cm}<{$}ccc}
\hline
$[[n, k, d]]$ & $\ell$ & $f$ & $p$ & \multicolumn{1}{c}{$q$}  & $w$ & $|\mathrm{Aut}|$ & $|\mathrm{Gates}|$ \\
\hline
$[[30, 6, 5]]$ & 15 & $x^3 + 1$ & $1$ & x^{10} + x^9 + x^6 + x^5 + 1 & 8 & 240 & 20 \\
$[[66, 6, 8]]$ & 33 & $x^3 + 1$ & $1$ & x^{27} + x^{22} + x^{15} + x^{12} + x^{11} + x^9 + x^3 & 12 & 1320 & 22 \\
$[[78, 6, 9]]$ & 39 & $x^3 + 1$ & 1 & x^{33} + x^{26} + x^{24} + x^{21} + x^{18} + x^{15} + x^{13} + x^6 + 1 & 12 & 1872 & 20\\
$[[90, 10, 10]]$ & 45 & $x^5 + x^3 + x  + 1$  & 1 & x^{38} + x^{36} + x^{35} + x^{32} + x^{27} + x^{23} + x^{20} + x^{18} + x^{17} + x^{15} + x^9 + x^8 + x^5 + x^2 + 1 & $\leq 18$ & 1080 & 30 \\
$[[102, 6, 11]]$ & 51 & $x^3 + 1$  & $1$  & x^{45} + x^{34} + x^{33} + x^{30} + x^{27} + x^{24} + x^{21} + x^{18} + x^{17} + x^6 + 1 & $\leq 12$ & 1632 & 20\\
$[[102, 18, \leq 12]]$ & 51 & $x^9 + x^4 + x^2 + 1$  & $1$ & x^{39} + x^{38} + x^{36} + x^{35} + x^{34} + x^{33} + x^{32} + x^{30} + x^{29} + x^{28} + x^{27} + x^{25} + x^{23} + x^{21} + x^{19} + x^{18} + x^{15} + x^{14} + x^{11} + x^{10} + x^8 + x^6 + x^4 + x^3 + x & $ \leq 22$ & 816 & 76\\ 
$[[110, 10, 10]]$ & 55 & $x^5 + 1$ & $1$ &x^{45} + x^{44} + x^{33} + x^{25} + x^{22} + x^{20} + x^{15} + x^{11} + x^5 + 1  & $\leq 16$ & 4400 & 40 \\
\hline
\end{tabular}
\caption{MCR codes with $k > 2$. Only codes with competitive distance for their size are listed; a complete list is available from the Sage code in the GitHub repository. These codes were selected to demonstrate a range of $k$ values yielding increasing distances, and a variety of stabilizer weights and automorphism gates achievable at the same $k, d$ values. Codes with $n < 90$ did not yield competitive distance $k = 10$ codes, and are omitted from this list. Observe that as $k$ grows, the number of possible $k$-qubit Clifford gates achievable via automorphism grows, and we no longer exhaust the possible gate collection. Codes here were collected via \texttt{MCR\_code\_search.sage}, and analyzed with \texttt{gb\_code\_analysis.sage} found in \url{https://github.com/ajdav136/GBAutomorphisms}. All $w$ entries listed are obtainable with the code found in \texttt{low\_weight\_generator.sage}. Entries without a $\leq$ symbol have been validated as minima by an ILP.}
\label{tab:k-large-mcr}
\end{sidewaystable}

\FloatBarrier
\begin{table}
\centering
\begingroup
\renewcommand{\arraystretch}{1.18}
\setlength{\tabcolsep}{5pt}
\adjustbox{max width=\linewidth,max totalheight=\dimexpr\textheight-2cm\relax}{%
\begin{tabular}{@{} >{\raggedright\arraybackslash}p{4.3cm}
                    c
                    >{\raggedright\arraybackslash}p{3.0cm}
                    >{\centering\arraybackslash}p{3.8cm}
                    >{\raggedright\arraybackslash}p{5.0cm} @{}}
\toprule
\textbf{$\psi$ \,/\, cycle on $S_\ell$} &
\textbf{Type} &
\textbf{Logical Action} &
\textbf{Logical Circuit} &
\textbf{Other sources (same logical gate)}\\
\midrule

$\phi_{1} = \psi_1 \oplus \psi_1$ \newline $\psi_1 = (0\,1\,2\,3\,4\,5\,6\,7\,8)$ &
\tAut &
Id &
\circwrap{\begin{quantikz}[lqstyle]
 \lstick{$1$} & \qw & \qw\\
 \lstick{$2$} & \qw & \qw
\end{quantikz}} &
All $9$ shifts $\phi_0,\dots,\phi_8$, $M_{\sim}$ multipliers $\phi_{x^{1}},\phi_{x^{4}},\phi_{x^{7}}$. $12$ sources total.\\
\midrule

$\sigma\!\cdot\!\phi_{x^{2}} = \sigma \circ (\psi_{x^2} \oplus \psi_{x^2})$\newline $\psi_{x^2} = (1\,2\,4\,8\,7\,5)(3\,6)$\; &
\tAut &
SWAP &
\circwrap{\begin{quantikz}[lqstyle]
 \lstick{$1$} & \swap{1} & \qw\\
 \lstick{$2$} & \targX{} & \qw
\end{quantikz}} &
$\sigma\phi_{x^{5}}, \ \ \sigma\phi_{x^{8}}$. $3$ sources total.\\
\midrule

$\phi_{x^{1}} = \psi_{x^1} \oplus \psi_{x^1}$\newline $\psi_{x^1} = \mathrm{id};\ +\,H^{\otimes n}$ &
\tH &
$H^{\otimes 2}$ &
\circwrap{\begin{quantikz}[lqstyle]
 \lstick{$1$} & \gate[style={inner sep=-1pt}]{H} & \qw\\
 \lstick{$2$} & \gate[style={inner sep=-1pt}]{H} & \qw
\end{quantikz}} &
$\phi_{x^{4}}, \ \ \phi_{x^{7}}$. $3$ sources total.\\
\midrule

$\Sigma = \sigma \circ (\psi_{x^{-1}} \oplus \psi_{x^{-1}})$\newline $\psi_{x^{-1}} = (1\,8)(2\,7)(3\,6)(4\,5);\ +\,H^{\otimes n}$ &
\tH &
SWAP $\circ\, H^{\otimes 2}$&
\circwrap{\begin{quantikz}[lqstyle]
 \lstick{$1$} & \gate[style={inner sep=-1pt}]{H} & \swap{1} & \qw\\
 \lstick{$2$} & \gate[style={inner sep=-1pt}]{H} & \targX{} & \qw
\end{quantikz}} &
$\sigma\phi_{x^{2}}, \ \sigma\phi_{x^{5}}, \ \sigma\phi_{x^{8}}$. $4$ sources total.\\
\midrule

$\tau = \sigma \circ (\psi_{x^{-1}} \oplus \psi_{x^{-1}})$\newline $\psi_{x^{-1}} = (1\,8)(2\,7)(3\,6)(4\,5);\ +\!\bigotimes\! CZ$ &
\tS &
CZ &
\circwrap{\begin{quantikz}[lqstyle]
 \lstick{$1$} & \ctrl{1} & \qw\\
 \lstick{$2$} & \control{} & \qw
\end{quantikz}} &
Unique ($j{=}1$). $1$ source total.\\
\midrule

$\tau = \psi_{x^{1}} \oplus \psi_{x^{1}}$\newline $\psi_{x^1} = \mathrm{id};\ +\,S^{\otimes n}$ &
\tS &
$S^{\otimes 2}$ &
\circwrap{\begin{quantikz}[lqstyle]
 \lstick{$1$} & \gate[style={inner sep=-1pt}]{S} & \qw\\
 \lstick{$2$} & \gate[style={inner sep=-1pt}]{S} & \qw
\end{quantikz}} &
Unique ($j{=}8$). $1$ source total.\\
\midrule

$\phi_{x^{2}}\!\circ CX_{i\to i+\ell} = (\psi_{x^2} \oplus \psi_{x^2})\circ CX_{i\to i+\ell}$\newline $\psi_{x^2} = (1\,2\,4\,8\,7\,5)(3\,6)$ &
\tCX &
CNOT$(1,2)$ &
\circwrap{\begin{quantikz}[lqstyle]
 \lstick{$1$} & \ctrl{1} & \qw\\
 \lstick{$2$} & \targ{} & \qw
\end{quantikz}} &
$\phi_{x^{5}}, \ \ \phi_{x^{8}}$. $3$ sources total.\\
\midrule

$\phi_{x^{2}}\!\circ CX_{i+\ell\to i} = (\psi_{x^2} \oplus \psi_{x^2})\circ CX_{i+\ell\to i}$\newline $\psi_{x^2} = (1\,2\,4\,8\,7\,5)(3\,6)$ &
\tCX &
CNOT$(2,1)$ &
\circwrap{\begin{quantikz}[lqstyle]
 \lstick{$1$} & \targ{} & \qw\\
 \lstick{$2$} & \ctrl{-1} & \qw
\end{quantikz}} &
$\phi_{x^{5}}, \ \ \phi_{x^{8}}$. $3$ sources total.\\
\midrule

$\sigma\!\cdot\!\phi_{x^{1}}\!\circ CX_{i\to i+\ell} = \sigma\circ(\psi_{x^1} \oplus \psi_{x^1})\circ CX_{i\to i+\ell}$\newline $\psi_{x^1} = \mathrm{id};\ +\,H^{\otimes n}$ &
\tCX &
CNOT$(2,1) \circ$ CNOT$(1,2) \circ H^{\otimes 2}$ &
\circwrap{\begin{quantikz}[lqstyle]
 \lstick{$1$} & \gate[style={inner sep=-1pt}]{H} & \ctrl{1} & \targ{} & \qw\\
 \lstick{$2$} & \gate[style={inner sep=-1pt}]{H} & \targ{} & \ctrl{-1} & \qw
\end{quantikz}} &
$\sigma\phi_{x^{4}}, \ \ \sigma\phi_{x^{7}}$. $3$ sources total.\\
\midrule

$\sigma\!\cdot\!\phi_{x^{1}}\!\circ CX_{i+\ell\to i} = \sigma\circ(\psi_{x^1} \oplus \psi_{x^1})\circ CX_{i+\ell\to i}$\newline $\psi_{x^1} = \mathrm{id};\ +\,H^{\otimes n}$ &
\tCX &
CNOT$(2,1) \circ$  SWAP $\circ H^{\otimes 2}$&
\circwrap{\begin{quantikz}[lqstyle]
 \lstick{$1$} & \gate[style={inner sep=-1pt}]{H} & \swap{1} & \targ{} & \qw\\
 \lstick{$2$} & \gate[style={inner sep=-1pt}]{H} & \targX{} & \ctrl{-1} & \qw
\end{quantikz}} &
$\sigma\phi_{x^{4}}, \ \ \sigma\phi_{x^{7}}$. $3$ sources total.\\
\bottomrule
\end{tabular}}
\endgroup

\caption{Unique logical Clifford gates of the $[[18,2,5]]$ Generalized Bicycle code from Table~\ref{tab:k2-mcr}. Permutations act on the coordinates $\{0,\dots,8\}$ of
each length-$\ell$ block; $\sigma$ is the full $(i,i{+}\ell)$ block swap, $\phi_i$ a cyclic
shift, and $\phi_{x^j}=\psi_{x^j}\!\oplus\!\psi_{x^j}$ the multiplier $x\mapsto x^{j}$; $n=2\ell$
is the physical-qubit count, so $H^{\otimes n}$, $S^{\otimes n}$ and $\bigotimes CZ$ denote the
transversal single-qubit / fold layers. Logical qubit~$1$ is the top wire, qubit~$2$ the bottom;
circuits read left$\to$right (left gate first). The $36$ enumerated physical operators collapse
to these $10$ distinct logical actions. By Corollay~\ref{cor:2qubitcliffordgroup}, these gates generate the 2-qubit Clifford group. \emph{Type:} \tAut\ bare permutation automorphism
($M_\sim$); \tH\ $H$-type fold-transversal ($M_\leftrightarrow$ ZX-duality $+$ transversal $H$);
\tS\ $S$-type ($S$/$CZ$ on the fold of an $M_\leftrightarrow$ involution); \tCX\ $CX$-type
(transversal CNOT along the $(i,i{+}\ell)$ fold $+$ a repairing multiplier). Codes with $|\mathrm{Gates}|= 10$ in Table~\ref{tab:k2-mcr} have the same logical action/circuit entries as here. Codes with $|\mathrm{Gates}| = 9$ are missing the $S^{\otimes2}$ logical action.}
\label{tab:k2-logical-gates}
\end{table}

\FloatBarrier
\begingroup
\footnotesize
\renewcommand{\arraystretch}{1.1}
\setlength{\tabcolsep}{3pt}
\newcommand{\cir}[1]{\adjustbox{max width=\linewidth,max totalheight=2.9cm}{%
  \begin{quantikz}[
  row sep={0.34cm,between origins},
  column sep=0.25cm,
  operator/.append style={
      font=\fontsize{2}{3}\selectfont,
      inner xsep=0.5pt,
      inner ysep=0.25pt,
      minimum width=0.2cm,
      minimum height=0.1cm
      }
  ]#1\end{quantikz}}}

\def\gH{\gate[style={inner sep=-1pt}]{H}}
\def\gHall{\gate[6]{H^{\otimes 6}}}

\centering
\begin{longtable}{@{} >{\raggedright\arraybackslash}p{3.75cm}
                       c
                       >{\raggedright\arraybackslash}p{3.0cm}
                       >{\centering\arraybackslash}p{4.6cm}
                       >{\raggedright\arraybackslash}p{2.4cm} @{}}
\caption{Unique logical Clifford gates of the $[[30,6,5]]$ Generalized Bicycle code from ~\ref{tab:k-large-mcr}.
Permutations act on coordinates $\{0,\dots,14\}$ of each length-$\ell$ block; $\sigma$ is the full $(i,i{+}\ell)$ block swap,
$\phi_i$ a cyclic shift, $\phi_{x^j}=\psi_{x^j}\!\oplus\!\psi_{x^j}$ the multiplier $x\mapsto x^{j}$; $n=2\ell=30$, so
$H^{\otimes n}$ and $\bigotimes CZ$ denote transversal/fold layers. There are $k=6$ logical qubits ($1$=top wire);
circuits read left$\to$right (left gate first). $H^{\otimes 6}$ denotes a logical Hadamard on each of the $k=6$ logical qubits, not to be confused with $H^{\oplus n}$, a physical Hadamard on each of the physicals. $52$ physical operators collapse to these $20$ distinct logical actions.
\emph{Type:} \tAut\ permutation automorphism; \tH\ $H$-type; \tS\ $S$-type; \tCX\ $CX$-type fold-transversal.}
\label{tab:k6-logical-gates}\\
\toprule
\textbf{$\psi$ \,/\, cycle on $S_\ell$} & \textbf{Type} & \textbf{Logical Action} & \textbf{Logical Circuit} & \textbf{Other sources}\\
\midrule
\endfirsthead
\caption[]{(continued)}\\
\toprule
\textbf{$\psi$ \,/\, cycle on $S_\ell$} & \textbf{Type} & \textbf{Logical Action} & \textbf{Logical Circuit} & \textbf{Other sources}\\
\midrule
\endhead
\bottomrule
\endfoot

$\phi_3 = \psi_3\oplus\psi_3$\newline $\psi_3=(0\,3\,6\,9\,12)$\allowbreak$(1\,4\,7\,10\,13)$\allowbreak$(2\,5\,8\,11\,14)$ &
\tAut & Id &
\cir{\lstick{$1$}&\qw\\ \lstick{$2$}&\qw\\ \lstick{$3$}&\qw\\ \lstick{$4$}&\qw\\ \lstick{$5$}&\qw\\ \lstick{$6$}&\qw} &
shifts $\phi_0,\phi_6,\phi_9,\phi_{12}$; $\phi_{x^1},\phi_{x^4}$. $7$ total.\\
\midrule

$\phi_1 = \psi_1\oplus\psi_1$\newline $\psi_1:i\mapsto i{+}1$ &
\tAut & SWAP$(5,6)$ SWAP$(4,5)$ SWAP$(2,3)$ SWAP$(1,2)$ &
\cir{\lstick{$1$}&\swap{1}&\qw&\qw&\qw&\qw\\ \lstick{$2$}&\targX{}&\swap{1}&\qw&\qw&\qw\\ \lstick{$3$}&\qw&\targX{}&\qw&\qw&\qw\\ \lstick{$4$}&\qw&\qw&\swap{1}&\qw&\qw\\ \lstick{$5$}&\qw&\qw&\targX{}&\swap{1}&\qw\\ \lstick{$6$}&\qw&\qw&\qw&\targX{}&\qw} &
$\phi_4,\phi_7,\phi_{10},\phi_{13}$. $5$ total.\\
\midrule

$\phi_2 = \psi_2\oplus\psi_2$\newline $\psi_2:i\mapsto i{+}2$ &
\tAut & SWAP$(5,6)$ SWAP$(4,6)$ SWAP$(2,3)$ SWAP$(1,3)$ &
\cir{\lstick{$1$}&\swap{2}&\qw&\qw&\qw&\qw\\ \lstick{$2$}&\qw&\swap{1}&\qw&\qw&\qw\\ \lstick{$3$}&\targX{}&\targX{}&\qw&\qw&\qw\\ \lstick{$4$}&\qw&\qw&\swap{2}&\qw&\qw\\ \lstick{$5$}&\qw&\qw&\qw&\swap{1}&\qw\\ \lstick{$6$}&\qw&\qw&\targX{}&\targX{}&\qw} &
$\phi_5,\phi_8,\phi_{11},\phi_{14}$. $5$ total.\\
\midrule

$\phi_{x^{11}} = \psi_{x^{11}}\oplus\psi_{x^{11}}$\newline $\psi_{x^{11}}=(1\,11)$\allowbreak$(2\,7)$\allowbreak$(4\,14)$\allowbreak$(5\,10)$\allowbreak$(8\,13)$ &
\tAut & SWAP$(5,6)$ SWAP$(2,3)$ &
\cir{\lstick{$1$}&\qw&\qw \\ \lstick{$2$}&\swap{1}&\qw \\ \lstick{$3$}&\targX{}&\qw \\ \lstick{$4$}&\qw&\qw \\ \lstick{$5$}&\swap{1}&\qw \\ \lstick{$6$}&\targX{}&\qw  } &
$\phi_{x^{14}}$. $2$ total.\\
\midrule

$\sigma\!\cdot\!\phi_{x^{7}} = \sigma\circ(\psi_{x^{7}}\oplus\psi_{x^{7}})$\newline $\psi_{x^{7}}=(1\,7\,4\,13)$\allowbreak$(2\,14\,8\,11)$\allowbreak$(3\,6\,12\,9)$ &
\tAut & SWAP$(3,6)$ SWAP$(2,5)$ SWAP$(1,4)$ &
\cir{\lstick{$1$}&\swap{3}&\qw&\qw&\qw\\ \lstick{$2$}&\qw&\swap{3}&\qw&\qw\\ \lstick{$3$}&\qw&\qw&\swap{3}&\qw\\ \lstick{$4$}&\targX{}&\qw&\qw&\qw\\ \lstick{$5$}&\qw&\targX{}&\qw&\qw\\ \lstick{$6$}&\qw&\qw&\targX{}&\qw} &
$\sigma\phi_{x^{13}}$. $2$ total.\\
\midrule

$\sigma\!\cdot\!\phi_{x^{2}} = \sigma\circ(\psi_{x^{2}}\oplus\psi_{x^{2}})$\newline $\psi_{x^{2}}=(1\,2\,4\,8)$\allowbreak$(3\,6\,12\,9)$\allowbreak$(5\,10)$\allowbreak$(7\,14\,13\,11)$ &
\tAut & SWAP$(3,5)$ SWAP$(2,6)$ SWAP$(1,4)$ &
\cir{\lstick{$1$}&\swap{3}&\qw&\qw&\qw\\ \lstick{$2$}&\qw&\swap{4}&\qw&\qw\\ \lstick{$3$}&\qw&\qw&\swap{2}&\qw\\ \lstick{$4$}&\targX{}&\qw&\qw&\qw\\ \lstick{$5$}&\qw&\qw&\targX{}&\qw\\ \lstick{$6$}&\qw&\targX{}&\qw&\qw} &
$\sigma\phi_{x^{8}}$. $2$ total.\\
\midrule

$\phi_{x^{7}} = \psi_{x^{7}}\oplus\psi_{x^{7}}$\newline $\psi_{x^{7}}=(1\,7\,4\,13)$\allowbreak$(2\,14\,8\,11)$\allowbreak$(3\,6\,12\,9)$\newline $+\,H^{\otimes n}$ &
\tH & $H^{\otimes 6}$ &
\cir{\lstick{$1$}&\gHall&\qw\\ \lstick{$2$}&&\qw\\ \lstick{$3$}&&\qw\\ \lstick{$4$}&&\qw\\ \lstick{$5$}&&\qw\\ \lstick{$6$}&&\qw} &
$\phi_{x^{13}}$. $2$ total.\\
\midrule

$\phi_{x^{2}} = \psi_{x^{2}}\oplus\psi_{x^{2}}$\newline $\psi_{x^{2}}=(1\,2\,4\,8)$\allowbreak$(3\,6\,12\,9)$\allowbreak$(5\,10)$\allowbreak$(7\,14\,13\,11)$\newline $+\,H^{\otimes n}$ &
\tH & SWAP$(5,6)$ SWAP$(2,3)\;H^{\otimes 6}$ &
\cir{\lstick{$1$}&\gHall&\qw&\qw\\ \lstick{$2$}&&\swap{1}&\qw\\ \lstick{$3$}&&\targX{}&\qw\\ \lstick{$4$}&&\qw&\qw\\ \lstick{$5$}&&\swap{1}&\qw\\ \lstick{$6$}&&\targX{}&\qw} &
$\phi_{x^{8}}$. $2$ total.\\
\midrule

$\Sigma = \sigma\circ(\psi_{x^{-1}}\oplus\psi_{x^{-1}})$\newline $\psi_{x^{-1}}=(1\,14)$\allowbreak$(2\,13)$\allowbreak$(3\,12)$\allowbreak$(4\,11)$\allowbreak$(5\,10)$\allowbreak$(6\,9)$\allowbreak$(7\,8)$\newline $+\,H^{\otimes n}$ &
\tH & SWAP$(3,5)$ SWAP$(2,6)$ SWAP$(1,4)\;H^{\otimes 6}$ &
\cir{\lstick{$1$}&\gHall&\swap{3}&\qw&\qw&\qw\\ \lstick{$2$}&&\qw&\swap{4}&\qw&\qw\\ \lstick{$3$}&&\qw&\qw&\swap{2}&\qw\\ \lstick{$4$}&&\targX{}&\qw&\qw&\qw\\ \lstick{$5$}&&\qw&\qw&\targX{}&\qw\\ \lstick{$6$}&&\qw&\targX{}&\qw&\qw} &
$\sigma\phi_{x^{11}},\sigma\phi_{x^{14}}$. $3$ total.\\
\midrule

$\sigma\!\cdot\!\phi_{x^{1}} = \sigma\circ(\psi_{x^{1}}\oplus\psi_{x^{1}})$\newline $\psi_{x^{1}}=\mathrm{id};\ +\,H^{\otimes n}$ &
\tH & SWAP$(3,6)$ SWAP$(2,5)$ SWAP$(1,4)\;H^{\otimes 6}$ &
\cir{\lstick{$1$}&\gHall&\swap{3}&\qw&\qw&\qw\\ \lstick{$2$}&&\qw&\swap{3}&\qw&\qw\\ \lstick{$3$}&&\qw&\qw&\swap{3}&\qw\\ \lstick{$4$}&&\targX{}&\qw&\qw&\qw\\ \lstick{$5$}&&\qw&\targX{}&\qw&\qw\\ \lstick{$6$}&&\qw&\qw&\targX{}&\qw} &
$\sigma\phi_{x^{4}}$. $2$ total.\\
\midrule

$\tau = \sigma\circ(\psi_{x^{-1}}\oplus\psi_{x^{-1}})$\;($j{=}1$)\newline $\psi_{x^{-1}}=(1\,14)$\allowbreak$(2\,13)$\allowbreak$(3\,12)$\allowbreak$(4\,11)$\allowbreak$(5\,10)$\allowbreak$(6\,9)$\allowbreak$(7\,8)$\newline $+\!\bigotimes\! CZ$ &
\tS & CZ$(3,5)$ CZ$(2,6)$ CZ$(1,4)$ &
\cir{\lstick{$1$}&\ctrl{3}&\qw&\qw&\qw\\ \lstick{$2$}&\qw&\ctrl{4}&\qw&\qw\\ \lstick{$3$}&\qw&\qw&\ctrl{2}&\qw\\ \lstick{$4$}&\control{}&\qw&\qw&\qw\\ \lstick{$5$}&\qw&\qw&\control{}&\qw\\ \lstick{$6$}&\qw&\control{}&\qw&\qw} &
$S$-fold $j{=}4$. $2$ total.\\
\midrule

$\tau = \sigma\circ(\psi_{x^{4}}\oplus\psi_{x^{4}})$\;($j{=}11$)\newline $\psi_{x^{4}}=(1\,4)$\allowbreak$(2\,8)$\allowbreak$(3\,12)$\allowbreak$(6\,9)$\allowbreak$(7\,13)$\allowbreak$(11\,14)$\newline $+\!\bigotimes\! CZ$ &
\tS & CZ$(3,6)$ CZ$(2,5)$ CZ$(1,4)$ &
\cir{\lstick{$1$}&\ctrl{3}&\qw&\qw&\qw\\ \lstick{$2$}&\qw&\ctrl{3}&\qw&\qw\\ \lstick{$3$}&\qw&\qw&\ctrl{3}&\qw\\ \lstick{$4$}&\control{}&\qw&\qw&\qw\\ \lstick{$5$}&\qw&\control{}&\qw&\qw\\ \lstick{$6$}&\qw&\qw&\control{}&\qw} &
$S$-fold $j{=}14$. $2$ total.\\
\midrule

$\phi_{x^{7}}\!\circ CX_{i\to i+\ell} = (\psi_{x^{7}}\oplus\psi_{x^{7}})\circ CX_{i\to i+\ell}$\newline $\psi_{x^{7}}=(1\,7\,4\,13)$\allowbreak$(2\,14\,8\,11)$\allowbreak$(3\,6\,12\,9)$ &
\tCX & CX$(3,6)$ CX$(2,5)$ CX$(1,4)$ &
\cir{\lstick{$1$}&\ctrl{3}&\qw&\qw&\qw\\ \lstick{$2$}&\qw&\ctrl{3}&\qw&\qw\\ \lstick{$3$}&\qw&\qw&\ctrl{3}&\qw\\ \lstick{$4$}&\targ{}&\qw&\qw&\qw\\ \lstick{$5$}&\qw&\targ{}&\qw&\qw\\ \lstick{$6$}&\qw&\qw&\targ{}&\qw} &
$\phi_{x^{13}}$. $2$ total.\\
\midrule

$\phi_{x^{2}}\!\circ CX_{i\to i+\ell} = (\psi_{x^{2}}\oplus\psi_{x^{2}})\circ CX_{i\to i+\ell}$\newline $\psi_{x^{2}}=(1\,2\,4\,8)$\allowbreak$(3\,6\,12\,9)$\allowbreak$(5\,10)$\allowbreak$(7\,14\,13\,11)$ &
\tCX & SWAP$(5,6)$ CX$(3,5)$ CX$(2,6)$ SWAP$(2,3)$ CX$(1,4)$ &
\cir{\lstick{$1$}&\ctrl{3}&\qw&\qw&\qw&\qw&\qw\\ \lstick{$2$}&\qw&\swap{1}&\ctrl{4}&\qw&\qw&\qw\\ \lstick{$3$}&\qw&\targX{}&\qw&\ctrl{2}&\qw&\qw\\ \lstick{$4$}&\targ{}&\qw&\qw&\qw&\qw&\qw\\ \lstick{$5$}&\qw&\qw&\qw&\targ{}&\swap{1}&\qw\\ \lstick{$6$}&\qw&\qw&\targ{}&\qw&\targX{}&\qw} &
$\phi_{x^{8}}$. $2$ total.\\
\midrule

$\phi_{x^{7}}\!\circ CX_{i+\ell\to i} = (\psi_{x^{7}}\oplus\psi_{x^{7}})\circ CX_{i+\ell\to i}$\newline $\psi_{x^{7}}=(1\,7\,4\,13)$\allowbreak$(2\,14\,8\,11)$\allowbreak$(3\,6\,12\,9)$ &
\tCX & CX$(6,3)$ CX$(5,2)$ CX$(4,1)$ &
\cir{\lstick{$1$}&\targ{}&\qw&\qw&\qw\\ \lstick{$2$}&\qw&\targ{}&\qw&\qw\\ \lstick{$3$}&\qw&\qw&\targ{}&\qw\\ \lstick{$4$}&\ctrl{-3}&\qw&\qw&\qw\\ \lstick{$5$}&\qw&\ctrl{-3}&\qw&\qw\\ \lstick{$6$}&\qw&\qw&\ctrl{-3}&\qw} &
$\phi_{x^{13}}$. $2$ total.\\
\midrule

$\phi_{x^{2}}\!\circ CX_{i+\ell\to i} = (\psi_{x^{2}}\oplus\psi_{x^{2}})\circ CX_{i+\ell\to i}$\newline $\psi_{x^{2}}=(1\,2\,4\,8)$\allowbreak$(3\,6\,12\,9)$\allowbreak$(5\,10)$\allowbreak$(7\,14\,13\,11)$ &
\tCX & CX$(6,3)$ CX$(5,2)$ SWAP$(5,6)$ CX$(4,1)$ SWAP$(2,3)$ &
\cir{\lstick{$1$}&\qw&\targ{}&\qw&\qw&\qw&\qw\\ \lstick{$2$}&\swap{1}&\qw&\qw&\targ{}&\qw&\qw\\ \lstick{$3$}&\targX{}&\qw&\qw&\qw&\targ{}&\qw\\ \lstick{$4$}&\qw&\ctrl{-3}&\qw&\qw&\qw&\qw\\ \lstick{$5$}&\qw&\qw&\swap{1}&\ctrl{-3}&\qw&\qw\\ \lstick{$6$}&\qw&\qw&\targX{}&\qw&\ctrl{-3}&\qw} &
$\phi_{x^{8}}$. $2$ total.\\
\midrule

$\sigma\!\cdot\!\phi_{x^{1}}\!\circ CX_{i\to i+\ell} = \sigma\circ(\psi_{x^{1}}\oplus\psi_{x^{1}})\circ CX_{i\to i+\ell}$\newline $\psi_{x^{1}}=\mathrm{id};\ +\,H^{\otimes n}$ &
\tCX & CX$(6,3)$ CX$(5,2)$ CX$(4,1)$ CX$(3,6)$ CX$(2,5)$ CX$(1,4)\;H^{\otimes 6}$ &
\cir{\lstick{$1$}&\gHall&\ctrl{3}&\qw&\qw&\targ{}&\qw&\qw&\qw\\ \lstick{$2$}&&\qw&\ctrl{3}&\qw&\qw&\targ{}&\qw&\qw\\ \lstick{$3$}&&\qw&\qw&\ctrl{3}&\qw&\qw&\targ{}&\qw\\ \lstick{$4$}&&\targ{}&\qw&\qw&\ctrl{-3}&\qw&\qw&\qw\\ \lstick{$5$}&&\qw&\targ{}&\qw&\qw&\ctrl{-3}&\qw&\qw\\ \lstick{$6$}&&\qw&\qw&\targ{}&\qw&\qw&\ctrl{-3}&\qw} &
$\sigma\phi_{x^{4}}$. $2$ total.\\
\midrule

$\sigma\!\cdot\!\phi_{x^{11}}\!\circ CX_{i\to i+\ell} = \sigma\circ(\psi_{x^{11}}\oplus\psi_{x^{11}})\circ CX_{i\to i+\ell}$\newline $\psi_{x^{11}}=(1\,11)$\allowbreak$(2\,7)$\allowbreak$(4\,14)$\allowbreak$(5\,10)$\allowbreak$(8\,13);\ +\,H^{\otimes n}$ &
\tCX & CX$(6,3)$ CX$(5,2)$ SWAP$(5,6)$ CX$(4,1)$ CX$(3,5)$ CX$(2,6)$ SWAP$(2,3)$ CX$(1,4)\;H^{\otimes 6}$ &
\cir{\lstick{$1$}&\gHall&\ctrl{3}&\qw&\qw&\qw&\targ{}&\qw&\qw&\qw&\qw\\ \lstick{$2$}&&\qw&\swap{1}&\ctrl{4}&\qw&\qw&\qw&\targ{}&\qw&\qw\\ \lstick{$3$}&&\qw&\targX{}&\qw&\ctrl{2}&\qw&\qw&\qw&\targ{}&\qw\\ \lstick{$4$}&&\targ{}&\qw&\qw&\qw&\ctrl{-3}&\qw&\qw&\qw&\qw\\ \lstick{$5$}&&\qw&\qw&\qw&\targ{}&\qw&\swap{1}&\ctrl{-3}&\qw&\qw\\ \lstick{$6$}&&\qw&\qw&\targ{}&\qw&\qw&\targX{}&\qw&\ctrl{-3}&\qw} &
$\sigma\phi_{x^{14}}$. $2$ total.\\
\midrule

$\sigma\!\cdot\!\phi_{x^{1}}\!\circ CX_{i+\ell\to i} = \sigma\circ(\psi_{x^{1}}\oplus\psi_{x^{1}})\circ CX_{i+\ell\to i}$\newline $\psi_{x^{1}}=\mathrm{id};\ +\,H^{\otimes n}$ &
\tCX & CX$(6,3)$ CX$(5,2)$ CX$(4,1)$ SWAP$(3,6)$ SWAP$(2,5)$ SWAP$(1,4)\;H^{\otimes 6}$ &
\cir{\lstick{$1$}&\gHall&\swap{3}&\qw&\qw&\targ{}&\qw&\qw&\qw\\ \lstick{$2$}&&\qw&\swap{3}&\qw&\qw&\targ{}&\qw&\qw\\ \lstick{$3$}&&\qw&\qw&\swap{3}&\qw&\qw&\targ{}&\qw\\ \lstick{$4$}&&\targX{}&\qw&\qw&\ctrl{-3}&\qw&\qw&\qw\\ \lstick{$5$}&&\qw&\targX{}&\qw&\qw&\ctrl{-3}&\qw&\qw\\ \lstick{$6$}&&\qw&\qw&\targX{}&\qw&\qw&\ctrl{-3}&\qw} &
$\sigma\phi_{x^{4}}$. $2$ total.\\
\midrule

$\sigma\!\cdot\!\phi_{x^{11}}\!\circ CX_{i+\ell\to i} = \sigma\circ(\psi_{x^{11}}\oplus\psi_{x^{11}})\circ CX_{i+\ell\to i}$\newline $\psi_{x^{11}}=(1\,11)$\allowbreak$(2\,7)$\allowbreak$(4\,14)$\allowbreak$(5\,10)$\allowbreak$(8\,13);\ +\,H^{\otimes n}$ &
\tCX & CX$(6,3)$ CX$(5,2)$ CX$(4,1)$ SWAP$(3,5)$ SWAP$(2,6)$ SWAP$(1,4)\;H^{\otimes 6}$ &
\cir{\lstick{$1$}&\gHall&\swap{3}&\qw&\qw&\targ{}&\qw&\qw&\qw\\ \lstick{$2$}&&\qw&\swap{4}&\qw&\qw&\targ{}&\qw&\qw\\ \lstick{$3$}&&\qw&\qw&\swap{2}&\qw&\qw&\targ{}&\qw\\ \lstick{$4$}&&\targX{}&\qw&\qw&\ctrl{-3}&\qw&\qw&\qw\\ \lstick{$5$}&&\qw&\qw&\targX{}&\qw&\ctrl{-3}&\qw&\qw\\ \lstick{$6$}&&\qw&\targX{}&\qw&\qw&\qw&\ctrl{-3}&\qw} &
$\sigma\phi_{x^{14}}$. $2$ total.\\
\end{longtable}
\endgroup
\FloatBarrier

\section{Conclusion and Future Directions}

This work demonstrates that the rich algebraic machinery developed over decades for classical cyclic codes tying together the polynomial ring, circulant matrix, and the coordinate spaces extends naturally to the quantum setting through Generalized Bicycle codes. By recognizing that the stabilizer rowspaces of a GB code are cyclic submodules of $R_\ell^2$, we inherit an analogous three-space dependency for quantum codes, one in which ring automorphisms of $R_\ell$ lift directly to code automorphisms on the physical qubit space. This shifts the search for automorphisms, fold-transversal gates, and logical actions from a brute-force combinatorial problem on $S_{2\ell}$ to a tractable algebraic problem on the defining polynomials $f, p, q$ and the block length $\ell$. The Maximal Cube Root construction of Sec.~\ref{sec:MaxCubeRoot} serves as a proof of concept of this framework; by imposing a single algebraic identity, $r^2 + r + 1 = 0$ with $r = pq^{-1}$, we can guarantee, by construction, a large and well-characterized automorphism group together with an accompanying set of fold-transversal gates.

We emphasize, however, that this paper is a first step rather than a complete program. The inverse design demonstrated here, building a code around a rich automorphism structure, prioritizes algebraic simplicity at the cost of distance and stabilizer weight, a tradeoff that is not useful in practice. The Maximal Cube Root codes we present should be viewed as illustrative rather than practically optimal, and designing a code around desired automorphisms while simultaneously controlling distance and locality is an important next step. We hope the framework here is useful for that effort.

A few directions strike us as particularly interesting:
\begin{itemize}

    \item \textbf{Beyond $\bm \psi_L \bm = \bm \psi_R$.} In the MCR family, addressability when $k$ gets large is prohibited by the symmetric nature of the action on each half of the code imposed by $\psi_L = \psi_R$. Can addressable gates be achieved from affine $\psi_L, \psi_R$ where $\psi_L \neq \psi_R$? Or, moving beyond affine permutations entirely, what automorphism structure is yielded from choosing classical cyclic codes that are known to have degenerate structure allowing non-affine permutations, such as quadratic residue and Golay codes, or simplex and Hamming codes? Finally, can we better characterize when block-mixing automorphisms are allowed?
    
    \item \textbf{Extending the three-space dependency to broader code families.} The central technical idea of this work yields a tangible way to assess the automorphism structure of GB codes beyond the guaranteed automorphisms of any 2BGA code as outlined in \cite{automorphisms_BBcodes}. This central idea is not inherently restricted to GB codes; it seems natural that a quantum code built from other ring-theoretic ingredients could inherit a three-space dependency in which ring automorphisms become code automorphisms. It would be interesting to investigate this possibility for broader families of 2GBA codes

    \item \textbf{Distance, stabilizer weight, automorphism co-design.} The analysis of this work strictly pertains to how to assess and design the automorphism structure of GB codes, and pays no attention to the resulting distance or stabilizer weight of a given code. It would be fascinating to see if viewing the structure of GB codes as cyclic submodules yields any insight into the distance or stabilizer weight of a given code, or if any results pertaining to distance from classical coding theory become applicable.

    \item \textbf{Targeted automorphism gate sets.} The MCR construction leverages properties of $f, p, q, \ell$ that yield a large number of multiplier and CX automorphisms; however, this construction says nothing about the actual logical action of any given automorphism. The extent to which this formalism and the canonical logical operators derived in Sec.~\ref{sec:LogicalOps} can be used to derive codes with automorphisms that can enact specific logical gates remains unexplored.
\end{itemize}

More broadly, the perspective advocated here is that classical cyclic coding theory is not merely a historical antecedent to quantum coding theory but an active algebraic toolkit whose ring-theoretic and module-theoretic machinery have direct quantum analogues. We view the three-space dependency as a unifying lens through which many existing and future constructions may be profitably re-examined, and we hope that the conditions, constructions, and conjectures presented here seed a broader effort to design quantum codes, with their full complement of logical gates, from the algebra outward.

\section*{Acknowledgments}

A.D. acknowledges support of the  MIT Jacobs Family Presidential Fellowship. J.B. acknowledges support from the MIT Center for Quantum Engineering / Laboratory for Physical Sciences Doc Bedard Fellowship. We thank Zhiyang He for helpful discussions and reviewer comments, and Noah Berthusen and Drew Potter for helpful discussions about fold CNOTs. This work made use of resources provided by SubMIT at MIT Physics \cite{submit_cluster}.

\bibliographystyle{unsrtnat}
\bibliography{bibliography}

@article{breuckmann2021balancedproduct,
	title        = {{Balanced Product Quantum Codes}},
	author       = {Breuckmann, Nikolas P. and Eberhardt, Jens N.},
	year         = 2021,
	month        = oct,
	journal      = {IEEE Transactions on Information Theory},
	volume       = 67,
	number       = 10,
	pages        = {6653--6674},
	url          = {https://ieeexplore.ieee.org/document/9490244/}
}

@article{kovalev2013quantumkronecker,
	title        = {{Quantum Kronecker sum-product low-density parity-check codes with finite rate}},
	author       = {Kovalev, Alexey A. and Pryadko, Leonid P.},
	year         = 2013,
	month        = jul,
	journal      = {Physical Review A},
	publisher    = {American Physical Society},
	volume       = 88,
	number       = 1,
	pages        = {012311},
	url          = {https://link.aps.org/doi/10.1103/PhysRevA.88.012311}
}

@article{panteleev2022quantumldpc,
	title        = {{Quantum {LDPC} Codes With Almost Linear Minimum Distance}},
	author       = {Panteleev, Pavel and Kalachev, Gleb},
	year         = 2022,
	month        = jan,
	journal      = {IEEE Transactions on Information Theory},
	volume       = 68,
	number       = 1,
	pages        = {213--229},
	url          = {https://ieeexplore.ieee.org/document/9567703},
	urldate      = {2025-12-17}
}

@misc{yoder2025tourgrossmodularquantum,
	title        = {{Tour de gross: A modular quantum computer based on bivariate bicycle codes}},
	author       = {Yoder, Theodore J. and Schoute, Eddie and Rall, Patrick and Pritchett, Emily and Gambetta, Jay M. and Cross, Andrew W. and Carroll, Malcolm and Beverland, Michael E.},
	year         = 2025,
	url          = {https://arxiv.org/abs/2506.03094},
	eprint       = {2506.03094},
	archivePrefix = {arXiv},
	primaryClass = {quant-ph}
}

@article{Breuckmann_foldtransversal,
	title        = {{Fold-Transversal Clifford Gates for Quantum Codes}},
	author       = {Breuckmann, Nikolas P. and Burton, Simon},
	year         = 2024,
	month        = jun,
	journal      = {Quantum},
	publisher    = {Verein zur Forderung des Open Access Publizierens in den Quantenwissenschaften},
	volume       = 8,
	pages        = 1372,
	url          = {https://quantum-journal.org/papers/q-2024-06-13-1372/}
}

@misc{wang2022distanceboundsgeneralizedbicycle,
	title        = {{Distance bounds for generalized bicycle codes}},
	author       = {Wang, Renyu and Pryadko, Leonid P.},
	year         = 2022,
	url          = {https://arxiv.org/abs/2203.17216},
	eprint       = {2203.17216},
	archivePrefix = {arXiv},
	primaryClass = {quant-ph}
}

@misc{parsimonioussurgery,
	title        = {{Parsimonious Quantum Low-Density Parity-Check Code Surgery}},
	author       = {Yuan, Andrew C. and Cowtan, Alexander and He, Zhiyang and Lin, Ting-Chun and Williamson, Dominic J.},
	year         = 2026,
    journal      = {arXiv preprint arXiv:2603.05082},
	url          = {https://arxiv.org/abs/2603.05082},
	archivePrefix = {arXiv},
	primaryClass = {quant-ph}
}

@article{Panteleev_2021,
	title        = {{Degenerate Quantum {LDPC} Codes With Good Finite Length Performance}},
	author       = {Panteleev, Pavel and Kalachev, Gleb},
	year         = 2021,
	month        = nov,
	journal      = {Quantum},
	publisher    = {Verein zur Forderung des Open Access Publizierens in den Quantenwissenschaften},
	volume       = 5,
	pages        = 585,
	url          = {https://quantum-journal.org/papers/q-2021-11-22-585/pdf/}
}

@article{sayginel2025faulttolerantlogicalcliffordgates,
  title          = {Fault-Tolerant Logical Clifford Gates from Code Automorphisms},
  author         = {Sayginel, Hasan and Koutsioumpas, Stergios and Webster, Mark and Rajput, Abhishek and Browne, Dan E.},
  journal        = {PRX Quantum},
  volume         = {6},
  issue          = {3},
  pages          = {030343},
  numpages       = {24},
  year           = {2025},
  month          = {Sep},
  publisher      = {American Physical Society},
  url            = {https://link.aps.org/doi/10.1103/vf7v-cpq9}
}

@misc{phantomcodes,
	title        = {{Entangling logical qubits without physical operations}},
	author       = {Koh, Jin Ming and Gong, Anqi and Diaconu, Andrei C. and Tan, Daniel Bochen and Geim, Alexandra A. and Gullans, Michael J. and Yao, Norman Y. and Lukin, Mikhail D. and Majidy, Shayan},
	year         = 2026,
	url          = {https://arxiv.org/abs/2601.20927},
	journal      = {arXiv preprint arXiv:2601.20927},
	primaryClass = {quant-ph}
}

@misc{ExhaustiveAtomorphisms,
	title        = {{Exhaustive Optimisation of Automorphism Groups for Stabiliser Codes}},
	author       = {Mac Aree, Aisling and Howard, Mark},
	year         = 2026,
	url          = {https://arxiv.org/abs/2604.01282},
	eprint       = {2604.01282},
	archivePrefix = {arXiv},
	primaryClass = {quant-ph}
}

@article{Original_BBpaper,
	title        = {{High-threshold and low-overhead fault-tolerant quantum memory}},
	author       = {Bravyi, Sergey and Cross, Andrew W. and Gambetta, Jay M. and Maslov, Dmitri and Rall, Patrick and Yoder, Theodore J.},
	year         = 2024,
	month        = mar,
	journal      = {Nature},
	publisher    = {Springer Science and Business Media LLC},
	volume       = 627,
	number       = 8005,
	pages        = {778--782},
	url          = {https://www.nature.com/articles/s41586-024-07107-7}
}

@misc{Zhao_highrate,
	title        = {{Towards Ultra-High-Rate Quantum Error Correction with Reconfigurable Atom Arrays}},
	author       = {Zhao, Chen and Duckering, Casey and Gu, Andi and Maskara, Nishad and Zhou, Hengyun},
	year         = 2026,
	url          = {https://arxiv.org/abs/2604.16209},
	journal      = {arXiv preprint arXiv:2604.16209},
	primaryClass = {quant-ph}
}

@misc{automorphisms_BBcodes,
	title        = {{Logical Operators and Fold-Transversal Gates of Bivariate Bicycle Codes}},
	author       = {Eberhardt, Jens Niklas and Steffan, Vincent},
	year         = 2024,
	url          = {https://arxiv.org/abs/2407.03973},
	eprint       = {2407.03973},
	archivePrefix = {arXiv},
	primaryClass = {quant-ph}
}

@misc{feng2026permutationautomorphismgroupsirreducible,
	title        = {{The permutation automorphism groups of irreducible cyclic codes}},
	author       = {Feng, Tao and Hollmann, Henk D. L. and Li, Weicong and Xiang, Qing},
	year         = 2026,
	url          = {https://arxiv.org/abs/2603.01904},
	eprint       = {arXiv preprint arXiv:2603.01904},
	primaryClass = {math.CO}
}

@article{BergerCharpin,
	title        = {{The Automorphism Groups of {BCH} Codes and of Some Affine-Invariant Codes Over Extension Fields}},
	author       = {Berger, Thierry and Charpin, Pascale},
	year         = 1999,
	month        = dec,
	journal      = {Designs, Codes and Cryptography},
	publisher    = {Springer Science and Business Media LLC},
	volume       = 18,
	pages        = {29--53},
	url          = {https://dl.acm.org/doi/abs/10.1023/A%3A1008372800005}
}

@book{dummit2003abstract,
	title        = {{Abstract Algebra}},
	author       = {Dummit, D.S. and Foote, R.M.},
	year         = 2003,
	publisher    = {Wiley},
	isbn         = {9780471433347},
	lccn         = {2003057652}
}

@article{Abdukhalikov_2026,
	title        = {{Quasi-cyclic codes of index 2}},
	author       = {Abdukhalikov, Kanat and Dzhumadil'daev, Askar S. and Ling, San},
	year         = 2026,
	month        = jun,
	journal      = {Discrete Mathematics},
	publisher    = {Elsevier BV},
	volume       = 349,
	number       = 6,
	pages        = 115004,
	url          = {https://arxiv.org/abs/2504.00568}
}

@article{dastbasteh2023polynomialrepresentationadditivecyclic,
  title          = {Polynomial representation of additive cyclic codes and new quantum codes},
  author         = {Reza Dastbasteh and Khalil Shivji},
  journal        = {Adv. Math. Commun.},
  year           = {2023},
  volume         = {19},
  pages          = {49-68},
  url            = {https://api.semanticscholar.org/CorpusID:255372576}
}

@article{QCLally,
	title        = {{Quasicyclic codes of index $\ell$ over $\bbF_q$ viewed as $\bbF_q[x]$-submodules of $\bbF_{q^\ell}[x]/\langle x^m-1\rangle$}},
	author       = {Lally, Kristine},
	year         = 2003,
	journal      = {LCNS: Applied algebra, algebraic algorithms and error-correcting codes},
	publisher    = {Springer Berlin Heidelberg},
	volume       = 2643,
	pages        = {244--253},
    url          = {https://link.springer.com/chapter/10.1007/3-540-44828-4_26},
}

@article{martinez2017codesaffinealgebrasfinite,
    author       = {Mart\'{i}nez-Moro, E. and Pi\~{n}era-Nicol\'{a}s, A. and R\'{u}a, I.F.},
    title        = {Codes over affine algebras with a finite commutative chain coefficient ring},
    year         = {2018},
    issue_date   = {January 2018},
    publisher    = {Elsevier Science Publishers B. V.},
    address      = {NLD},
    volume       = {49},
    number       = {C},
    url          = {https://dl.acm.org/doi/10.1016/j.ffa.2017.09.008},
    journal      = {Finite Fields Appl.},
    month        = jan,
    pages        = {94–107},
    numpages     = {14},

}

@article{cyclic_negacyclic_chainrings,
	title        = {{Cyclic and negacyclic codes over finite chain rings}},
	author       = {Dinh, Hai Quang and Lopez-Permouth, S.R.},
	year         = 2004,
	journal      = {IEEE Transactions on Information Theory},
	volume       = 50,
	number       = 8,
	pages        = {1728--1744},
    url          = {https://ieeexplore.ieee.org/document/1317117}
}

@article{norton2000structure,
	title        = {{On the Structure of Linear and Cyclic Codes over a Finite Chain Ring}},
	author       = {Norton, Graham H. and S{\u{a}}l{\u{a}}gean, Ana},
	year         = 2000,
	journal      = {Applicable Algebra in Engineering, Communication and Computing},
	publisher    = {Springer},
	volume       = 10,
	number       = 6,
	pages        = {489--506},
    url          = {https://link.springer.com/article/10.1007/PL00012382},
}

@article{Palfy1987,
	title        = {{Isomorphism problem for relational structures with a cyclic automorphism}},
	author       = {P\'alfy, P. P.},
	year         = 1987,
	journal      = {European Journal of Combinatorics},
	volume       = 8,
	number       = 1,
	pages        = {35--43},
    url          = {https://www.sciencedirect.com/science/article/pii/S0195669887800185},
}

@article{HuffmanJobPless1993,
	title        = {{Multipliers and generalized multipliers of cyclic objects and cyclic codes}},
	author       = {Huffman, W. C. and Job, V. and Pless, V.},
	year         = 1993,
	journal      = {Journal of Combinatorial Theory, Series A},
	volume       = 62,
	number       = 2,
	pages        = {183--215},
    url          = {https://dl.acm.org/doi/10.1016/0097-3165%2893%2990043-8}
}

@article{Guenda2010,
	title        = {{The permutation groups and the equivalence of cyclic and quasi-cyclic codes}},
	author       = {Guenda, Kenza},
	year         = 2010,
	journal      = {arXiv preprint},
	eprint       = {1002.2456},
    url          = {https://arxiv.org/abs/1002.2456}
}

@article{DastbastehLisonek2022,
    author       = {Dastbasteh, Reza and Lison\v{e}k, Petr},
    title        = {On the equivalence of linear cyclic and constacyclic codes},
    year         = {2023},
    issue_date   = {Sep 2023},
    publisher    = {Elsevier Science Publishers B. V.},
    address      = {NLD},
    volume       = {346},
    number       = {9},
    url          = {https://dl.acm.org/doi/10.1016/j.disc.2023.113489},
    journal      = {Discrete Math.},
    month        = sep,
    numpages     = {16},
}

@book{huffmanpless,
	title        = {{Fundamentals of Error-Correcting Codes}},
	author       = {Huffman, W. Cary and Pless, Vera},
	year         = 2003,
	publisher    = {Cambridge University Press},
	address      = {Cambridge, UK},
	isbn         = {9780521782807}
}

@phdthesis{gottesman_thesis,
	title        = {{Stabilizer Codes and Quantum Error Correction}},
	author       = {Gottesman, Daniel},
	year         = 1997,
	school       = {California Institute of Technology},
	address      = {Pasadena, California},
	url          = {https://arxiv.org/abs/quant-ph/9705052},
	eprint       = {quant-ph/9705052},
	archivePrefix = {arXiv}
}

@misc{gurobi,
	title        = {{Gurobi Optimizer Reference Manual}},
	author       = {{Gurobi Optimization, LLC}},
	year         = 2026,
	url          = {https://www.gurobi.com}
}

@article{og_HGP,
	title        = {{Quantum {LDPC} Codes With Positive Rate and Minimum Distance Proportional to the Square Root of the Blocklength}},
	author       = {Tillich, Jean-Pierre and Zemor, Gilles},
	year         = 2014,
	month        = feb,
	journal      = {IEEE Transactions on Information Theory},
	publisher    = {Institute of Electrical and Electronics Engineers (IEEE)},
	volume       = 60,
	number       = 2,
	pages        = {1193--1202},
	url          = {http://dx.doi.org/10.1109/TIT.2013.2292061}
}

@incollection{eczoo_hgp,
	title        = {{Hypergraph product ({HGP}) code}},
	author       = {},
	year         = 2024,
	booktitle    = {The Error Correction Zoo},
	url          = {https://errorcorrectionzoo.org/c/hypergraph_product},
	editor       = {Albert, Victor V. and Faist, Philippe}
}

@inproceedings{PK_asmyptoticallygood,
	title        = {{Asymptotically good Quantum and locally testable classical {LDPC} codes}},
	author       = {Panteleev, Pavel and Kalachev, Gleb},
	year         = 2022,
	booktitle    = {Proceedings of the 54th Annual ACM SIGACT Symposium on Theory of Computing},
	publisher    = {Association for Computing Machinery},
	address      = {New York, NY, USA},
	series       = {STOC 2022},
	pages        = {375--388},
	url          = {https://dl.acm.org/doi/10.1145/3519935.3520017},
	numpages     = 14,
	location     = {Rome, Italy}
}

@article{qianzhao_constantoverhead,
	title        = {{Constant-overhead fault-tolerant quantum computation with reconfigurable atom arrays}},
	author       = {Xu, Qian and Bonilla Ataides, J. Pablo and Pattison, Christopher A. and Raveendran, Nithin and Bluvstein, Dolev and Wurtz, Jonathan and Vasi\'{c}, Bane and Lukin, Mikhail D. and Jiang, Liang and Zhou, Hengyun},
	year         = 2024,
	journal      = {Nature Physics},
	volume       = 20,
	number       = 7,
	pages        = {1084--1090},
	url          = {https://www.nature.com/articles/s41567-024-02479-z}
}

@misc{kasai2026,
	title        = {{Breaking the Orthogonality Barrier in Quantum {LDPC} Codes}},
	author       = {Kasai, Kenta},
	year         = 2026,
	url          = {https://arxiv.org/abs/2601.08824},
	journal       = {arXiv preprint arXiv:2601.08824},
	archivePrefix = {arXiv},
	primaryClass = {quant-ph}
}

@article{cohen2022low-overhead,
	title        = {{Low-overhead fault-tolerant quantum computing using long-range connectivity}},
	author       = {Cohen,  Lawrence Z. and Kim,  Isaac H. and Bartlett,  Stephen D. and Brown,  Benjamin J.},
	year         = 2022,
	month        = {May},
	journal      = {Science Advances},
	publisher    = {American Association for the Advancement of Science (AAAS)},
	volume       = 8,
	number       = 20,
	issn         = {2375-2548},
	url          = {https://www.science.org/doi/10.1126/sciadv.abn1717}
}

@article{cross2024improved,
	title        = {{Improved {QLDPC} Surgery: Logical Measurements and Bridging Codes}},
	author       = {Andrew Cross and Zhiyang He and Patrick Rall and Theodore Yoder},
	year         = {2024},
	journal      = {arXiv preprint arXiv:2407.18393},
	pages        = {},
	url          = {https://arxiv.org/abs/2407.18393}
}

@article{williamson2024low-overhead,
	title        = {{Low-overhead fault-tolerant quantum computation by gauging logical operators}},
	author       = {Williamson, Dominic J. and Yoder, Theodore J.},
	year         = 2026,
	journal      = {Nature Physics},
	volume       = 22,
	number       = 4,
	pages        = {598--603},
	url          = {https://www.nature.com/articles/s41567-026-03220-8}
}

@article{Ide_2025,
	title        = {{Fault-Tolerant Logical Measurements via Homological Measurement}},
	author       = {Ide, Benjamin and Gowda, Manoj G. and Nadkarni, Priya J. and Dauphinais, Guillaume},
	year         = 2025,
	month        = jun,
	journal      = {Physical Review X},
	publisher    = {American Physical Society (APS)},
	volume       = 15,
	number       = 2,
	url          = {https://journals.aps.org/prx/abstract/10.1103/PhysRevX.15.021088}
}

@article{swaroop2024universal,
	title        = {{Universal Adapters between Quantum Low-Density Parity Check Codes}},
	author       = {Swaroop, Esha and Jochym-O'Connor, Tomas and Yoder, Theodore J.},
	year         = 2026,
	journal      = {PRX Quantum},
	volume       = 7,
	number       = 1,
	url          = {https://journals.aps.org/prxquantum/abstract/10.1103/1g44-jp62}
}

@article{he2025extractors,
	title        = {{Extractors: {QLDPC} Architectures for Efficient Pauli-Based Computation}},
	author       = {Zhiyang He and Alexander Cowtan and Dominic J Williamson and Theodore J Yoder},
	year         = {2025},
	journal      = {arXiv preprint arXiv:2503.10390},
	pages        = {},
	url          = {https://arxiv.org/abs/2503.10390}
}

@misc{submit_cluster,
    title        ={Sub{MIT}: A Physics Analysis Facility at {MIT}},
    author       ={Josh Bendavid and Mariarosaria D'Alfonso and Jan Eysermans and Chad Freer and Maxim Goncharov and Matthew Heine and Luca Lavezzo and Marianne Moore and Christoph Paus and Xuejian Shen and David Walter and Zhangqier Wang},
    year         ={2025},
    eprint       ={2506.01958},
    archivePrefix={arXiv},
    primaryClass ={cs.DC},
    url          ={https://arxiv.org/abs/2506.01958},
}

@book{conway_sloane_1999,
  author         = {Conway, John H. and Sloane, Neil J. A.},
  title          = {Sphere Packings, Lattices and Groups},
  edition        = {3},
  publisher      = {Springer-Verlag},
  address        = {New York},
  year           = {1999},
  series         = {Grundlehren der mathematischen Wissenschaften},
  volume         = {290},
  isbn           = {978-0-387-98585-5},
}

@article{prange_1962,
  author         = {Prange, Eugene},
  title          = {The Use of Information Sets in Decoding Cyclic Codes},
  journal        = {IRE Transactions on Information Theory},
  volume         = {8},
  number         = {5},
  pages          = {5--9},
  year           = {1962},
  url            = {https://ieeexplore.ieee.org/document/1057777}
}

@INPROCEEDINGS{grassl_autos,
  author         = {Grassl, Markus and Roetteler, Martin},
  booktitle      = {2013 IEEE International Symposium on Information Theory}, 
  title          = {Leveraging automorphisms of quantum codes for fault-tolerant quantum computation}, 
  year           = {2013},
  volume         = {},
  number         = {},
  pages          = {534-538},
  url            = {https://ieeexplore.ieee.org/document/6620283}
}

@misc{berthusen2025automorphismgadgetshomologicalproduct,
  title          ={Automorphism gadgets in homological product codes}, 
  author         ={Noah Berthusen and Michael J. Gullans and Yifan Hong and Maryam Mudassar and Shi Jie Samuel Tan},
  year           ={2025},
  eprint         ={2508.04794},
  archivePrefix  ={arXiv},
  primaryClass   ={quant-ph},
  url            ={https://arxiv.org/abs/2508.04794}, 
}

@misc{berthusen2025simplelogicalquantumcomputation,
  title          = {Simple logical quantum computation with concatenated symplectic double codes}, 
  author         = {Noah Berthusen and Elijah Durso-Sabina},
  year           = {2025},
  eprint         = {2510.18753},
  archivePrefix  = {arXiv},
  primaryClass   = {quant-ph},
  url            = {https://arxiv.org/abs/2510.18753}, 
}
\appendix

%

\section{List of Symbols}\label{appen:symbols}

The tables below collect the principal symbols used throughout the paper,
organized according to the three-space dependency framework introduced in
Sec.~\ref{sec:3SpaceIso}: the polynomial ring space ($R_\ell$ and its
quotients), the matrix space ($\bbF_2$ binary circulant blocks), and the
physical qubit space ($\bbF_2^{2\ell}$ and its automorphism group).  Each
table is tinted with the same color used for that space in
Fig.~\ref{fig:3spaceiso}.  The ``Main usage'' column points to the
principal theorem, definition, lemma, or proposition where the symbol is
introduced or plays its central role.  Indices and other purely auxiliary
scaffolding ($i$, $j$, $z$, $t$, $d_i$, $\zeta_i$, $\rho_i$, $\sigma_i$,
etc.) are omitted.

\renewcommand{\arraystretch}{1.3}

\begin{table*}[htbp]
\centering
\small
\begin{tabular}{@{} p{2cm} p{8.4cm} l @{}}
\arrayrulecolor{green!50!black}\hline
\rowcolor{green!20}
\multicolumn{3}{c}{\textbf{Polynomial Ring Space}} \\
\hline
\rowcolor{green!8}
\textbf{Symbol} & \textbf{Description} & \textbf{Main or first usage} \\
\hline
$\ell$
  & block length of the underlying classical cyclic code; $n=2\ell$ physical qubits
  & Sec.~\ref{sec:threespaces} \\
$R_\ell$
  & base ring $\bbF_2[x]/\langle x^\ell-1\rangle$ in which all defining polynomials live
  & Sec.~\ref{sec:threespaces} \\
$R_\ell^2$
  & pair module $R_\ell\oplus R_\ell$ whose cyclic submodules realize GB rowspaces
  & Thm.~\ref{thm:cyclic_submodule} \\
$\langle g\rangle$
  & principal ideal of $R_\ell$ generated by $g$
  & Thm.~\ref{thm:cyclic_submodule} \\
$f_1,\,f_2$
  & defining polynomials of the Generalized Bicycle code
  & Sec.~\ref{sec:gbreview} \\
$f$
  & shared factor $\gcd(f_1,f_2,x^\ell-1)$; controls $k=2\deg f$
  & Thm.~\ref{thm:cyclic_submodule} \\
$\hat f$
  & annihilator generator $(x^\ell-1)/f$ of $\langle f\rangle$ in $R_\ell$
  & Thm.~\ref{thm:cyclic_submodule} \\
$p,\,q$
  & transfer polynomials, $f_1=pf$ and $f_2=qf$, with $\gcd(p,q,x^\ell-1)=1$
  & Thm.~\ref{thm:cyclic_submodule} \\
$\overleftarrow{\alpha}$
  & reciprocal/reversal $\alpha(x^{-1})$; encodes circulant transposition
  & Thm.~\ref{thm:cyclic_submodule} \\
$M_g(\alpha,\beta)$
  & cyclic submodule $\{(\alpha a,\beta a):a\in\langle g\rangle\}\subseteq R_\ell^2$
  & Thm.~\ref{thm:cyclic_submodule} \\
$S$
  & quotient ring $R_\ell/\langle\hat f\rangle$; ambient space for transfer polynomial identities
  & Thm.~\ref{thm:cyclic_submodule}, Thm.~\ref{thm:substitution_automorphisms} \\
$\overleftarrow{S}$
  & reversed quotient $R_\ell/\langle\hat f(x^{-1})\rangle$; ambient space for $H_Z$ side identities
  & Thm.~\ref{thm:fold_cx} \\
$L_i$
  & local CRT chain ring $\bbF_2[x]/\langle f_i^{2^s}\rangle$ at component $i$
  & Thm.~\ref{thm:explicit-basis} \\
$f_i$
  & $i$-th distinct irreducible factor of $x^m-1$ over $\bbF_2$ ($\ell=2^s m$)
  & Thm.~\ref{thm:explicit-basis} \\
$h_i$
  & CRT lifting indicator $(x^\ell-1)/f_i^{2^s}$; selects component $i$
  & Thm.~\ref{thm:explicit-basis} \\
$\psi, \psi_L, \psi_R$
  & generic coordinate permutation / ring automorphism of $R_\ell$
  & Prop.~\ref{prop:block_automorphisms} \\
$\psi_i$
  & cyclic shift $f(x)\mapsto x^i f(x)$ on $R_\ell$
  & Sec.~\ref{subsubsec:bsa_cyclic} \\
$\psi_{x^j}$
  & substitution multiplier $f(x)\mapsto f(x^j)$ with $\gcd(j,\ell)=1$
  & Thm.~\ref{thm:substitution_automorphisms} \\
$\Pres(f)$
  & multipliers $j\in(\Z/\ell\Z)^\times$ with $\langle f(x^j)\rangle=\langle f\rangle$
  & Thm.~\ref{thm:substitution_automorphisms} \\
${\rm Stab}$, ${\rm Swap}$, ${\rm Inv}$, ${\rm SwapInv}$
  & four sets partitioning useful multipliers by their action on $(p,q)$ in $S$
  & Thm.~\ref{thm:substitution_automorphisms} \\
$r$
  & transfer ratio $r=pq^{-1}\in S$; defining MCR data $r^2+r+1=0$
  & Def.~\ref{def:MCR} \\
\hline
\end{tabular}
\end{table*}

\vspace{1em}

\begin{table*}
\centering
\small
\begin{tabular}{@{} l p{8.4cm} l @{}}
\arrayrulecolor{blue!50!black}\hline
\rowcolor{blue!20}
\multicolumn{3}{c}{\textbf{Matrix Space}} \\
\hline
\rowcolor{blue!8}
\textbf{Symbol} & \textbf{Description} & \textbf{Main or first usage} \\
\hline
$\mathrm{circ}(g)$
  & $\ell\times\ell$ binary circulant matrix associated to $g\in R_\ell$
  & Thm.~\ref{thm:cyclic_submodule} \\
$A,\,B$
  & binary circulant blocks $\mathrm{circ}(f_1),\mathrm{circ}(f_2)$
  & Sec.~\ref{sec:gbreview} \\
$H_X,\,H_Z$
  & GB stabilizer parity-check matrices $H_X=[A\,|\,B]$, $H_Z=[B^T\,|\,A^T]$
  & Thm.~\ref{thm:cyclic_submodule} \\
$\operatorname{rs}(\cdot)$
  & rowspace of a parity-check matrix; equals a cyclic submodule of $R_\ell^2$
  & Thm.~\ref{thm:cyclic_submodule} \\
\hline
\end{tabular}
\end{table*}

\vspace{1em}

\begin{table*}
\centering
\small
\begin{tabular}{@{} l p{8.4cm} l @{}}
\arrayrulecolor{red!50!black}\hline
\rowcolor{red!20}
\multicolumn{3}{c}{\textbf{Physical Qubit Space}} \\
\hline
\rowcolor{red!8}
\textbf{Symbol} & \textbf{Description} & \textbf{Main or first usage} \\
\hline
$\bbF_2^{2\ell}$
  & $2\ell$-dimensional physical qubit space, partitioned into Left/Right blocks
  & Sec.~\ref{sec:3spaceGB} \\
$S_{2\ell}$
  & symmetric group of coordinate permutations on the $2\ell$ qubits
  & Cor.~\ref{cor:automorphism_criterion} \\
$\phi$
  & code automorphism, an element of $S_{2\ell}$ preserving the stabilizer
  & Cor.~\ref{cor:automorphism_criterion} \\
$\sigma$
  & full block swap, $(i,i+\ell)\mapsto(i+\ell,i)$
  & Prop.~\ref{prop:block_automorphisms} \\
$\phi_i$
  & block-wise cyclic shift $\psi_i\oplus\psi_i$
  & Sec.~\ref{subsubsec:bsa_cyclic} \\
$\phi_{x^j}$
  & block-wise substitution multiplier $\psi_{x^j}\oplus\psi_{x^j}$
  & Thm.~\ref{thm:substitution_automorphisms} \\
$\Sigma$
  & universal block-swap involution $\sigma\circ\phi_{x^{-1}}$; gives the $H$-type fold gate
  & Sec.~\ref{subsubsec:bsa_block} \\
$M_\sim,\,M_\leftrightarrow$
  & rowspace-preserving and rowspace-swapping classes of code automorphisms
  & Cor.~\ref{cor:automorphism_criterion} \\
$\tau$
  & $M_\leftrightarrow$ involution / ZX-duality used as a fold for transversal gates
  & Prop.~\ref{prop:S_type_swapping}, \ref{prop:S_type_preserving} \\
$S_\tau$
  & S-type fold-transversal logical gate built from $\tau$
  & Prop.~\ref{prop:S_type_swapping} \\
  Fix($\tau$) & Fixed point set of the involution $\tau$ & Thm.~\ref{thm:fold_s_gate}, Prop.~\ref{prop:S_type_preserving}\\
$\mathcal{L}_X,\,\mathcal{L}_Z$
  & spaces of $X$- and $Z$-type logical operators of the GB code
  & Sec.~\ref{subsubsec:crt_setup} \\
$\bar X_i,\,\bar Z_i$
  & canonical logical operator representatives in $\mathcal{L}_X,\mathcal{L}_Z$
  & Thm.~\ref{thm:explicit-basis} \\
$k$
  & number of encoded logical qubits, $k=2\deg(f)$
  & Sec.~\ref{sec:gbreview} \\
\hline
\end{tabular}
\end{table*}

\arrayrulecolor{black}

\section{Ring Theory and Classical Cyclic Coding Theory}\label{appen:codingtheory}

The theory contained in this section is provided to establish notation for the main paper and is in no way novel. This appendix seeks not to re-derive existing results, but is presented for readers who may not be intimately familiar with the algebraic notions used in the paper, or arguments underlying key concepts.  A few explicit proofs are provided where the understanding gained serves a greater purpose for this work. For each Lemma, if a rigorous proof is not provided, a citation is provided to a standard text that the reader may reference.

\subsection{Ring Theory Review}
This section reviews known facts about the study of $R_\ell = \bbF_2[x]/\langle x^\ell-1\rangle$. $R_\ell$ is a polynomial quotient ring where all elements are polynomials of degree less than $\ell$ with coefficients in $\bbF_2$. ``Working in $R_\ell$'' is functionally equivalent to working $\bmod{\ x^\ell-1}$ when performing polynomial arithmetic. 

Let $\ell = 2^sm$ where $m$ is odd and $s \geq 0$. The Frobenius automorphism for $\bbF_q$ says that over $\bbF_q[x], \ g(x^q) = g(x)^q$. As such, over $\bbF_2$, we have that $h(x) = x^\ell-1 = x^{2^sm} - 1 = (x^{2^s})^m - 1 = (x^m-1)^{2^s}$. As $\bbF_2[x]$ is a unique factorization domain, $x^\ell-1$ admits a unique factorization into a product of irreducible polynomials:
\bea
x^\ell -1 = (x^m-1)^{2^s} =  \prod_{i = 1}^a g_i^{2^s}
\eea
where each $g_i$ is unique. Thus, when $\ell$ is odd, no factor is repeated, while when $\ell$ is even, each factor has the same multiplicity, $2^s$. The Chinese Remainder Theorem then gives
\bea
    R_\ell \;\cong\; \prod_{i=1}^a \mathbb{F}_2[x]/\langle g_i^{2^s} \rangle 
\eea
When $s = 0$ and $\ell$ is odd, each individual $\bbF_2[x]/\langle g_i \rangle$ is itself a field with $d = 2^{\deg(g_i)}$ elements. When $s>0$, each  $\bbF_2[x]/\langle g_i^{2^s} \rangle$ forms both a chain ring and an Artinian local ring. Chain rings are discussed in depth in Sec.~\ref{subsubsec:crt_setup}.

$R_\ell$ is a principal ideal ring, which means that every ideal of $R_\ell$ is principal --- i.e., generated by a single element, and all its ideals are generated by divisors of $x^\ell-1$. We investigate why this must be true in Lemma~\ref{lem:equivalent_ideals}.

\begin{lem}\label{lem:equivalent_ideals}
    Let $g$ be an irreducible factor of $x^\ell - 1$. For some polynomial $r \in R_\ell$, 
    \begin{enumerate}[label=(\roman*)]
        \item $\langle r\rangle = \langle g \rangle $ if and only if $\gcd(r, x^\ell-1) = g$
        \item $\langle r \rangle = \langle g \rangle$ if and only if $g = ur$ for $u \in R_\ell^\times$, i.e., $g, r$ are associates
        \item All ideals of $R_\ell$ are generated by divisors of $x^\ell-1$
    \end{enumerate}
\end{lem}
\begin{proof}
    (i) is a standard fact from polynomial ring theory, and a proof can be found in Theorem 4.4.4, page 144 of \cite{huffmanpless}. 

    (ii) is slightly more involved. Adopting the notation in Sec.~\ref{subsubsec:crt_setup}, let $L_i = \mathbb{F}_2[x]/\langle g_i^{2^s} \rangle$ denote the chain ring at each CRT component. All $r \in R_\ell$ corresponds to a tuple $(r_1, r_2, \dots, r_a)$, with each $r_i \in L_i$, and $r \in R_\ell^\times$ if and only if each $r_i$ is a unit in $L_i$. 
    
    As each $g_i$ is irreducible, we have that each $L_i$ has maximal ideal $\langle g_i\rangle$. In a local ring $L_i$, any element that is not contained in the maximal ideal is a unit (\cite{dummit2003abstract}, section 7.4). Thus, $r_i$ is a unit in $L_i$ if $r_i \notin \langle g_i\rangle$. 

    If $\langle g \rangle = \langle r \rangle$ in $R_\ell$ we must then have $\langle g_i \rangle = \langle r_i \rangle \ \forall L_i$. If  $\langle g_i \rangle = \langle r_i \rangle = 0$ they are trivially separated by units, and so we consider the non-trivial case. If $\langle g_i \rangle = \langle r_i \rangle \neq 0$, we have that
    \bea
    g_i = \alpha_ir_i \quad r_i = \beta_ig_i \implies (1-\alpha_i\beta_i)g_i = 0
    \eea
    If $\alpha_i$ or $\beta_i \in \langle g_i\rangle$, then $\alpha_i\beta_i \in \langle g_i \rangle$, which implies that $(1-\alpha_i\beta_i)$ is a unit. However, $(1-\alpha_i\beta_i)g_i = 0$, and if $1-\alpha_i\beta_i$ is a unit, then $g_i = 0$, which is a contradiction and so $\alpha_i, \beta_i$ must themselves be units in $L_i$. Thus, in every local ring $L_i$, $g, r$ in the same ideals are associates. 
    
    By the CRT isomorphism of $R_\ell$, we can always construct a tuple $(u_1, \dots, u_a)$ that associates $g_i, r_i$ in each local ring, and as such, corresponds to a unit in $R_\ell$, and so $g, r$ are associates in $R_\ell.$ 

    (iii) is a consequence of (i) and (ii).
\end{proof}

Because each ideal is generated by a divisor of $x^\ell-1$, even when working with $\langle r \rangle$, it is typical to assess the properties of the ideal in terms of $\langle g \rangle$, the equivalent ideal generated by the divisor of $x^\ell-1$. 

We record some facts about ideals and their elements over the next two lemmas. 
\begin{lem}\label{lem:annihilator_invertible}
Let $\langle g \rangle$ be an ideal of $R_\ell$. Then $q \in R_\ell$ is invertible in $R_\ell / \langle g \rangle $ if and only if $\gcd(q, g) = 1$ in $\mathbb{F}_2[x]$.
\end{lem}
\begin{proof}

$(\Leftarrow)$\; If $\gcd(q, g) = 1$, then by B\'{e}zout's identity there exist $u, v \in \mathbb{F}_2[x]$ with
\bea
uq + vg = 1.
\eea
Reducing modulo $g$ gives $uq \equiv 1$, so $q$ is a unit.

\medskip\noindent
$(\Rightarrow)$\; Suppose $d = \gcd(q, g)$ has $\deg d \geq 1$.  Write $g = dh$ with $0 < \deg h < \deg g$.  Since $d \mid q$ we have $qh \equiv 0 \pmod{g}$, while $h \not\equiv 0 \pmod{g}$ (as $\deg h < \deg g$).  Thus $q$ is a zero divisor in $\mathbb{F}_2[x]/\langle g \rangle$ and cannot be a unit.
\end{proof}

\begin{lem}\label{lem:subsetideals}
    Let $g$ be an irreducible divisor of $x^\ell - 1$, and $a$ be any element in $R_\ell$. Then, $\langle ag\rangle = \{r\cdot ag | r \in R_\ell \} = \{ as| s \in \langle g \rangle\}$. 
\end{lem}
\begin{proof}
    $\{as | s \in \langle g \rangle\} = \{as | s = r\cdot g, r \in R_\ell\} = \{arg | r \in R_\ell\} = \{r\cdot (ag) | r \in R_\ell \} = \langle ag \rangle$
\end{proof}

A polynomial of particular interest, given an ideal, is the corresponding polynomial that generates the \textit{annihilator}, denoted $\ann(g) = \langle \hat{g} \rangle$, and defined as follows:
\bea
\hat{g} = \frac{x^\ell - 1}{g}
\eea
The annihilator is the set of things that sends any element of $\langle g \rangle$ to 0 in $R_\ell$ under multiplication. If $p \in \langle g \rangle$, then $p = p'g$ for some $p' \in R_\ell$. Similarly $q = q'\hat{g}$, for $q \in \langle \hat{g}\rangle$. Observe then that
\[
pq = p'gq'\hat{g} = p'q'g\hat{g} = p'q'g\frac{x^\ell-1}{g} = p'q'(x^\ell-1) \equiv 0 \pmod{x^\ell-1}
\]
Thus, elements from $\langle \hat{g}\rangle$ annihilate elements of $\langle g \rangle$ under multiplication. We might be interested in a few properties of the annihilator, which we explore in the next three lemmas:

\begin{lem}\label{lem:annihilatordimension}
    $\dim(\langle\hat{g}\rangle) = \deg(g) = d$, and $\hat{g}, x\hat{g}, \dots, x^{d-1}\hat{g}$ are a linearly independent set of $\bbF_2^\ell$ vectors.  
\end{lem}
\begin{proof}
    (Theorem 4.2.1, page 125 of \cite{huffmanpless})
    $\hat{g}$ defines a cyclic code with generator polynomial $\hat{g}$. We know that
    \bea
    \hat{g} = \frac{x^\ell-1}{g}
    \eea
    and thus has degree $\ell-d$ for $d = \deg(g)$. Thus, by the dimension formula for cyclic codes, we have
    \bea
    \dim(\langle \hat{g} \rangle) = \ell - \deg(\hat{g}) = \ell - (\ell - d) = d = \deg(g)
    \eea
    It remains to show that  $\hat{g}, x\hat{g}, \dots, x^{d-1}\hat{g}$ are linearly independent over  $\mathbb{F}_2$. For $0 \leq i \leq d-1$, the element $x^i \hat{g}$ has degree  $(\ell - d) + i \leq \ell - 1$, so no reduction modulo $x^\ell - 1$ occurs. These are therefore genuine polynomials of strictly increasing degrees  $\ell - d,\; \ell - d + 1,\; \dots,\; \ell - 1$, and polynomials of  distinct degrees are linearly independent over any field.
\end{proof}

\begin{lem}\label{lem:annihilator_sufficiency}
Let $q, q' \in R_\ell$.  If $q \equiv q' \pmod{\hat{g}}$, i.e., $q \equiv q'$ in $S = R_\ell/\langle \hat{g} \rangle$, then $qa = q'a$ for every $a \in \langle g \rangle$.
\end{lem}
\begin{proof}
The condition $q \equiv q' \pmod{\hat{g}}$ means $q - q' \in \langle \hat{g} \rangle$, i.e.\ $(q - q')g \equiv 0$ in $R_\ell$.  Since every $a \in \langle g \rangle$ is a multiple of $g$, say $a = rg$, we have
\bea
(q - q')a = (q - q')rg = r \cdot (q - q')g = 0,
\eea
so $qa = q'a$.
\end{proof}

\begin{lem}\label{lem:f_unit_in_S}
    Given $\ell$ odd with $g, \hat{g}, R_\ell$ defined as throughout this paper, $g \in S^\times = (R_\ell/\langle \hat{g} \rangle)^\times$
\end{lem}
\begin{proof}
    As $g, \hat{g}$ share no divisors, we have that, by B\'{e}zout's identity, $ag + b\hat{g} = 1$ in $R_\ell$. When we quotient into $S$, this becomes
    \bea
    ag + b\hat{g} = ag = 1 \qquad \mathrm{in} \ S
    \eea
    That is, there exists some $a$ such that $ag = 1$ and so $g$ is a unit.
\end{proof}

Given the content of this work, a particularly important class of $R_\ell$ properties are the ring automorphisms. Automorphisms as found in abstract algebra carry properties that are often not discussed in the context of quantum code automorphisms. An automorphism on $R_\ell$ is an isomorphism that sends $R_\ell$ back to $R_\ell$, and crucially, retains all algebraic structure on $R_\ell$. Thus, if $\psi$ is a ring automorphism, any statement that is true of $R_\ell$ or its ideals must also be true of $\psi(R_\ell)$.

In particular, we are interested in the multiplier automorphisms of $R_\ell$, characterized in the two lemmas below:
\begin{lem}\label{lem:ringautomorphism}
    For all $j$ with $\gcd(j, \ell) = 1$, the map $\phi:x \mapsto x^{j}$ is a ring automorphism of $R_\ell$
\end{lem}
\begin{proof}
    Page 138 of \cite{huffmanpless}, discussion on multipliers
\end{proof}
\begin{cor}
    By Lemma \ref{lem:ringautomorphism}, $\phi: x \mapsto x^{-1}$ is always a ring automorphism. As such, if $x^j$ is a ring automorphism of $R_\ell$, then $x^{-j}$ is always a ring automorphism.
\end{cor}
\begin{proof}
    Automorphisms form a closed group under composition.
\end{proof}

Combining the previous discussions, we have the following:
\begin{cor}\label{cor:ringautomorphism_idealequiv}
    Given a ring automorphism $\psi_{x^j}: x \rightarrow x^j$ of $R_\ell$ and an ideal $\langle g \rangle$ where $g$ is an irreducible factor in $R_\ell$, $\psi_{x^j}$ preserves the ideal $\langle g \rangle$ if and only if $\langle g(x^j)\rangle = \langle g \rangle$. By Lemma~\ref{lem:equivalent_ideals}, this holds if and only if $\gcd(g(x^j), x^\ell - 1) = g$.
\end{cor}
\begin{proof}
    Follows from Lemmas~\ref{lem:ringautomorphism} and~\ref{lem:equivalent_ideals}.
\end{proof}

Finally, we end the study of $R_\ell$ with a statement on cyclic submodules of $R_\ell \oplus R_\ell$.
\begin{lem}\label{lem:equivalent_modules}
Let $\ell$ be odd, $g \mid x^\ell - 1$, $\hat{g} = (x^\ell - 1)/g$, and $S = R_\ell/\langle \hat{g} \rangle$. For $(p, q) \in R_\ell^2$ with $\gcd(p, q, x^\ell - 1) = 1$, write
\bea
    M_g(p, q) \;=\; \{(pa, qa) : a \in \langle g \rangle\}.
\eea
Given two such pairs $(p_1, q_1)$ and $(p_2, q_2)$, both satisfying the gcd condition, the following are equivalent:
\begin{enumerate}[label=\textup{(\alph*)}]
    \item $M_g(p_1, q_1) = M_g(p_2, q_2)$
    \item $(p_1, q_1) = u \cdot (p_2, q_2)$ in $S^2$ for some unit $u \in S^\times$.
\end{enumerate}
\end{lem}

\begin{proof}
Under $\gcd(p, q, x^\ell - 1) = 1$, the pair $(p, q) \in S^2$ has trivial annihilator: the only $s \in S$ with $s(p, q) = (0, 0)$ is $s = 0$. Indeed, $\hat{g} \mid x^\ell - 1$ gives $\gcd(p, q, \hat{g}) = 1$, so B\'ezout produces $a, b \in \mathbb{F}_2[x]$ with $ap + bq = 1$ in $S$; if $sp = sq = 0$ then $s = s(ap + bq) = 0$.

\smallskip
\noindent $(b) \Rightarrow (a)$. If $(p_1, q_1) = u(p_2, q_2)$ in $S^2$ with $u \in S^\times$, then every scalar multiple of $(p_1, q_1)$ is a scalar multiple of $(p_2, q_2)$, and conversely via $u^{-1}$. So the two submodules coincide.

\smallskip
\noindent $(a) \Rightarrow (b)$. Suppose the submodules coincide. Then $(p_1, q_1) \in M_g(p_2, q_2)$, so $(p_1, q_1) = u (p_2, q_2)$ for some $u \in S$. Symmetrically, $(p_2, q_2) = v(p_1, q_1)$ for some $v \in S$. Substituting,
\bea
    (p_2, q_2) \;=\; vu \cdot (p_2, q_2),
\eea
so $(vu - 1)(p_2, q_2) = 0$. By the key fact applied to $(p_2, q_2)$, $vu - 1 = 0$, hence $u$ is a unit with inverse $v$.
\end{proof}

\subsection{Cyclic codes as ideals of $\bbF_q[x]/\langle x^\ell-1\rangle$}\label{subsec:cycliccodes}

Given a polynomial $\langle g \rangle$ defining a cyclic code $C$, the codespace of $C$ is all $\bbF_2^\ell$ vectors obtained from polynomials of $\langle g \rangle$ under the polynomial to coefficient mapping
\bea
g = c_0 + c_1x + c_2x^2 + \dots + c_{\ell-1}x^{\ell-1} \mapsto \vec{g} = (c_0, c_1, \dots, c_{\ell-1})
\eea
Recall that the coefficients are drawn from $\bbF_2$. Cyclic codes can also be identified with a corresponding binary circulant matrix, obtained by building a $\bbF_2^{\ell \times \ell}$ matrix with the first row equal to $(c_0, c_1, \dots, c_{\ell-1})$, and every subsequent row a cyclic shift:
\bea
\vec{g} = (c_0, c_1, \dots, c_{\ell-1}) \mapsto \mathrm{circ}(g) = \begin{bmatrix}
    c_0 & c_1  & c_2 & \dots & c_{\ell-2} & c_{\ell-1} \\
    c_{\ell-1} & c_0  & c_1 & \dots & c_{\ell-3} & c_{\ell-2} \\
    c_{\ell-2} & c_{\ell-1}  & c_0 & \dots & c_{\ell-4} & c_{\ell-3} \\
    \vdots & & & \dots & & \vdots \\
    c_2 & c_3  & c_4 & \dots & c_0 & c_1 \\
    c_1 & c_2  & c_3 & \dots & c_{\ell-1} & c_0 \\
\end{bmatrix}
\eea
Observe that the vertical columns of the binary circulant matrix are also circulant with respect to the vector and polynomial
\[
(c_0, c_{\ell-1}, c_{\ell-2}, \dots, c_2, c_1) \mapsto g(x^{-1})
\]
and $\mathrm{circ}(g)^T$ is equivalent to defining the binary circulant matrix defined by $g(x^{-1})$. As $g(x^{-1})$ is obtained from $g$ via the ring automorphism $\psi_{x^{-1}}$, all properties of $C = \langle g \rangle$ similarly hold for $C' = \langle g(x^{-1}) \rangle$.

The permutations of cyclic codes that map cyclic codes back to themselves or to other cyclic codes are well studied in the literature, and we have the following definitions:

\begin{definition} Let $C = \langle g \rangle$ be a cyclic code of $R_\ell$. A permutation $\sigma \in S_\ell$ is a \emph{permutation automorphism} of $C$ if and only if $\sigma(C) = C$. The set of all such permutations forms a subgroup of $S_\ell$, called the \emph{permutation automorphism group} of $C$, and denoted
\[
    \mathrm{PAut}(C) \;=\; \{\, \sigma \in S_\ell \,:\, \sigma(C) = C \,\}.
\]
 
Two distinguished families of permutations either always belong to $\mathrm{PAut}(C)$ or admit a clean characterization for membership in it.
 
\begin{enumerate}
    \item[(i)] \emph{Cyclic shift.} The map $\psi_i : g \mapsto x^ig\pmod{x^\ell-1}$ satisfies $\psi_i \in \mathrm{PAut}(C)$, since $C$ is cyclic by hypothesis.
    \item[(ii)] \emph{Multipliers.} For $j \in (\bbZ/\ell\bbZ)^\times$, define the multiplier permutation $\psi_{x^j} : g \mapsto g(x^j) \pmod{x^\ell-1}$. Then
    \[
        \psi_{x^j} \in \mathrm{PAut}(C) \;\Longleftrightarrow\; \langle g(x^j) \rangle = \langle g \rangle \;\Longleftrightarrow\; g(x^j) = u \cdot g \text{ in } R_\ell
    \]
\end{enumerate}
We can represent the multiplier group of $C$ as
\[
    M(C) \;=\; \{\, j \in (\bbZ/\ell\bbZ)^\ast \,:\, g(x^j) = ug, u \in R_\ell^\times \,\}.
\]
The affine maps $x \mapsto j x + i$ with $j \in M(C)$ and $i \in \bbZ/\ell\bbZ$ all lie in $\mathrm{PAut}(C)$, yielding the subgroup
\[
    \bbZ/\ell\bbZ \rtimes M(C) \;\le\; \mathrm{PAut}(C).
\]
\end{definition}

For a more thorough treatment see \cite{huffmanpless}. For most cyclic codes, one has equality $\mathrm{PAut}(C) = \bbZ/\ell\bbZ \rtimes M(C)$, but a small family of \emph{degenerate} codes admits permutation automorphisms outside this affine subgroup. We review the definition of a degenerate cyclic code:
\begin{definition}
    A code $C = \langle g \rangle$ is degenerate if every codeword has period less than $\ell$. The following statements are equivalent to $C$ being degenerate:
    \begin{enumerate}
        \item There exists a proper divisor $d \mid \ell$, $d < \ell$, such that $\psi_d(c) = c \ \forall c \in C$, i.e., cyclically shifting by $d$ positions returns every codeword back to itself for $d < \ell$.
        \item $\hat{g} \mid x^d-1 $ for $d \mid \ell, \ d < \ell$
    \end{enumerate}
\end{definition}

Non-trivial exceptions to $\mathrm{PAut}(C) = \bbZ/\ell\bbZ \rtimes M(C)$ include the punctured Reed--Muller codes and the binary Golay code, with $\mathrm{PAut}(C) = M_{23}$. It is a longstanding conjecture by Berger and Charpin~\cite{BergerCharpin}, recently confirmed for irreducible cyclic codes~\cite{feng2026permutationautomorphismgroupsirreducible}, that for almost all cyclic codes the permutation automorphism group coincides with the affine group generated by cyclic shifts and multiplier ring automorphisms of $R_\ell$. We note that the general case --- in particular codes defined by products of irreducible factors, which is the setting in this work --- remains open.

\begin{definition} Let $C = \langle g \rangle$ and $C' = \langle h \rangle$ be cyclic codes in $R_\ell = \bbF_2[x]/(x^\ell - 1)$. We say that $C$ and $C'$ are \emph{permutation equivalent}, written $C \sim C'$, if there exists $\sigma \in S_\ell$ with $\sigma(C) = C'$. The set of permutations realizing this equivalence is denoted
\[
    \mathrm{PEq}(C, C') \;=\; \{\, \sigma \in S_\ell \,:\, \sigma(C) = C' \,\}.
\]
This set is either empty (when $C \not\sim C'$) or a left coset of $\mathrm{PAut}(C)$ in $S_\ell$: fixing any $\sigma_0 \in \mathrm{PEq}(C, C')$, one has $\mathrm{PEq}(C, C') = \sigma_0 \cdot \mathrm{PAut}(C)$. In particular, $\mathrm{PEq}(C, C) = \mathrm{PAut}(C)$.

Two distinguished families of permutations carry cyclic codes to cyclic codes, paralleling the previous definition.

\begin{enumerate}
    \item[(i)] \emph{Cyclic shifts.} For $i \in \bbZ/\ell\bbZ$, the map $\psi_i : g \mapsto x^i g \pmod{x^\ell - 1}$ satisfies $\psi_i(C) = C$, so $\psi_i \in \mathrm{PEq}(C, C')$ if and only if $C = C'$.
    \item[(ii)] \emph{Multipliers.} For $j \in (\bbZ/\ell\bbZ)^\ast$, the multiplier is a ring automorphism of $R_\ell$, hence sends the ideal $\langle g \rangle$ to the ideal $\langle g(x^j) \rangle$:
    \[
        \psi_{x^j}(C) \;=\; \langle g(x^j) \rangle,
    \]
    which is itself a cyclic code. Consequently,
    \[
        \psi_{x^j} \in \mathrm{PEq}(C, C') \;\Longleftrightarrow\; \langle g(x^j) \rangle = \langle h \rangle \;\Longleftrightarrow\; g(x^j) = u \cdot h \text{ in } R_\ell, \text{ for some } u \in R_\ell^\times.
    \]
\end{enumerate}

When such a $j$ exists, we say that $C$ and $C'$ are \emph{multiplier equivalent}, and we write $C \sim_M C'$. The set of multipliers realizing the equivalence is
\[
    M(C, C') \;=\; \{\, j \in (\bbZ/\ell\bbZ)^\ast \,:\, g(x^j) = u h \text{ for some } u \in R_\ell^\times \,\},
\]
which is either empty or a left coset of $M(C)$ in $(\bbZ/\ell\bbZ)^\ast$. More generally, for any $j \in M(C, C')$ and $i \in \bbZ/\ell\bbZ$, the affine permutation $k \mapsto j k + i$ of $\bbZ/\ell\bbZ$ lies in $\mathrm{PEq}(C, C')$, contributing the coset
\[
    \sigma_j \cdot \bigl(\bbZ/\ell\bbZ \rtimes M(C)\bigr) \;\subseteq\; \mathrm{PEq}(C, C'),
\]
where $\sigma_j$ is any affine permutation with multiplier part $j$.
\end{definition}

The classical literature contains the following results regarding permutation equivalence:
\begin{itemize}
    \item For $\ell$ with $\gcd(\ell, \varphi(\ell)) = 1$, where $\varphi(\ell)$ denotes the Euler totient function, all permutation equivalences between cyclic codes of length $\ell$ are realized by multipliers~\cite{Palfy1987}, and the affine restriction is provably exhaustive (P\'alfy's regime).
    \item For $\ell = p^r$ a prime power with $r \geq 2$, permutation equivalences between cyclic codes are characterized completely as compositions $M\mu$ of an ordinary multiplier $\mu$ and a \emph{generalized multiplier} $M$~\cite{HuffmanJobPless1993, Guenda2010}.  Generalized multipliers are a finite extension of the affine group, and the resulting permutation equivalences are well-understood.
    \item For composite $\ell$ outside of P\'alfy's regime, sporadic non-affine, non-generalized-multiplier permutation equivalences may exist.  Dastbasteh-Lison\v{e}k~\cite{DastbastehLisonek2022} establish such equivalences, however, these constructions explicitly require $\bbF_q$ with $q \neq 2$ and do not apply in the $q=2$ regime we are working in. However, framing GB codes as cyclic submodules is extendable beyond $q = 2$, and in such cases the results of Ref.~\cite{DastbastehLisonek2022} would potentially apply.
\end{itemize}

\end{document}